\documentclass[12pt]{article}
\usepackage{amsmath}
\usepackage{graphicx}
\usepackage{enumerate}
\usepackage{natbib}
\usepackage{url} 

\usepackage{bbm}
\usepackage{bm}

\usepackage{amsthm}
\usepackage{amssymb}

\usepackage{xr}
 
\usepackage{scalerel}
 
\usepackage{xcolor}
\usepackage{multirow}
\usepackage{arydshln}
 
\usepackage{soul}

\usepackage{float}

\usepackage{latexsym}
\usepackage{caption}
 
\usepackage{mathtools}

\theoremstyle{definition}

\newtheorem*{theorem*}{Theorem}
\newtheorem{theorem}{Theorem}
\newtheorem*{rmk*}{remark}

\newtheorem{proposition}{Proposition}
\newtheorem{lemma}{Lemma}
\newtheorem{example}{Example}

\newtheorem{condition}{Condition}

\newtheorem{definition}{Definition}
\newtheorem{remark}{Remark}

\newtheorem*{corollary*}{Corollary}

\usepackage{etoolbox} 
\apptocmd{\sloppy}{\hbadness 10000\relax}{}{} 

\usepackage{multibib}
\newcites{sec}{References}

\usepackage{color}
\usepackage{listings}
\usepackage{hyperref}
\hypersetup{
  colorlinks=true,
  linkcolor=blue,
  citecolor=blue,
  urlcolor=blue
}

\usepackage{booktabs}

\usepackage{lscape}

\RequirePackage[normalem]{ulem}

\usepackage{setspace}

\allowdisplaybreaks

 \def\spacingset#1{\renewcommand{\baselinestretch}%
{#1}\small\normalsize}

\def\ottb{o(T^{2-\beta})}
\def\sigozzp{\sig_{1,zz'}}
\def\sigtzzp{\sig_{2,zz'}}
\def\vtzzp{V_{2,zz'}}
\def\voz{V_{1,z}}

\def\hsig{\hat\sig}

\def\subm{\text{submultiplicity}}
\def\chisqa{\chi^2_{\rc,1-\alpha}}

\def\snrc{\mn(0,I)}
\def\nnb{\nonumber}
\def\beginar{\begin{array}}
\def\endar{\end{array}}

\def\rc{{\rank{C}}}

\def\irc{I_{\rc}}

\def\tsq{T^2}
\def\tsr{\sqrt T}
\def\tsrif{\dfrac{1}{\tsr}}
\def\tsri{T^{-1/2}}
\def\ti{T^{-1}}

\def\rhot{\rho_t}

\def\pt{\pi_t}
\def\ps{\pi_s}
\def\nt{n_t}

\def\hvab{\hat V_{\aipw,b}}
\def\tvab{\tilde V_{\aipw,b}}
\def\myti{\my_{ti}}

\def\ct{C^\T}
\newcommand{\indep}{\perp \!\!\! \perp}

\def\iit{i=\ot{n_t}}
\def\hytiaz{\hy_{ti,\aipw}(z)}

\def\mbr{\mathbb R}

\def\bz{{\bm z}}
\def\ttinf{T\to\infty}

\def\hys{\hy_*}
\def\hts{\htau_*}

\def\hygz{\hat Y_{[g]}(z)}

\def\mtgz{\mathcal T_{g,z}}
\def\mtg{\mathcal T_g}
\def\indz{\I(Z_t = z)}
\def\all{{\textup{all}}}
\def\cf{{\textup{cf}}}
\def\spm{{\textup{sm}}}
\def\mtds{\ipw, \spm, \aipw, \all, \cf}

\def\mn{\mathcal N}
\def\si{\mathcal S_{\ipw,C, \alpha}}
\def\sa{\mathcal S_{\aipw,C, \alpha}}
\def\tsa{\tilde{\mathcal S}_{\aipw,C, \alpha}}

\def\htfull{H_t =\{(Z_s, Y_s, X_s): s =\ot{t-1}\}\cup\{X_t\}}

\def\hmtiz{\hat m_{ti}(z)}

\def\ytiz{Y_{ti}(z)}

\def\I{\mathbbm{1}}

\newcommand{\cv}{\check{V}}

\newcommand{\vba}{V_{\aipw}}
\newcommand{\hvba}{\hat V_{\aipw}}

\newcommand{\bmyt}{\bar \my_t}

\def\omt{\omega_t}
\def\phit{\phi_t}

\def\lmts{\lmtp^2}
\def\lmtp{(\lt+\mt)}

\def\cvt{\check V_T}
\def\vt{V_T}

\def\bmd{\bar{\md}}
\def\md{{\bm D}}
\newcommand{\mdt}{\md_t}

\def\mart{a martingale difference sequence with respect to filtration $\{\mf_t = H_{t+1}:  1 \le t\le T\}$}

\def\tmyta{\tilde{\my}_{t,\aipw}}
\def\tmya{\tilde{\my}_{\aipw}}
\def\tva{\tilde V_{\aipw}}
\def\dtzp{D_t(z')}

\def\zti{Z_{ti}}
\def\etiz{e_{ti}(z)}
\def\ebtz{e_t(\bz)}
\def\zbt{Z_{t}}
\def\ybt{Y_{t}}
\def\xbt{X_{t}}

\newcommand{\my}{{\bm Y}}
\def\hmy{\hat{\my}}
\def\bmy{\bar{\my}}

\def\ls{\lesssim}
\def\gt{\gamma_t}

\def\maxs{\max_{1\leq s\leq t}}

\def\detzi{\Delta_t(z)}

\def\ekz{e_k(z)}
\def\bkzs{A_k(z)^2}

\def\lm{\lambda_{\min}}
\def\lma{\lm(\va)}
\def\lmap{\{\lma\}}
\def\lmai{\{\lma\}^{-1}}

\def\lmi{\lm(\vi)}
\def\lmip{\{\lmi\}}
\def\lmii{\{\lmi\}^{-1}}
	\def\spacingset#1{\renewcommand{\baselinestretch}%
		{#1}\small\normalsize} \spacingset{1}

\def\vi{V_{\ipw}}

\newcommand{\hmyti}{\hmy_{t, \ipw}}
\newcommand{\byt}{\by_t}
\newcommand{\bytz}{\by_t(z)}

\def\hmtzps{\hmt(z')^2}

\def\bkz{A_{k}(z)}
\def\bkzp{A_{k}(z')}

\def\lt{L_t}

\def\maxt{\max_{\ott}}

\def\maxz{\max_{\ziz}}
\newcommand{\hmt}{\hat m_t}
\newcommand{\hbmt}{\hat{\bm m}_t}
\newcommand{\bbm}{\bar{\bm m}}

\newcommand{\tc}{\tau_C}

\newcommand{\hmyta}{\hmy_{t,\aipw}}
\newcommand{\va}{V_{\aipw}}
\newcommand{\cvi}{\check{V}_{\ipw}}
\newcommand{\cva}{\check{V}_{\aipw}}
\newcommand{\mat}{\mb_t}
\newcommand{\matt}{\mb_t^\T}
\newcommand{\mb}{{\bm A}}
 
\newcommand{\atzs}{\atz^2}
\def\atzp{A_t(z')}

\def\cct{CC^\T}
\def\lmax{\lambda_{\max}}
\def\lmin{\lambda_{\min}}
\def\at{A_t}
\def\atz{A_t(z)}
\def\htau{\hat\tau}
\def\hti{\hat\tau_{C, \ipw}}
\def\uet{\underline{e}_t}

\def\ueti{\underline{e}_t^{-1}}
\def\ue{\underline{e}}
\def\hmyi{\hat{\my}_{\ipw}}

\def\ott{t = \ot{T}}
 
\def\vaipw{V_{\aipw}}

\def\ziz{z\in\mz}

\def\myt{\my_t}
 
\def\dtz{D_t(z)}

\def\beginp{\begin{pmatrix}}
\def\endp{\end{pmatrix}}

\def\etzinv{\dfrac{1}{\etz}}
\def\etzi{ e_t(z)^{-1}}

\def\etzp{e_t\zp}

\def\zt{Z_t}

\def\ytzp{\yt(z')}

\def\zp{{(z')}}

\def\dt{\delta_t}

\def\ytzs{Y_{t}(z)^2}

\def\oo{o(1)}

\def\prear{Assume the adaptive randomization in Definition \ref{def:ad}}
\def\prebar{Assume the block adaptive randomization in Definition \ref{def:block_ad}}

\def\sig{\Sigma}
\def\sums{\sum_{s=1}^t}

\def\hv{\hat V}
\def\hvi{\hat V_{\ipw}}
\def\hva{\hat V_{\aipw}}

\def\yt{Y_t}

\def\rtz{  r_t(z)  }

\def\wtzp{w_t\zp}

\def\mf{\mathcal{F}}

\def\wtz{w_t(z)}

\def\byz{\by(z)}
\def\hmtz{\hat m_t(z)}
\def\hmtzp{\hat m_t(z')}

\def\etz{e_t(z)}
\def\sumz{\sum_{z\in\mz}}

\def\hmyai{\hmy_{\aipw}}
\def\hmya{\hmy_{\aipw}}

\def\hyaz{\hy_{\aipw}(z)}
\def\hytaz{\hy_{t,\aipw}(z)}
\def\hyta{\hy_{t,\aipw}}
\def\hyti{\hy_{t,\ipw}}

\def\hytiz{\hyti(z)}

\def\hyi{\hy_{\ipw}}
\def\hya{\hy_{\aipw}}
\def\hti{\htau_{\ipw, C}}
\def\hta{\htau_{\aipw,C}}

\def\hyiz{\hyi(z)}

\def\ytz{Y_t(z)}
\def\yt{Y_t}
\def\yto{Y_t(1)}
\def\by{\bar Y}
\def\sumi{\sum_{i=1}^{n_t}}
\def\sumt{\sum_{t=1}^T}
\def\sums{\sum_{s=1}^T}

\def\meant{T^{-1}\sumt}
\newcommand{\meantf}{\dfrac{1}{T}\sumt}
\def\meani{n_t^{-1}\sumi}
\def\meanti{N^{-1}\sumt\sumi}

\def\ipw{\textup{ipw}}
\def\aipw{\textup{aipw}}

\def\ipwc{IPW}
\def\aipwc{AIPW}

\def\mz{\mathcal Z}

\def\hht{H_t}

\def\by{\bar Y}

\def\hy{\hat Y}

\def\begini{\begin{itemize}}
\def\endi{\end{itemize}}
\def\begine{\begin{enumerate}}
\def\ende{\end{enumerate}}

\def\sumi{\sum_{i=1}^{n_t}}

\def\sm{{Supplementary Material}}
\def\sms{{Supplementary Material }}

\newcommand{\mt}{M_t}  

\def\with{\quad\text{with} \ \ }
\def\where{\quad\text{where} \ \ }

\def\cl{\centerline}
\newcommand{\cll}[1]{\bigskip\\\cl{$\ds #1$}

\smallskip}

\newcommand{\cls}[1]{\smallskip\\\cl{$\ds #1$}}

\newcommand{\scov}[1]{\cov_\textup{s}(#1)}
\newcommand{\op}{o_{\pr}(1)}
\newcommand{\oop}{O_{\pr}(1)}
\newcommand{\pr}{{\mathbb P}}
\newcommand{\oeq}[1]{\overset{#1}{=}}
\newcommand{\oleq}[1]{\overset{#1}{\leq}}
\newcommand{\oeqt}[1]{\overset{\text{#1}}{=}}

 \newcommand{\ot}[1]{1, \ldots,#1}
 
 \newcommand{\fn}[1]{\|#1\|_{\textsc f}}  
 \newcommand{\infn}[1]{\left\|#1\right\|_{ \infty}} 
 \newcommand{\infns}[1]{\left\|#1\right\|_{ \infty}^2} 
 \newcommand{\tn}[1]{\|#1\|_2}
 
 \newcommand{\rank}[1]{{\rm rank}(#1)}

 \newcommand{\diag}{\text{diag}}

\def\T{{\top}}

\newcommand{\cov}{\text{cov} } 
\newcommand{\var}{\text{var} }

\newcommand{\E}{\mathbb{E}}

\def\begina{\begin{eqnarray*}}
\def\enda{\end{eqnarray*}}
\def\beginy{\begin{eqnarray}}
\def\endy{\end{eqnarray}}
\def\begine{\begin{enumerate}}
\def\ende{\end{enumerate}}

\usepackage[utf8]{inputenc}
\usepackage[english]{babel}

\def\converged{\stackrel{d}{\longrightarrow}}

\usepackage{setspace}
\def\fp{finite-population}
\def\lemd{\lem~\ref{lem:D}}
\def\lemchebf{\lem~\ref{lem:cheb} (\cheb)}
\def\cheb{Chebyshev's inequality}
\def\csf{\cs\ inequality}
\def\cs{Cauchy--Schwarz}

\def\avesqe{Mean squared error }
\def\coverage{Coverage  rate}
\def\aveci{Average  CI length}

\def\mab{multi-armed bandit}
\def\mabs{{\mab}s}

\def\wlg{without loss of generality that $c_0 T_0^\beta < T_0 -1$ such that for $t > T_0$}

\def\ds{\displaystyle}

\def\ipwd{inverse-propensity-weighted}
\def\ipwf{inverse propensity weighting}
\def\iipwf{Inverse propensity weighting}

\def\allc{any fixed matrix $C$ with $K$ columns}
\def\allcf{\allc\ and full row rank}

\def\limdep{depend only on information from the $\min\{t-1, c_0 t^\beta\}$ most recent units}
\def\limdept{after some fixed time point $T_0$}

\def\prelimdeps{{there exist constants $\beta\in [0,1)$ and $c_0 > 0$ such that }}
\def\prelimdepsc{{There exist constants $\beta\in [0,1)$ and $c_0 > 0$ such that }}
\def\prelimdep{\prelimdeps \limdept}
\def\prelimdepc{\prelimdepsc \limdept}

\def\lrd{limited-range dependence}
\def\llrd{Limited-range dependence}

\def\ar{adaptive randomization}
\def\are{adaptively reweighted estimator}
\def\ares{{\are}s}
\def\arw{adaptive reweighting}
\def\ars{adaptive reweighting strategy}

\def\oreg{outcome regression}
\def\na{nonadaptive}
\def\naa{\na\ adjustment}

\def\aae{adaptive AIPW estimator}
\def\aca{adaptive covariate adjustment}

\def\db{design-based}

\def\br{block adaptive randomization}
\def\ur{unit-level adaptive randomization}
\def\car{covariate-adaptive randomization}
\def\sr{sequential rerandomization}

\def\cfbc{cross-fitting with \bc}
\def\bc{Bonferroni correction}
\def\cff{cross-fitting}

\def\defo{Definition~\ref{def:ad}}
\def\cond{Condition}
\def\thm{Theorem}
\def\thms{{\thm}s}
\def\prop{Proposition}
\def\props{{\prop}s}
\def\lem{Lemma}

\def\sec{Section}
\def\secs{{\sec}s}
\begin{document}

\spacingset{1}
\date{} 
  \title{\bf \Large Design-based theory for causal inference from adaptive experiments}
  \date{}
  \author{Xinran Li\thanks{Authors contributed equally; names are listed alphabetically.
    X. L. acknowledges support from the U.S. National Science Foundation (Grants \# 2400961). 
    We thank Alex Belloni, Iavor Bojinov, Hongbin (Elijah) Huang, Jinze Cui, Peng Sun, and Lei Zhang for inspiring discussions.}\hspace{.2cm}\\
    Department of Statistics, University of Chicago\\
    and \\
    Anqi Zhao \\
    Fuqua School of Business, Duke University}
  \maketitle

\def\contri{Our theory accommodates nonexchangeable units, both nonconverging and vanishing treatment probabilities, and nonconverging outcome estimators, 
thereby justifying inference using AIPW estimators with black-box outcome models that integrate advances from machine learning methods.
}

\bigskip
\begin{abstract}
Adaptive designs dynamically update treatment probabilities using information accumulated during the experiment. 
Existing theory for causal inference from adaptive experiments primarily assumes the superpopulation framework with independent and identically distributed units, and may not apply when the distribution of units evolves over time.
This paper makes two contributions. 
First, we extend the literature to the \fp\ framework, which allows for possibly nonexchangeable units, and establish the \db\ theory for causal inference under general adaptive designs using \ipwd\ (IPW) and augmented IPW (AIPW) estimators. 
%
\contri\
To alleviate the conservativeness inherent in variance estimation under \fp\ inference, we also introduce a covariance estimator for the AIPW estimator that 
becomes sharp when the residuals from the adaptive regression of potential outcomes on covariates are additive across units.
Our framework encompasses widely used adaptive designs, such as multi-armed bandits, \car, and sequential rerandomization, advancing the design-based theory for causal inference in these specific settings.
Second, as a methodological contribution, we propose an {\it adaptive covariate adjustment} approach for analyzing even nonadaptive designs.
The martingale structure induced by adaptive adjustment enables valid inference with black-box outcome estimators that would otherwise require strong assumptions under standard nonadaptive analysis.
\end{abstract}

\noindent%

\vfill

\newpage
\spacingset{1.8} 
\section{Introduction}\label{sec:intro}
Adaptive experimental designs, also known as sequential experimental designs, have gained increasing attention in recent years. 
Unlike nonadaptive designs, which fix treatment probabilities prior to the experiment, 
adaptive designs use information accumulated during the experiment to dynamically update treatment probabilities for subsequent units, offering multiple advantages in applications across business, medicine, and policy, among other areas \citep{bhatt2016adaptive, hadad2021designing, offer2021adaptive}. 
As a widely used example, 
\mab\ algorithms aim to maximize cumulative reward over time by dynamically adjusting treatment probabilities based on realized outcomes.  
Contextual bandits extend the \mab\ framework by incorporating contextual features, enabling personalized treatment assignment policies \citep{bubeck2012regret, dimakopoulou2017estimation, agrawal2019recent}.
Another prominent class of adaptive designs---including \car\ \citep{efron1971forcing,  bugni2018inference, ye2022inference}, minimization methods \citep{pocock1975sequential}, and \sr\ \citep{zhou2018sequential}---dynamically adjusts treatment probabilities to improve covariate balance across treatment groups. 

Inference of treatment effects remains a central objective in analyzing adaptive experiments.
However, dependence among data points collected over time complicates this task, especially when the adaptive treatment assignment mechanism is designed to optimize objectives other than estimating treatment effects \citep{hadad2021confidence}.
In this context, 
 \cite{melfi2000estimation} studied inference from adaptive designs with converging treatment probabilities, and established large-sample guarantees of a class of estimators that are strongly consistent when applied to nonadaptive data; see also \cite{van2008construction}. 
 \cite{hadad2021confidence} extended this framework to designs with nonconverging treatment probabilities, providing conditions for doubly robust inference from a class of {\ares}.
 \cite{bibaut2021post} further extended these results to contextual bandit designs and established weaker conditions for inference.
Additionally, focusing on specific adaptive designs, 
 \cite{bugni2018inference} studied inference under \car, and established theoretical properties of three commonly used hypothesis-testing procedures.
\cite{zhou2018sequential} studied \sr\ for experimental units enrolled in groups, and established its advantage in improving covariate balance over non\sr.
\cite{ham2023design} studied \mabs, and provided both confidence intervals and confidence sequences for arm-specific reward means and mean reward differences between arms.

This paper makes two main contributions. 
First, we advance the {\it finite-population, {\db}} theory for causal inference from adaptive designs, and establish guarantees of Wald-type inference based on the \ipwd\ (\ipwc) and augmented IPW (\aipwc) estimators. 
\contri\
Second, as a methodological contribution, we propose an {\it adaptive covariate adjustment} approach for causal inference from even nonadaptive designs, enabling valid inference with black-box outcome estimators that would otherwise require strong assumptions under standard nonadaptive analysis.
We provide the details below. 

\paragraph*{Design-based theory for causal inference from adaptive experiments.} 
Existing literature on causal inference from adaptive experiments primarily assumes the superpopulation framework with independent and identically distributed units \citep{hadad2021confidence, bibaut2021post}. 
However, methods developed under this framework may be biased when the distribution of units evolves over time; see \sec~\ref{sec:add_res} of the \sm\ for a simulated illustration.
We instead adopt the \fp, \db\ framework, which conditions on potential outcomes and covariates---or, equivalently, views them as fixed---and takes treatment assignment as the sole source of randomness for inference \citep{fisher1935design}.  
Accordingly, our theory accommodates nonexchangeable units arising from a wide range of underlying data-generating processes.
As a technical contribution, we introduce a new covariance estimator for the AIPW estimator that leverages adaptive outcome estimation to reduce the conservativeness inherent in variance estimation under  \fp\ inference \citep{neyman1923on}. 
The proposed estimator is generally conservative but 
becomes sharp when the residuals from the adaptive regression of potential outcomes on covariates are additive across units.

Our theory encompasses widely used adaptive designs such as multi-armed bandits, \car, and sequential rerandomization, advancing the design-based theory for causal inference in these specific settings \citep{bugni2018inference, zhou2018sequential, ham2023design}.
In particular, the theory of \cite{zhou2018sequential} does not include inference under \sr, and our paper fills this gap.

\paragraph*{Adaptive covariate adjustment for nonadaptive designs.}
Existing methods for covariate adjustment in nonadaptive experiment uses either all units \citep{lin2013agnostic, guo2023generalized, cohen2024no, qu2025randomization} or cross-fitting \citep{aronow2013middleton, wager2016high, wu2021design} to estimate outcome models. 
Such practices may induce complex dependence among unit-level adjusted outcomes, requiring strong consistency or stability conditions on the outcome estimators, or \bc\ when cross-fitting is used, for valid inference.

In contrast, we propose analyzing nonadaptive designs as if they were adaptive, using the {\it adaptively adjusted} AIPW estimator we developed for adaptive designs to perform {\it adaptive covariate adjustment}. 
The martingale structure underlying the adaptive adjustment enables valid covariate-adjusted inference using black-box outcome estimators that would otherwise require strong assumptions under standard nonadaptive analysis. 
Similar idea appears in \cite{luedtke2016statistical} and \cite{wager2024sequential} when studying the optimal treatment rule and treatment heterogeneity, respectively, under the superpopulation framework. 

\paragraph*{Notation.}
Let $\I(\cdot)$ denote the indicator function.
Let $0_m$ denote the $m \times 1$ zero vector and $I_m$ the $m \times m$ identity matrix. 
We omit the subscript $m$ when the dimension is clear from context.
For an $m\times n$ matrix $A = (a_{ij})_{m\times n}$, let $\fn{A} = (\sum_{i=1}^m\sum_{j=1}^n a_{ij}^2)^{1/2}$ denote the Frobenius norm of $A$.
For an $n\times n$ square matrix $A$, let $\lm(A)$ denote its smallest eigenvalue. 
For a collection of numbers $\{a_z \in \mbr: z\in\mz\}$, where $\mz$ is the index set, let $\diag(a_z)_{z\in\mz}$ denote the diagonal matrix with $a_z$ on the diagonal.
For a set of vectors $\{a_i \in \mbr^m: i = \ot{n}\}$, 
define its {\it sample covariance matrix} as $(n-1)^{-1}\sum_{i=1}^n (a_i - \bar a)(a_i - \bar a)^\T$, where $\bar a = n^{-1}\sum_{i=1}^n a_i$.
Let $\converged$ denote convergence in distribution.

\section{Setting}\label{sec:setup}
\subsection{Adaptive randomization}
Consider an experiment with $K\geq 2$ treatment levels, indexed by $z\in\mathcal{Z} =\{1,\ldots, K\}$, and a study population of $T$ units, indexed by $t=\ot{T}$.
An adaptive design sequentially randomizes each unit $t$ to a treatment level at time $t =\ot{T}$, and determines the assignment probabilities adaptively based on the information accumulated up to time $t$, referred to as the {\it history}.
For each unit $t$, let $Z_t\in\mz$ denote its treatment assignment, $Y_t\in\mbr$ the outcome of interest, and $X_t\in\mathbb R^J$ a vector of $J$ baseline covariates. 
Let $H_t$ denote the history up to time $t$, which includes the covariates, treatment assignments, and outcomes of units $1$ through $t-1$, as well as the covariates of unit $t$: \smallskip\\
\cl{$\htfull.$}
We characterize the assignment mechanism by the conditional distribution of $Z_t$ given $H_t$, as formalized in Definition~\ref{def:ad} below.

\begin{definition}[Adaptive randomization]\label{def:ad}
Let $\etz =\pr(Z_t = z\mid\hht)$
denote the probability of unit $t$ receiving treatment level $z\in\mz$ given $H_t$. 
For $t =\ot{T}$ and $z\in\mz$, $\etz$ is a prespecified function of $\hht$ that satisfies (i) $\sumz\etz = 1$ and (ii) $\etz\in (0,1)$. 
\end{definition}

Definition~\ref{def:ad} requires positive assignment probability $\etz$ for all treatment levels at each time $t = \ot{T}$, but is otherwise general, allowing $\etz$ to depend arbitrarily on the history $\hht$, as illustrated in Examples \ref{ex:mab}--\ref{ex:cov} below.
When $\etz =\pr(\zt =z\mid\hht) =\pr(\zt =z\mid\xbt )$, so that the assignment $\zt$ of unit $t$ is independent of the previous units, Definition~\ref{def:ad} reduces to a nonadaptive design that randomizes each unit independently.

\begin{example}\label{ex:mab}
Multi-armed bandit designs dynamically update treatment probabilities to maximize the expected cumulative outcome. 
The assignment probability $\etz$ at time $t$ is a function of past assignments and outcomes up to time $t-1$, $\{(Z_s, Y_s): s = \ot{t-1}\}$, a subset of the history $H_t$.
Contextual bandits further incorporate time-specific contextual information, represented by a covariate vector $X_t$ for each time $t =\ot{T}$. The corresponding $\etz$ is a function of the full history $\hht$.
\end{example}

\begin{example}\label{ex:cov}
{\it Biased-coin designs} \citep{efron1971forcing}  dynamically update treatment probabilities to improve covariate balance. 
The assignment probability $\etz$ at time $t$ depends on the covariates of unit $t$ and the covariate balance among units $1$ through $t-1$, and is a function of $\{(Z_s,X_s): s=\ot{t-1}\}\cup\{X_t\}$, a subset of the history $H_t$.
\end{example}

\subsection{Potential outcomes and treatment effects}
We define treatment effects using the potential outcomes framework \citep{neyman1923on, imbens2015causal}. 
Let $\ytz$ denote the potential outcome of unit $t$ if assigned to treatment level $z$.
The observed outcome equals the potential outcome under the realized treatment level: $Y_t =\sumz\indz \ytz  = Y_t(Z_t)$. 
Let $\byz =\meant\ytz $ denote the population average potential outcome under treatment level $z$, and let 
$\bmy = (\by(1),\ldots,\by(K))^\T$ denote the vector of $\byz$ across all $K$ levels.
A general goal of \fp\ causal inference is to estimate linear combinations of $\{\byz: \ziz\}$, denoted by 
\cll{\tau_C = C\bmy =(c_1^\T\bmy, \ldots, c_Q^\T\bmy)^\T \in \mbr^Q,}
where $Q$ is the prespecified number of estimands, and $C = (c_1, \ldots, c_Q)^\T \in \mathbb{R}^{Q \times K}$ is a prespecified coefficient matrix with $c_q \in \mathbb{R}^K$ for $q = 1, \ldots, Q$.
Often, each $c_q$ is specified as a contrast vector, so that $c_q^\T\bmy$ represents a contrast among $\{\byz: \ziz\}$. 
For example, in a treatment-control experiment with $K = 2$, setting $C = (-1, 1)$ yields $\tau_C =\bar Y(2)-\bar Y(1)$, the \fp\ average treatment effect.
Alternatively, setting $C$ as the identity matrix yields $\tau_C =\bmy$, targeting the average potential outcomes by treatment.
Without loss of generality, we assume a general $C$ with $K$ columns and full row rank unless specified otherwise.

\subsection{Estimation by \ipwf}\label{sec:estimators_def}
\iipwf, also known as inverse probability-of-treatment weighting or importance sampling weighting, is standard for inference from adaptive designs.
Given the assignment probability $\etz$, the \ipwd\ (\ipwc) estimator of $Y_t(z)$ is 
\beginy\label{eq:hytiz}
\hat{Y}_{t,\ipw}(z) =\dfrac{\indz }{e_t(z)} Y_t. 
\endy
Define
\beginy\label{eq:hyi}
\hyi(z) =\meant\hyti(z) ,\quad\hmyi = (\hyi(1),\ldots,\hyi(K))^\T,\quad 
\hti = C\hmyi
\endy
as the corresponding \ipwc\ estimators of $\bar{Y}(z)$, $\bmy$, and $\tc$, respectively. We establish the {\db} guarantees of $\hti$ in \sec~\ref{sec:ipw}.

The augmented \ipwd\ (\aipwc) estimator augments the IPW estimator with \oreg\  \citep{robins1994estimation}.
Let $\hmtz$ denote an {\it adaptive} estimator of $\ytz$ constructed using only information in $\hht$ \citep{hadad2021confidence}. 
We define 
\beginy\label{eq:hytaz}
\hytaz =\hytiz +\left\{ 1 -\dfrac{\indz }{\etz }\right\}\hmtz 
\endy
as the {\it adaptive} \aipwc\ estimator of $\ytz$ based on $H_t$, and define  
\beginy\label{eq:hya}
\hyaz =\meant\hytaz,\quad\hmyai = (\hya(1),\ldots,\hya(K))^\T,\quad\hta = C\hmyai
\endy as the corresponding \aipwc\ estimators of $\byz$, $\bmy $, and $\tau_C $, respectively, generalizing $\hyiz$, $\hmyi$, and $\hti$ in \eqref{eq:hyi}. 
We establish the {\db} guarantees of $\hta$ in \sec~\ref{sec:aipw}. 

\section{Design-based theory of the IPW estimator}\label{sec:ipw}
We establish in this section the {\db} guarantees of the \ipwc\ estimator $\hti$ defined in \eqref{eq:hyi} for estimating $\tc$.
The {\db} framework views covariates and potential outcomes as fixed---or, equivalently, conditions on them---and evaluates the sampling properties of $\hti$ with respect to the randomness in treatment assignments $\{Z_t: t =\ot{T}\}$. 

\subsection{Sampling properties}
Let 
\beginy\label{eq:vi}
\vi  =\diag\left[\meant\E\left\{\etzinv\right\}\ytzs\right]_{\ziz} -\meant\myt\myt^\T,
\endy
where $\myt = (Y_t(1),\ldots, Y_t(K))^\T$.

\begin{theorem}\label{thm:ipw_fs}
\prear. 
Then 
$\E(\hti) =\tau_C$ and $\cov(\hti) = T^{-1} C\vi\ct$ for \allc.  
\end{theorem}

\thm~\ref{thm:ipw_fs} establishes the unbiasedness of $\hti$ for $\tc$ and provides the explicit form of its covariance matrix. 
Setting $C = I_K$ implies that $\hyi(z)$ is unbiased for the population average potential outcome $\byz$, with $\vi = T\cov(\hmyi)$.
 
We next establish the central limit theorem for $\hti$ as $T\to\infty$. 
The \fp\ asymptotic regime embeds the study population into an infinite sequence of finite populations with increasing sizes $T = 1, 2,\ldots,$ and defines the asymptotic distribution of $\hti$ as the limit of its {\db} distributions along this sequence \citep{lehmann1975nonparametrics, lehmann1999elements, li2017general}.
Under adaptive randomization, it is intuitive to imagine a sequentially enrolled infinite population $\{X_t, Z_t,\ytz:\ziz\}_{t=1}^\infty$, and define the $T$-th finite population as the first $T$ units of this infinite population 
for $T = 1, 2,\ldots$.
However, our theory is general and applies not only to this intuitive scheme but also to the alternative scheme in which the finite populations consist of different units for each $T$.

Under the \db\ framework, at time $t$, the treatment assignments up to time $t-1$, $(Z_1, \ldots, Z_{t-1})$, are the only random elements in the history $H_t$, so that $H_t$ can take at most $|\mz|^{t-1}$ possible values. 
Let $\uet =\min_{H_t,\, \ziz}\etz$ denote the minimum of $\{e_t(z): z\in\mz\}$ over the $|\mz|^{t-1}\times |\mz|$ possible values of $(H_t,z)$ at time $t$.
Definition~\ref{def:ad} implies that $\uet > 0$ for all $t$. 
Let
$\lt =\max_{z\in\mathcal{Z}}|Y_t(z)|$ denote the maximum of $|\ytz|$ at time $t$. 
Recall the definition of $\vi$ from \eqref{eq:vi}. 
Let $\lm(\vi)$ denote the smallest eigenvalue of $\vi$, and let
\cll{\cvi=\diag\left[\ds\meant\dfrac{1}{\etz }\ytzs\right]_{\ziz} - \meant\myt\myt^\T}

with $\E(\cvi)=\vi$.
\cond~\ref{cond:ipw_clt} below specifies the regularity conditions we impose to establish the central limit theorem for $\hti$. 

\begin{condition}\label{cond:ipw_clt}
As $T\to\infty$, 
\begine[(i)]
\item\label{item:variance convergence_ipw} $\lmii \cdot\fn{\cvi  - \vi } =\op$.
\item\label{item:linde_ipw} $\lmii \cdot T^{-1}\max_{\ott }(\lt /\uet )^2 =\oo$. 
\ende
\end{condition}

We show in the \sms that  
\cll{
\vi = T^{-1}\ds\sumt\E\left\{\cov(\hmyti -\myt\mid\hht)\right\}, \quad 
\cvi = T^{-1}\ds\sumt\cov(\hmyti -\myt\mid\hht),
}
where $\{\hmyti -\myt: t = \ot{T}\}$ form a martingale difference sequence. 
\cond~\ref{cond:ipw_clt}\eqref{item:variance convergence_ipw} and \eqref{item:linde_ipw} correspond to the variance convergence and Lindeberg conditions, respectively, for applying the martingale central limit theorem to $\{\hmyti -\myt: t = \ot{T}\}$ \citep{brown1971martingale}.

\begin{theorem}\label{thm:ipw_clt}
\prear. If \cond~\ref{cond:ipw_clt} holds, then as $\ttinf$,
\cls{( C\vi\ct)^{-1/2}\cdot\sqrt{T} (\hat{\tau}_{\ipw, C} -\tau_C)
\converged\snrc}
for \allcf. 
\end{theorem}

Recall from \thm~\ref{thm:ipw_fs} that $\vi = T\cov(\hmyi)$.
This implies 
\beginy\label{eq:lm_vi}
\lm(\vi) = \ds\min_{C\in\mathbb R^{1\times K},\, \|C\|_2 =1} C \vi C ^\T= \ds\min_{C\in\mathbb R^{1\times K},\, \|C\|_2 =1} T \var(\hti),
\endy 
so that $\lm(\vi)$ represents the minimum scaled variance of $\hti = C\hmyi$ over all $C\in\mathbb R^{1\times K}$ with $\|C\|_2 = 1$.
Given $T$ as the total sample size, 
when potential outcomes vary across units and not all units are assigned to the same treatment arm, 
we do not expect to estimate $\hti = C\hmyi$, as a linear combination of $\{\byz:\ziz\}$, at a rate faster than $T^{-1/2}$ for any nonzero $C\in\mathbb R^{1\times K}$. 
As a result, $\var(\hti)$ generally does not vanish at a rate faster than $T^{-1}$, so \eqref{eq:lm_vi} suggests it is reasonable to assume that $\lm(\vi)$ is uniformly bounded away from 0.
In addition, outcomes in practice are typically bounded, making it reasonable to assume that $\ytz$ is uniformly bounded. 
\prop~\ref{prop:ipw_sc} below builds on these intuitions and provides sufficient conditions for \cond~\ref{cond:ipw_clt}.  

\begin{proposition}\label{prop:ipw_sc}
\prear. Let $v_t =\maxz\var\{\etzi\}$ denote the maximum variance of $\etzi$ over $\ziz$ at time $t$.
If $\lm(\vi)$ is uniformly bounded away from 0 as $T\to\infty$, then 
\begine[(i)]
\item\label{item:ipw_sc_cond1} 
\cond~\ref{cond:ipw_clt}\eqref{item:variance convergence_ipw} holds if either of the following conditions is satisfied: 
\begine[(a)]
\item\label{item:ipw_ve = 0}\textbf{Vanishing variance:}
$T^{-1}\sumt v_t = o(1)$.
\item\label{item:ipw_limited dep}
\textbf{\llrd:} \prelimdepsc $\maxt v_t= o(T^{1-\beta})$, and \limdept, $\{\etz:\ziz\}$ \limdep. 
\ende 
\item\label{item:ipw_sc_cond2} 
\cond~\ref{cond:ipw_clt}\eqref{item:linde_ipw} holds if $  \max_{\ott } \lt /\uet = \tsr  \cdot \oo$.
If further $\ytz$ is uniformly bounded, then \cond~\ref{cond:ipw_clt}\eqref{item:linde_ipw} holds if $\max_{\ott } \uet^{-1} = \tsr  \cdot \oo$. 
\ende
\end{proposition} 

\prop~\ref{prop:ipw_sc}\eqref{item:ipw_sc_cond1} implies that when $\lm(\vi)$ is uniformly bounded away from 0, \cond~\ref{cond:ipw_clt}\eqref{item:variance convergence_ipw} holds if either $\etz$ has vanishing variance or depends only on recent history.
Specifically, \prop~\ref{prop:ipw_sc}\eqref{item:ipw_ve = 0} requires $\var\{\etzi\}$ to vanish asymptotically, which is likely to hold when $\etz$ converges. 
In contrast, \prop~\ref{prop:ipw_sc}\eqref{item:ipw_limited dep} imposes much weaker restrictions on $\var\{\etzi\}$ but requires $\etz$ to  depend only on recent history.
When $\etz$ is uniformly bounded away from 0, as in the $\epsilon$-greedy \mab, the corresponding $v_t$ is uniformly bounded, with $\maxt v_t= o(T^{1-\beta})$ for any $\beta < 1$. 
This allows choosing $\beta$ arbitrarily close to 1 to accommodate longer-range dependence. 

\prop~\ref{prop:ipw_sc}\eqref{item:ipw_sc_cond2} provides intuition for the rate at which $\uet$ may vanish under \cond~\ref{cond:ipw_clt}\eqref{item:linde_ipw}.
When $\ytz$ is uniformly bounded, the condition in \prop~\ref{prop:ipw_sc}\eqref{item:ipw_sc_cond2} holds if there exist constants $\delta, c_0 > 0$ such that 
$\min_{\ott} \uet \geq c_0 T^{-1/2 +\delta}$.
The rate $-1/2 + \delta$ is more stringent than the condition assumed by \cite{bibaut2021post} for identically distributed units, as their approach used reweighting to mitigate the impact of highly variable components in $\hyiz$.
However, this reweighting strategy may not be desirable when units are nonexchangeable. See \sec~\ref{sec:add_res} of the \sm\ for a simulated illustration.

\begin{remark}\label{rmk:rate}
When the treatment probability $\etz$ vanishes asymptotically for one or more treatment arms, 
the diagonal elements of $\vi$ may differ in order due to the small effective sample sizes for those arms.
As a result, the estimated average potential outcomes $\hyi(z)$ may converge at different rates across $\ziz$. When contrasts among treatment arms are of interest,
the rate required for the central limit theorem is typically determined by the arm with the smallest effective sample size.
Therefore, the central limit theorem may hold under weaker sufficient conditions than those in \prop~\ref{prop:ipw_sc}. 
For example, when the number of units in the first arm is substantially smaller than $T$ in order, the corresponding scaled variance of $\hyi(1)$, $T\var\{\hyi(1)\}$, may diverge as $T\rightarrow \infty$. 
In such cases, we may relax the sufficient conditions in \prop~\ref{prop:ipw_sc}, which are derived under only the assumption that $T\var\{\hyi(1)\}$, as the first diagonal element of $\vi$, is bounded away from 0.

For ease of presentation and to avoid technical clutter, we focus on the central limit theorem in \thm~\ref{thm:ipw_clt}, which simultaneously accommodates all contrast matrices and allows for arms with much smaller effective sample sizes relative to $T$. 
\hfill \qedsymbol
\end{remark}

\subsection{Variance estimation and confidence sets}\label{sec:ipw_var_est} 
We now construct confidence sets for $\tc$ based on $\hti$. 
Recall from \eqref{eq:vi} that $\myt = (\yt(1), \ldots, \yt(K))^\T$. 
Let $\hmyti = (\hat{Y}_{t,\ipw}(1),\ldots,\hat{Y}_{t,\ipw}(K))^\T$ be its \ipwc\ estimator, and define
\beginy\label{eq:hvi}
\hvi
=\dfrac{1}{T-1}\sumt (\hmyti -\hmyi ) (\hmyti -\hmyi )^\T 
\endy
as the sample covariance matrix of $\{\hmyti: t = \ot{T}\}$. \prop~\ref{prop:ipw_cov_est} below justifies using $\hvi$ to estimate $\vi$. 

\begin{condition}\label{cond:ipw_cs}
As $T\to\infty$, 
 $ \lmip^{-2}\cdot  T^{-2}\sumt\lt ^4/\uet ^3 = o(1)$.
\end{condition}

\begin{proposition}\label{prop:ipw_cov_est}
\prear. Let $
S = (T-1)^{-1}\sumt (\myt -\bmy ) (\myt -\bmy )^\T$
denote  the sample covariance matrix of $\{\myt: t = \ot{T}\}$. For \allc, 
\begine[(i)]
\item $
\E(C\hvi\ct) = C\vi\ct + CS\ct$, where $CS\ct = 0$ if and only if $\tau_{C,t} = C \myt$ is constant across $t = \ot{T}$.  
\item If Conditions~\ref{cond:ipw_clt}--\ref{cond:ipw_cs} hold, then $C\hvi\ct =C\vi\ct + C S\ct +  \lm(\vi) \cdot \op$. 
\ende
\end{proposition}


Recall from \thm~\ref{thm:ipw_fs} that $\cov(\hti) = T^{-1}C\vi\ct$.
\prop~\ref{prop:ipw_cov_est} implies that $C\hvi\ct$ is a conservative, and therefore valid, estimator of $C\vi\ct = T\cov(\hti)$ in both expectation and probability limit, with an upward bias of $CS\ct$. 
This conservativeness in variance estimation is inherent to the {\db} framework \citep{neyman1923on, imbens2015causal}, and vanishes if and only if the individual analogs of $\tc$, $\tau_{C,t} = C\myt$, are constant across all units. 
When covariates are available, this upward bias can be reduced by constructing a lower bound for $S$ to adjust $\hvi$. 
We present the details in \prop~\ref{prop:aipw_cov_est_2} in \sec~\ref{sec:aipw} as part of the theory for the AIPW estimator.

Let
\cls{
\si =\left\{\tau: (\hat{\tau}_{\ipw, C} -\tau )^\T (T^{-1}C\hat{V}_{\ipw}\ct)^{-1} 
(\hat{\tau}_{\ipw, C} -\tau )\le \chisqa\right\}
}

denote the $100(1-\alpha)\%$ confidence set for $\tc$ based on $(\hti,\hvi)$ and the normal approximation, where $\rc$ denotes the rank of $C$ and ${\chisqa}$ denotes the $1-\alpha$ quantile of the chi-square distribution with $\rc$ degrees of freedom. 
 \thm~\ref{thm:ipw_cs} below follows from \thm~\ref{thm:ipw_clt} and \prop~\ref{prop:ipw_cov_est}, and establishes the validity of large-sample Wald-type inference based on $\si$. This concludes our discussion of the IPW estimator.

\begin{theorem}\label{thm:ipw_cs}
\prear. 
If Conditions~\ref{cond:ipw_clt}--\ref{cond:ipw_cs} hold, then\\
\cl{$\liminf_{T\to\infty}\pr (\tau_C\in\si) \ge 1 -\alpha$ \quad for all $\alpha\in (0, 1)$}
 for \allcf. 
\end{theorem}

\section{Design-based theory of the AIPW estimator}\label{sec:aipw}
We now extend \sec~\ref{sec:ipw} to the \aipwc\ estimator $\hta$, as defined in \eqref{eq:hya}.
Recall from \eqref{eq:hytaz} that $\hta$ is constructed using
\cls{
\hytaz =\hytiz +\left\{ 1 -\dfrac{\indz }{\etz }\right\}\hmtz,}

where $\hmtz$ denotes an adaptive estimator of $\ytz$, constructed using only the information in  $H_t$. The results for $\hta$ include those for $\hti$ as a special case when $\hmtz = 0$.

\subsection{Sampling properties}\label{sec:aipw_clt}
Let $\atz =\ytz -\hmtz$ denote an adjusted version of $\ytz$, interpreted as the residual from the outcome estimator $\hmtz$. 
Let $\mat = (\at(1),\ldots,\at(K))^\T$, analogous to $\myt$. 
Let 
\beginy\label{eq:va}
\begin{array}{rcl}
\va &=& \diag\left[\ds\meant\E\left\{\dfrac{\atzs}{\etz }\right\}\right]_{\ziz} - \ds\meant\E(\mat\matt),
\bigskip\\
\cva &=& \diag\left[\ds\meant\dfrac{\atzs}{\etz }\right]_{\ziz} - \ds\meant\mat\matt
\end{array}
\endy
denote the  \aipwc\ analogs of $\vi$ and $\cvi$, respectively, with $\E(\cva)=\va$. 

Recall that $\lt =\max_{z\in\mathcal{Z}}|Y_t(z)|$, and $\uet =\min_{H_t,\, \ziz}\etz$ denotes the minimum of $\etz$ over all $|\mz|^t$ possible values of $(H_t,z)$. 
Let $\mt=\max_{H_t,\, \ziz} |\hmtz|$ denote the maximum of $|\hmtz|$ over the same set. 
\thm~\ref{thm:aipw} below extends Theorems~\ref{thm:ipw_fs}--\ref{thm:ipw_clt}, and establishes the unbiasedness and central limit theorem for $\hta$.

\begin{condition}\label{cond:aipw_clt}
As $T\to\infty$, 
\begine[(i)] 
\item\label{item:clt_aipw_var_conv} $\lmai \cdot\fn{\cva - \va } =\op$.
\item\label{item:clt_aipw_lindeberg}$\lmai \cdot T^{-1}\max_{\ott } (\lt + \mt)^2/\uet^2 =\oo$. 
\ende
\end{condition}

\begin{theorem}\label{thm:aipw}
\prear. 
For \allcf, 
\begine[(i)]
\item\label{item:aipw_mean_var}
 $\E(\hta) =\tc$, \quad $\cov(\hta) = T^{-1} C\va\ct $. 
\item\label{item:aipw_clt} If \cond~\ref{cond:aipw_clt} holds, then $( C\va \ct)^{-1/2}\cdot\sqrt{T} 
(\hat{\tau}_{\aipw, C} -\tau_C)
\converged\snrc $.
\ende
\end{theorem}

Recall from \thm~\ref{thm:ipw_fs} that the covariance matrix of $\hti$ is $T^{-1} C\vi\ct$.
\thm~\ref{thm:aipw} further establishes that the covariance matrix of $\hta$ is $T^{-1}C\va\ct$.
The relative efficiency of $\hta$ compared to $\hti$ therefore depends on the difference $\va-\vi$.
Let $\hmyta = (\hyta(1),\ldots,\hyta(K))^\T$ denote the  \aipwc\ estimator of $\myt = (Y_t(1),\ldots, Y_t(K))^\T$, generalizing the \ipwc\ estimator $\hmyti$ as defined in \eqref{eq:hvi}.
Echoing the discussion following \cond~\ref{cond:ipw_clt}, 
we show in the \sms that 
\beginy\label{eq:aipw_va_vi}
\vi = T^{-1}\sumt\E\left\{\cov(\hmyti -\myt\mid\hht)\right\}, \quad \va  
= T^{-1}\sumt\E\left\{\cov(\hmyta -\myt\mid\hht)\right\},
\endy 
where the elements of $\hmyti -\myt$ and $\hmyta -\myt$ are
\cll{\ds
\hyti(z) - \yt = \left\{\dfrac{\indz }{\etz } - 1\right\}   \yt,
\quad 
\hyta(z) - \yt 
= \left\{\dfrac{\indz }{\etz } - 1\right\}  \left\{\yt - \hmtz\right\} 
}

for $\ziz$, respectively, by \eqref{eq:hytiz} and \eqref{eq:hytaz}.
Heuristically,  if $\hmtz$ is a reasonable estimator of $\ytz$, then $\hmyta$ tends to be less variable than $\hmyti$, with $\cov(\hmyta -\myt\mid\hht) < \cov(\hmyti -\myt\mid\hht)$. This implies $\va < \vi$ by \eqref{eq:aipw_va_vi}, suggesting that $\hta$ improves efficiency over $\hti$. We illustrate this efficiency gain through simulation in \sec~\ref{sec:simulation}.  

\cond~\ref{cond:aipw_clt} extends \cond~\ref{cond:ipw_clt}, and specifies the regularity conditions for the central limit theorem of $\hta$ in \thm~\ref{thm:aipw}\eqref{item:aipw_clt}.
Specifically, \cond~\ref{cond:aipw_clt}\eqref{item:clt_aipw_var_conv} and \eqref{item:clt_aipw_lindeberg} are the variance convergence and Lindeberg conditions, respectively, for applying the martingale central limit theorem to the sequence $\{\hmyta - \myt: t = \ot{T}\}$. 
Parallel to the discussion following \thm~\ref{thm:ipw_clt}, $\lm(\va)$ is the minimum scaled variance of $\hta = C\hmya$ over all $C \in \mathbb R^{1\times K}$ with $\|C\|_2=1$, and it is reasonable to assume that $\lm(\va)$ is uniformly bounded away from 0.
In addition, outcomes in practice are typically bounded, making it also reasonable to assume that both $\ytz$ and $\hmtz$ are uniformly bounded. 
\prop~\ref{prop:aipw_sc} below builds on these intuitions and provides four sufficient conditions for \cond~\ref{cond:aipw_clt}. 
See \prop~\ref{prop:aipw_sc_app} and Lemma \ref{lem:aipw_sc_app} in the \sms for more general versions that allow for vanishing $\etz$ and unbounded outcomes. 

Recall that $v_t =\maxz\var\{\etzi\}$ from \prop~\ref{prop:ipw_sc}.
Let $\omt =\maxz\var\{\hmtz\}$ denote the maximum variance of $\hmtz$ at time $t$.

\begin{proposition}\label{prop:aipw_sc}
\prear.
If as $\ttinf$, (a) both $\lm(\va)$ and $\etz$ are uniformly bounded away from 0, and (b) both $ \ytz $ and $ \hmtz$ are uniformly bounded, 
then \cond~\ref{cond:aipw_clt} holds if any of the following conditions holds:
\begine[(i)]
\item\label{item:V_i_main_c}{\bf Vanishing $\var\{\etzi\}$ and $\var\{\hmtz\}$:} $T^{-1}\sumt v_t = \oo$; \ \
$T^{-1}\sumt\omt = \oo$.
\item\label{item:V_ii_main_c}{\bf Vanishing $\var\{\etzi\}$ and \lrd\ of $\hmtz$:}\\
$\meant\sqrt{v_t} =\oo$;\\
\prelimdep, $\{\hmtz\}_{\ziz}$ \limdep. 
\item\label{item:V_iii_main_c}{\bf Vanishing $\var\{\hmtz\}$ and \lrd\ of $\etz$:}\\
$T^{-1}\sumt\omt = o(1)$;\\
\prelimdep, $\{\etz\}_{\ziz}$ \limdep.
\item\label{item:V_iv_main_c} 
{\bf \llrd\ of $\etz$ and $\hmtz$:}
\prelimdepc, $\{\etz,\hmtz:\ziz\}$ \limdep.
\ende
\end{proposition}

Each of the four sufficient conditions in \prop~\ref{prop:aipw_sc} imposes restrictions on both $\etz$ and $\hmtz$.
Together, they require:\\
- \ $\etz$ has {\it either} vanishing $\var\{\etzi\}$ {\it or} \lrd\ on past units, {\it and}\\
- \ $\hmtz$ has {\it either} vanishing $\var\{\hmtz\}$ {\it or} \lrd\ on past units,\\
as illustrated by the table below:

\begin{center}
    \begin{tabular}{c|c|c} \hline
        & Vanishing $\var\{\hmtz\}$ & \renewcommand{\arraystretch}{.6}\begin{tabular}{c}Limited-range\\ dependence of $\hmtz$\end{tabular}\\\hline
Vanishing $\var\{\etzi\}$ & \prop~\ref{prop:aipw_sc}\eqref{item:V_i_main_c} & \prop~\ref{prop:aipw_sc}\eqref{item:V_ii_main_c}\\
\llrd\ of $\etz$ & \prop~\ref{prop:aipw_sc}\eqref{item:V_iii_main_c} & \prop~\ref{prop:aipw_sc}\eqref{item:V_iv_main_c} \\\hline
    \end{tabular}
\end{center}

\smallskip

Consider the special case in which the sequence of finite populations consists of the first $T$ units of a sequentially enrolled infinite population $\{X_t, Z_t,\ytz:\ziz\}_{t=1}^\infty$ for $T = 1, 2,\ldots$.
By the dominated convergence theorem \citep{spall2003introduction}, 
the conditions of vanishing $\var\{\etzi\}$ in \prop~\ref{prop:aipw_sc}\eqref{item:V_i_main_c}--\eqref{item:V_ii_main_c}, $T^{-1}\sumt v_t = \oo$ and $T^{-1}\sumt \sqrt v_t = \oo$,  hold if $\etz$ converges in probability to a constant as $t\to\infty$.
In contrast, the \lrd\ condition on $\etz$ in \prop~\ref{prop:aipw_sc}\eqref{item:V_iii_main_c}--\eqref{item:V_iv_main_c} allows for nonconverging $\etz$. 
The same applies to the outcome estimator $\hmtz$. 
Moreover, there is no restriction on $\beta$ beyond $\beta \in [0,1)$, so we can choose $\beta$ arbitrarily close to 1 to allow for longer-range dependence.

\subsection{Variance estimation and confidence sets}\label{sec:aipw_var_est}
We now construct confidence sets for $\tc$ based on $\hta$.
Recall from \eqref{eq:aipw_va_vi} that $\hmyta = (\hyta(1),\ldots,\hyta(K))^\T$ denotes the  \aipwc\ estimator of $\myt = (Y_t(1),\ldots, Y_t(K))^\T$. 
Parallel to $\hvi$ as defined in \eqref{eq:hvi}, define 
\cll{
\hva  
=\dfrac{1}{T-1}\ds\sumt (\hmyta -\hmya) (\hmyta -\hmya)^\T,  
}

as the sample covariance matrix of $\{\hmyta : t =\ot{T}\}$, where $\hmya = \meant \hmyta$.
Recall that $S$ denotes the sample covariance matrix of $\{\myt: t = \ot{T}\}$.  %
\prop~\ref{prop:aipw_cov_est} below extends \prop~\ref{prop:ipw_cov_est}, and justifies using $C\hva\ct$ to estimate $C\va \ct$, with an upward bias of $CS\ct$ that vanishes when $C\myt$ is constant across $\ott$. 

\begin{condition}\label{cond:aipw_cs}
As $T\to\infty$, 
 $ \lmap^{-2}\cdot  T^{-2}\sumt (\lt +\mt) ^4/\uet ^3 = o(1)$.
\end{condition}

 
\begin{proposition}\label{prop:aipw_cov_est}
\prear. For \allc,
\begine[(i)]
\item\label{item:aipw_var_est_exact} $\E(C\hva \ct) = C\va\ct + CS\ct$, where 
$CS\ct = 0$ if and only if $\tau_{C,t} = C \myt$ is constant across $t = \ot{T}$. 
\item\label{item:aipw_var_est_asym} If Conditions~\ref{cond:aipw_clt}--\ref{cond:aipw_cs} hold, then $C\hva\ct = C\va\ct  + CS\ct + \lm(\va) \cdot \op$.  
\ende
\end{proposition}

As previewed after \prop~\ref{prop:ipw_cov_est}, we can reduce the upward bias $CS\ct$ by adjusting $\hva$ using a lower bound for $S$.
We formalize the improved covariance estimator below. 

Let $\hbmt = (\hmt(1),\ldots,\hmt(K))^\T$ denote the vector of outcome estimators for unit $t$ across all $K$ levels.
Let $\tmyta =\hmyta -\hbmt$, and define
\cll{\tva  
=\dfrac{1}{T-1}\ds\sumt (\tmyta -\tmya) (\tmyta -\tmya)^\T}

as the sample covariance matrix of $\{\tmyta: t = \ot{T}\}$, where $\tmya = \meant \tmyta$.

\begin{proposition}\label{prop:aipw_cov_est_2}
\prear. 
If Conditions~\ref{cond:aipw_clt}--\ref{cond:aipw_cs} hold, then $C\tva\ct =  C\vaipw\ct + C\Omega\ct + \lm(\va)\cdot \op$, where $\Omega$ denotes the sample covariance matrix of $\{\myt - \hbmt: t = \ot{T}\}$.
\end{proposition}

\props~\ref{prop:aipw_cov_est}--\ref{prop:aipw_cov_est_2} imply that 
both $\hva$ and $\tva$ are asymptotically valid for estimating $\va$, 
with upward biases $S$ and $\Omega$, respectively. 
Since $S$ and $\Omega$ are the sample covariance matrices of $\myt$ and $\myt - \hbmt$, respectively, 
we heuristically expect $\Omega \leq S$ when $\hmtz$ is a reasonable estimator of $\ytz$, making $\tva$ less conservative than $\hva$ for inference. 
We illustrate the improved precision through simulation in \sec~\ref{sec:simulation}.

Let
\bigskip\\
\cl{$
\begin{array}{l}
\sa =\left\{\tau: (\hat{\tau}_{\aipw, C} -\tau )^\T (T^{-1}C\hva \ct)^{-1} 
(\hat{\tau}_{\aipw, C} -\tau )\le \chisqa\right\},\\
\tsa =\left\{\tau: (\hat{\tau}_{\aipw, C} -\tau )^\T (T^{-1}C\tva \ct)^{-1} 
(\hat{\tau}_{\aipw, C} -\tau )\le\chi^2_{\rc, 1-\alpha}\right\}
\end{array}
$}

\smallskip

denote the $100(1-\alpha)\%$ Wald-type confidence sets for $\tc$ based on $(\hta,\hva)$ and $(\hta, \tva)$, respectively. 
\thm~\ref{thm:aipw_cs} below follows from \thm~\ref{thm:aipw} and  \props~\ref{prop:aipw_cov_est}--\ref{prop:aipw_cov_est_2}, and establishes the validity of large-sample Wald-type inference based on $\sa$ and $\tsa$. This concludes our theory of the AIPW estimator for estimating treatment effects under \ar.

\begin{theorem}\label{thm:aipw_cs}
\prear. 
If Conditions~\ref{cond:aipw_clt}--\ref{cond:aipw_cs} hold,  
then for \allcf,  
\cls{
\ds\liminf_{T\to\infty}\pr (\tau_C\in\sa)
\ge 1 -\alpha, \quad \liminf_{T\to\infty}\pr (\tau_C\in\tsa)
\ge 1 -\alpha \ \ \text{for all $\alpha\in (0, 1)$}.}
\end{theorem}

\subsection{Adaptive analysis of nonadaptive designs}\label{sec:adaptive analysis}
Recall that \defo\ includes both adaptive and nonadaptive designs. 
Therefore, the definition of the AIPW estimator $\hta$ in \eqref{eq:hya} with an adaptively constructed outcome estimator $\hmt(z)$, along with all its \db\ guarantees in \secs~\ref{sec:aipw_clt}--\ref{sec:aipw_var_est}, also applies to nonadaptive designs.

As previewed in the \sec~\ref{sec:intro}, we propose using $\hta$ for {\it \aca} even in nonadaptive designs, especially when outcome estimation via black-box machine learning methods is desired, for the following benefits:
\begine[(i)]
\item Compared with covariate adjustment strategies that construct $\hmtz$ using all units, \aca\ induces a martingale structure that enables valid inference under weaker regularity conditions on $\hmtz$, thereby allowing more flexible use of black-box machine learning estimation methods to improve efficiency. 
\item Compared with covariate adjustment strategies that construct $\hmtz$ using \cfbc, \aca\ improves efficiency.
\ende  
We illustrate these improvements through simulation in \sec~\ref{sec:simulation}. 

\section{Extension to block adaptive designs}\label{sec:block_ad}
\subsection{Treatment assignment mechanism}
The discussion so far has focused on designs with treatment assigned adaptively at the unit level.
We now extend our theory to {\it block adaptive randomization}, where treatment is assigned sequentially, group by group, and the treatment probabilities for a new group may depend on the covariates, assignments, and outcomes of earlier groups.

Consider an experiment with $K \geq 2$ treatment levels, indexed by $z\in\mz =\{\ot{K}\}$, and a study population of $T$ groups of units, indexed by $t=\ot{T}$.
A block adaptive design sequentially randomizes the units in group $t$ to treatments at time $t$, with the assignment probabilities determined by the history up to time $t$.
Let $\nt$ denote the size of group $t$, and index the units in group $t$ by $\{ti: i =\ot{\nt}\}$. 
For each unit $ti$, let $Z_{ti} \in \mz$ denote its treatment assignment, $Y_{ti} \in \mathbb{R}$ its outcome, and $X_{ti} \in \mathbb{R}^J$ its covariate vector.
Renew 
\cll{\zbt = (Z_{t1},\ldots, Z_{t,\nt}), \quad\ybt =  (Y_{t1},\ldots, Y_{t,\nt}), \quad \xbt = (X_{t1},\ldots, X_{t, \nt})}
as the collections of $Z_{ti}$, $Y_{ti}$, and $X_{ti}$ for units in group $t$, and\cll{\htfull}
as the history up to time $t$ with the renewed $(Z_s, Y_s, X_s)$ and $X_t$. 
This and all subsequent reuse of notation are justified by the fact that $(Z_t, Y_t, X_t, H_t)$ revert to their original definitions in \sec~\ref{sec:setup} when $\nt = 1$ for all $t$.
Definition~\ref{def:block_ad} below generalizes Definition~\ref{def:ad} and characterizes block adaptive randomization via the conditional distribution of $\zbt $ given $\hht$.

\begin{definition}[Block adaptive randomization]\label{def:block_ad}
Let $\ebtz =\pr(Z_t =\bz\mid H_t)$ denote the probability that $\zbt = (Z_{t1},\ldots, Z_{t,\nt})$ equals $\bz = (z_1,\ldots, z_{\nt})\in\mz^{\nt} $ given $H_t$. 
Let $\etiz =\pr(Z_{ti} = z\mid\hht)$  
denote the marginal probability of $Z_{ti} = z$ given $H_t$ for $\ziz$, obtained by summing $\ebtz$ over all possible values of $\zt \setminus \{Z_{ti}\}$.
For $t =\ot{T}$ and all $\bz\in\mz^{\nt}$, $\ebtz$ is a prespecified function of $\hht$ that satisfies: (i) $\sum_{\bz\in\mathcal{Z}^{\nt}}\ebtz = 1$; (ii) $\etiz\in (0,1)$ for all  $i =\ot{\nt}$ and $z\in\mz$.
\end{definition}

Definition~\ref{def:block_ad} requires that each unit-treatment pair $(ti, z)$ has positive assignment probability $\etiz$ at each time $t = \ot{T}$, but is otherwise general, allowing $\etiz$ to depend arbitrarily on the history $\hht$.
Examples \ref{ex:paired}--\ref{ex:rerand} below review pairwise sequential randomization and \sr\ as two special cases previously studied in the literature. 
When $\ebtz =\pr(\zbt =\bz\mid\hht) =\pr(\zbt =\bz\mid\xbt )$, so that the assignment $\zt$ of group $t$ is independent of the previous groups, Definition~\ref{def:block_ad} reduces to a nonadaptive block design that conducts an independent randomization within each group.
When $\nt=1$ for all $t$, Definition~\ref{def:block_ad} reverts to the unit-level adaptive randomization in Definition~\ref{def:ad}.  

\begin{example}[Pairwise sequential randomization \citep{psr}]\label{ex:paired}
Consider a treatment-control experiment with $T$ pairs of units, indexed by $t =\ot{T}$.
Pairwise sequential randomization randomly assigns one unit in pair $t$ to treatment and the other to control at each time $t$, with assignment probabilities determined by the covariates of these two units, $X_t = (X_{t1}, X_{t2})$, and the current covariate balance among pairs $1$ through $t-1$. 
The corresponding $\ebtz$ is a function of $\{(Z_s, X_s): s =\ot{t-1}\}\cup\{X_t\}$.
\end{example}

\begin{example}[Sequential rerandomization \citep{zhou2018sequential}]\label{ex:rerand}
Consider an experiment with $T$ groups of units, indexed by $t=\ot{T}$. 
Sequential rerandomization randomly assigns the units in group $t$ to treatments at each time $t$ as follows: for a tentative assignment $\bz\in\mz^{\nt}$ for the $\nt$ units in group $t$,
 accept $\zbt =\bz$ if $\{(Z_{s}, X_{s}): s =\ot{t-1}\}\cup (\bz, \xbt)$ satisfies a prespecified balance criterion and rerandomize by permuting the elements of $\bz$ otherwise. 
The corresponding $\ebtz $ is a function of $\{(Z_{s}, X_{s}): s =\ot{t-1}\}\cup\{\xbt\}$. 
\end{example}

\subsection{Inference using the IPW and AIPW estimators}
Let $\ytiz$ denote the potential outcome of unit $ti$ under treatment level $z$.
Renew 
\cll{
\byz = \meanti \ytiz,\where N =\sumt \nt,}

as the population average potential outcome under treatment level $z$, with $N$ denoting the total number of units across all groups.
The goal is to infer a set of linear combinations of the renewed $\{\byz: \ziz\}$, denoted by $\tau_C = C\bmy$, where $\bmy = (\by(1),\ldots,\by(K))^\T$ is the vector of the renewed $\byz$ and $C$ is a prespecified coefficient matrix with $K$ columns.

Recall from Definition~\ref{def:block_ad} that $\etiz =\pr(Z_{ti} = z\mid\hht)$ denotes the treatment probability for unit $ti$ given the history $\hht$. 
Parallel to the IPW and AIPW estimators of $\ytz$ under \ur\ in \eqref{eq:hytiz} and \eqref{eq:hytaz}, define 
\cll{
\hy_{ti,\ipw}(z) = \dfrac{\I(\zti = z)}{\etiz} Y_{ti}, \quad 
\hytiaz =\hat Y_{ti,\ipw}(z) +\left\{1-\dfrac{\I(Z_{ti} = z)}{\etiz }\right\}\hmtiz 
}

as the \ipwc\ and \aipwc\ estimators of $\ytiz$, respectively, where $\hmtiz$ is an {\it adaptive} estimator of $\ytiz$ constructed using only information in $\hht$. 
Renew
\beginy\label{eq:estimators_block}
\begin{array}{llll}
\hyiz = \displaystyle\meanti\hy_{ti,\ipw}(z),&\hmyi = (\hyi(1),\ldots,\hyi(K))^\T,&\hti = C\hmyi,\medskip\\
\hyaz = \displaystyle\meanti\hytiaz,&\hmya = (\hya(1),\ldots,\hya(K))^\T,&\hta = C\hmya 
\end{array}
\endy 
as the \ipwc\ and \aipwc\ estimators of the renewed $\by(z)$, $\bmy$, and $\tau_C$, respectively. 
We establish below the {\db} guarantees of the renewed $\hta$ in \eqref{eq:estimators_block} for estimating the renewed $\tc$. 
The results include $\hti$ as a special case with $\hat m_{ti}(z) = 0$ for all $ti$ and $z$. 

Define $\bytz = \meani \ytiz$ as the average potential outcome in group $t$ under treatment level $z$, and renew 
\beginy\label{eq:hytaz_block}
\hytaz = \meani\hytiaz
\endy as the \aae\ of $\bytz$. 
Let $\pt = \nt/N$ denote the proportion of units within group $t$, and define $\rhot = T \cdot \pt$. 
By direct algebra, 
\beginy\label{eq:hyaz_block}
\byz = \sumt \pt \bytz =  \meant \rhot \bytz,\quad 
\hyaz =\sumt \pt \hytaz = \meant \rhot \hytaz,
\endy
%
paralleling the expressions of $\byz$ and $\hyaz$ in \eqref{eq:hya} under \ur, with the unit-level $\{\ytz, \hytaz\}$ in \eqref{eq:hytaz} replaced by the group-level $\{\rhot\byt(z), \rhot\hytaz\}$ in  \eqref{eq:hyaz_block}. 
Therefore, all results in \sec~\ref{sec:aipw} extend to inference of the renewed $\tc$ using the renewed $\hta$ under \br\ after replacing $\{\ytz, \hytaz\}$ with $\{\rhot\byt(z), \rhot\hytaz\}$. 
We state the main results below.

Renew $\hmyta = (\hyta(1),\ldots,\hyta(K))^\T$ as the vector of the renewed $\hyta(z)$ in \eqref{eq:hytaz_block}.
Renew $\hbmt = (\hmt(1),\ldots,\hmt(K))^\T$, where $\hmtz = \meani \hmtiz$, and $\tmyta =\hmyta -\hbmt$ with the renewed $\hmyta$ and $\hbmt$. 
Renew 
\beginy\label{eq:hva_block}
\begin{array}{ll}
\hvba 
 =\dfrac{1}{T-1}\ds\sumt (\rhot\hmyta -\hmya ) (\rhot\hmyta -\hmya )^\T,\medskip\\
\tva  
=\dfrac{1}{T-1}\ds\sumt (\rhot\tmyta -\tmya) (\rhot\tmyta -\tmya)^\T
\end{array}
\endy
as the sample covariance matrices of $(\rhot\hmyta)_{t=1}^T$ and $(\rhot\tmyta)_{t=1}^T$, respectively, and
renew $\sa$ and $\tsa$ as the $100(1-\alpha)\%$ Wald-type confidence sets for $\tc$ based on the renewed $(\hta,\hva)$ and $(\hta,\tva)$ in \eqref{eq:estimators_block} and \eqref{eq:hva_block}, respectively.
Renew 
\cls{\vba = \ds\meant \E\left\{\cov(\rhot\hmyta\mid\hht)\right\}, 
\quad
\cva = \ds\meant  \cov(\rhot\hmyta\mid\hht),}

and renew
\cls{\lt =\ds\max_{z\in\mathcal{Z},\, \iit}|\ytiz|$,\quad $\mt = \ds\max_{H_t, \, \ziz, \, \iit} \hmtiz$, \quad $\uet = \ds\min_{H_t, \, \ziz, \, \iit} \etiz}

as the respective extrema of $|\ytiz|$, $\hmtiz$, and $\etiz$ over all possible values of $(H_t, z, i)$.
\thm~\ref{thm:block_clt} below extends \thms~\ref{thm:aipw}--\ref{thm:aipw_cs} to \br, and justifies large-sample Wald-type inference of  the renewed $\tc$ using the renewed $(\hta, \hva, \sa)$. 

\begin{condition}\label{cond:block_clt}
As $T\to\infty$, 
\begine[(i)]
\item $\lmai  \cdot  \fn{\cva -\vba} = o(1)$; 
\item $\lmai \cdot T^{-1}\max_{\ott }\{\rhot(\lt+\mt) /\uet\}^2 =\oo$. 
\ende 
\end{condition}

\begin{condition}\label{cond:block_cs}
As $T\to\infty$, 
 $ \lmap^{-2}\cdot  T^{-2}\sumt \{\rhot(\lt +\mt)\} ^4/\uet ^3 = o(1)$.
\end{condition}

\begin{theorem}\label{thm:block_clt}
\prebar. For \allcf, 
\begine[(i)]
\item\label{item:block_mean_var}
 $\E(\hta) =\tc$ and $\cov(\hta) = T^{-1} C\va\ct $. 
\item\label{item:block_clt} If \cond~\ref{cond:block_clt} holds, then $( C\va \ct)^{-1/2}\cdot\sqrt{T} 
(\hat{\tau}_{\aipw, C} -\tau_C)
\converged\snrc $.
\item\label{item:block_cs}
If Conditions~\ref{cond:block_clt}--\ref{cond:block_cs} hold,  
then\smallskip\\
\cl{$\liminf_{T\to\infty}\pr (\tau_C\in\sa)
\ge 1 -\alpha$, \quad  $\liminf_{T\to\infty}\pr (\tau_C\in\tsa)
\ge 1 -\alpha$}
for all $\alpha\in (0, 1)$. 
\ende
\end{theorem}

Theorem~\ref{thm:block_clt}\eqref{item:block_mean_var}--\eqref{item:block_clt} extend \thm~\ref{thm:aipw} and imply that $\hta$ is an unbiased, consistent, and asymptotically normal estimator of $\tc$, with $\cov(\hta) = T^{-1}C\va\ct$.
Theorem~\ref{thm:block_clt}\eqref{item:block_cs} extends \thm~\ref{thm:aipw_cs} and establishes the coverage guarantee of $\sa$ and $\tsa$.  

\begin{remark}[Adaptive covariate adjustment for nonadaptive block designs]
Parallel to our recommendation in \sec~\ref{sec:adaptive analysis}, the \aae\ in \eqref{eq:estimators_block} also applies to data from nonadaptive block randomization and offers several advantages over standard nonadaptive covariate adjustment methods.
\end{remark} 

\subsection{An alternative covariance estimator}\label{sec:block_covb}
Recall from \prop~\ref{prop:aipw_cov_est} that under the \ur, 
$C\hva\ct$ has an upward bias of $CS\ct$ in estimating $C\va\ct$, which vanishes if and only if $C\myt$ is constant across $\ott$.
Under \br,
let $\myti = (Y_{ti}(1), \ldots, Y_{ti}(K))^\T$ denote the potential outcomes vector for unit $ti$, and let $\bmyt = \meani \myti = (\byt(1), \ldots, \byt(K))^\T$  denote the group average.
The discussion following \eqref{eq:hyaz_block} implies that $\rhot\bmyt$ is the analog of $\myt$ under \br, so the renewed $C\hva\ct$ is unbiased if and only if $C(\rhot\bmyt)$ is constant across $\ott$. 

Note that $C(\rhot\bmyt) = \bar n^{-1} \sumi C\my_{ti} =  \bar n^{-1} \sumi \tau_{C,ti}$,
where $\tau_{C,ti} = C \my_{ti}$ denotes the individual analog of $\tc$. Therefore, $C\hva \ct$ is unbiased if and only if the group {\it total} of $\tau_{C,ti}$, $\sumi \tau_{C,ti}$, is constant across $\ott$.
When it is more plausible that the group {\it average} of $\tau_{C,ti}$, denoted by $\bar \tau_{C,t}= \meani \tau_{C,ti}$, is constant across $t$, we propose using
\cls{
\hvab = \ds\sumt b_t (\hmyta - \hmya)(\hmyta - \hmya)^\T,}
where 
\cls{
  b_t = T \cdot \dfrac{\pt^2/(1-2\pt)}{1 +  \sums \ps^2/(1-2\ps)}\quad\text{with $\pt = \nt/N$},}

as an alternative estimator of $\va$, motivated by variance estimation for stratified experiments in \cite{pashley2021insights} and \cite{ding2024first}. 
Proposition \ref{prop:hvab} below establishes the unbiasedness of $C\hvab \ct$  for $C\va\ct$ when $\bar \tau_{C,t}$ is constant across $t$. 
The same intuition extends to $\tvab =  \sumt b_t (\tmyta - \tmya)(\tmyta - \tmya)^\T$, a variant of $\tva$. We provide the details in the \sm.

\begin{proposition}\label{prop:hvab}
\prebar. If the proportion of units within any group is below $1/2$, i.e., $\pt = \nt / N < 1/2$ for all $\ott$, then the following holds:
\begine[(i)]
\item\label{item:hva_bt} $b_t >0$ for $\ott$
\item \label{item:hva_i} $\hvab= \hva$ when $n_t$ is constant across $\ott$. 
\item\label{item:hva_ii}  $\E(\hvab) = T\cov(\hmya) + \sumt b_t(\bmyt - \bmy)(\bmyt - \bmy)^\T$.
\item\label{item:hva_iii}  For \allc, 
$\E(C\hvab\ct) - T\cov(\hta)\geq 0$, 
where the equality holds if and only if $\bar \tau_{C,t}=\meani \tau_{C,ti}$ is constant across $\ott$.
\ende 
\end{proposition}

\section{Simulation}\label{sec:simulation}
\subsection{Adaptive covariate adjustment}\label{sec:simu_adaptive covariate}
We now illustrate the advantages of {\aca} introduced in \sec~\ref{sec:adaptive analysis} through simulation. 
Assume $K = 2$ treatment levels, labeled $z = 1, 2$, and $T =$ 2,000 units, indexed by $t = \ot{T}$. 
\def\iidsample{We generate potential outcomes and covariates $\{\yto, \yt(2), X_t\}$ as i.i.d. samples from the following model:}
\iidsample
\cll{
X \sim \mathcal N(0, 1), \quad Y(1) = 1+ 2X + \epsilon_{1}, \quad Y(2) = 1 + 4X + \epsilon_{2},}

where $\epsilon_1, \epsilon_2$ are independent $\mn(0,1)$, and assign treatments via Bernoulli randomization with $\pr(\zt=1) = \pr(\zt =2)= 0.5$.
The observed outcome is $\yt = \I(Z_t=1) \yto + \I(Z_t = 2) \yt(2)$ for unit $t$. 
The goal is to estimate the average treatment effect $\tau =\bar Y(2) -\bar Y(1)$. 

We estimate $\bar Y(z)$ using the following five strategies, indexed by $* = \mtds$, and estimate $\tau$ as $\hts = \hys(2) - \hys(1)$, where $\hys(z)$ denotes the estimator of $\byz$ by Strategy $*$.
Strategies ``all'' and ``cf'' represent methods for covariate adjustment that estimate the outcome model using all units and via \cff, respectively. 
\begine
\item[``ipw'':] Construct $\hyiz$ as the IPW estimator defined in \eqref{eq:hyi}.


\item[``sm'':]  Construct $\hy_\spm(z)$ as the sample mean of observed outcomes in treatment arm $z$. 

\item[``aipw'':] Construct  $\hyaz$ as the \aae\ defined in \eqref{eq:hya}.

\item[``all'':] Construct $\hy_\all(z)$ as $\hy_\all(z) =   N_z^{-1}\sum_{t: Z_t = z}\{Y_t -\hmt^\all(z)\} +\meant\hmt^\all(z)$, where $\hmt^\all(z)$ is an estimator of $\ytz$ constructed using all $T$ units. 
 
\item[``cf'':] Construct $\hy_\cf(z)$ as follows: 
(i) Split the units into $G$ folds, indexed by $g = \ot{G}$.
(ii) For each fold $g$, estimate the fold average of $\ytz$, denoted by $\hygz$, using a variant of Strategy ``all'' with $\hmtz$ estimated using units from the other folds. (iii) Construct $\hy_\cf(z)$ as a weighted average of $\{\hygz\}_{g=1}^G$. See \sec~\ref{sec:add_res} of the \sm\ for details.

\ende
We implement the \oreg\ in Strategies ``aipw'', ``all'', and ``cf'' using both ordinary least squares (OLS) regression and random forest regression of $Y_t$ on $X_t$ by treatment level. 
For $\htau_\aipw$ from Strategy ``aipw'', we construct confidence intervals using both $\hva$ and $\tva$ as defined in \sec~\ref{sec:aipw_var_est} as the variance estimator. 
For $\htau_\cf$ from Strategy ``cf'', we set $G = 2$ and construct confidence intervals with and without \bc. The variant without \bc\ may lack theoretical guarantee.

\def\violin{the distributions of $\htau_* - \tau$ for each strategy $* = \mtds$ across 1,000 independent replications}
\def\tbl{squared errors of $\hts$ ($* = \mtds$), along with the coverage rates and average lengths of the corresponding 95\% confidence intervals}
\def\tbls{squared errors of $\hts$, along with the coverage rates and average lengths of the corresponding 95\% confidence intervals}

Figure \ref{fig:simu_bias} shows \violin.
All distributions are approximately normal and centered at zero, coherent with the consistency and asymptotic normality of the five point estimators.
%
Table \ref{tb:simu_ci} reports the mean \tbls. 
Key observations are threefold: 
\begine[(i)]
\item Overall,  \aca\ (aipw1, aipw2) achieves valid coverage under both OLS and random forest regressions, and yields shorter confidence intervals than  unadjusted inferences (ipw, sm) and cross-fitting with \bc\ (cf1). 
\item Compared with \naa\ that uses all units for outcome regression (all), \aca\ with improved variance estimator (aipw2)\\
- performs comparably with OLS regression, where the outcome model is correctly specified so that $\htau_\all$ is asymptotically most efficient by standard theory;\\
- ensures valid inference with random forest regression, whereas Strategy ``all'' leads to substantial undercoverage (0.687). 
\item The precision of aipw2 with improved variance estimator is comparable to that of \cff\ without \bc\ (cf2), which may lack theoretical guarantee.  
\ende

%
%
\begin{figure}[t!]
\begin{center}
\includegraphics[width = .5\textwidth]{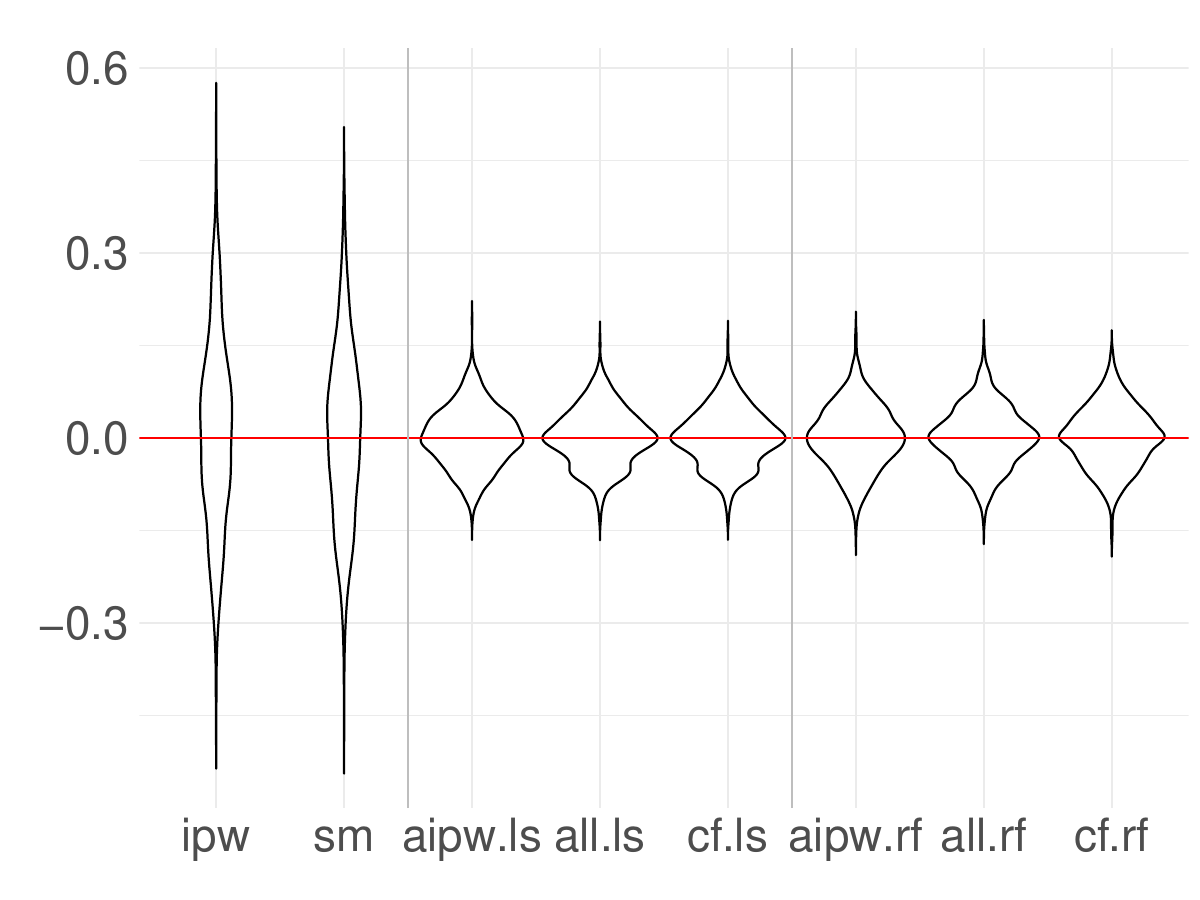}
\end{center}
\caption{\label{fig:simu_bias}Violin plots of \violin. 
The suffixes ``.ls'' and ``.rf'' indicate \oreg\ using OLS and random forest, respectively.}
\end{figure} 

%
%
\begin{table}[ht!]
\caption{\label{tb:simu_ci} 
Mean \tbl. 
``aipw1'' and ``aipw2'' refer to two variants of inference based on $\htau_\aipw$ from Strategy ``aipw'', using $\hva$ and $\tva$ to estimate $\var(\htau_\aipw)$, respectively.
``cf1'' and ``cf2'' refer to inference based on $\htau_\cf$ from Strategy ``cf'', with and without \bc, respectively.
}
 
\centering 

\resizebox{1.16\textwidth}{!}{
\begin{tabular}{l|rr|ccccc|ccccc}\hline
&&& \multicolumn{5}{c|}{ OLS} &\multicolumn{5}{c}{  Random forest}\\
\hline
 & ipw & sm & aipw1 & aipw2 & all & cf1 & cf2 & aipw1 & aipw2 & all & cf1 & cf2 \\ 
\hline
\avesqe & 0.023 & 0.020 & 0.002 & 0.002 & 0.002 & 0.002 & 0.002 & 0.003 & 0.003 & 0.002 & 0.002 & 0.002 \\\hline 
  \coverage & 0.963 & 0.962 & 0.996 & 0.945 & 0.946 & 0.998 & 0.948 & 0.994 & 0.965 & 0.687 & 1.000 & 0.965 \\ \hline
  \aveci & 0.607 & 0.581 & 0.254 & 0.185 & 0.175 & 0.284 & 0.175 & 0.266 & 0.214 & 0.099 & 0.330 & 0.204 \\ 
   \hline
\end{tabular}
}
\end{table}

\subsection{Application to \sr}

We now conduct a simulation study of \sr. 
Assume that units arrive sequentially in blocks of size $8$, with a total of $T=200$ blocks, indexed by $t = \ot{T}$. 
Within each block, we assign half of the units to treatment and the other half to control. 
Following \citet{zhou2018sequential}, 
when randomizing units in the $t$-th block, we measure the covariate imbalance of each possible treatment allocation by the Mahalanobis distance between the covariate means of the treated and control groups across all units in the first $t$ blocks. 
We then implement the best-choice rerandomization within each block \citep{WANG2025106049}, which randomly selects from the best $7$ treatment assignments with the smallest Mahalanobis distances out of the total $\binom{8}{4}=70$ possible assignments. 
In cases where some units receive only treatment or only control among the best 7 assignments, we iteratively expand the candidate set by one additional assignment at a time, each with the smallest Mahalanobis distance among the remaining, until every unit has a positive probability of receiving both treatment and control.

We generate potential outcomes and covariates as i.i.d.~samples from the following model: 
\cll{X \sim \mathcal{N}(0,1), 
    \quad Y(1) = Y(0) = 5X + \epsilon, 
    \quad 
    \text{where } \  \epsilon \indep X \ \text{ and }  \ \epsilon \sim \mathcal{N}(0,1).}
Once generated, they will be kept fixed, mimicking the finite population inference. 
We consider the sequential rerandomization design (SRD) described above and the classical completely randomized design (CRD), both of which assign half of the units to treatment and the remaining to control, and simulate $10^3$ treatment assignments from each design. 

Figure \ref{fig:seqre} shows the histogram of the difference-in-means estimator under the CRD, 
and that of the IPW estimator under the SRD, along with Gaussian approximations based on their sample averages and variances. 
From Figure \ref{fig:seqre}, the estimators under both designs are approximately unbiased and Gaussian distributed, coherent with the theory for these two designs in \citet{neyman1923on} and Section \ref{sec:block_ad}. 
Moreover, the estimator under the SRD is more precise than that under the CRD. 
Specifically, the root mean squared errors (RMSEs) for the two estimators under SRD and CRD are $0.077$ and $0.262$, respectively, indicating a $70.6\%$ reduction in RMSE under SRD.

We then investigate the performance of our confidence intervals from Section \ref{sec:block_ad}. 
Across the $10^3$ simulated assignments under the SRD, 
the $95\%$ confidence interval has $94.8\%$ coverage rate, with an average length of $0.296$. 
In contrast, \citet{neyman1923on}'s $95\%$ confidence interval under the CRD has $95.5\%$ coverage rate, with an average length of $1.03$.  
Both coverage rates are close to the nominal level, while confidence interval under the SRD is on average much shorter, with a $71.3\%$ reduction in average length. 

\begin{figure}[t!]
\begin{center}
\includegraphics[width = 0.55\textwidth]{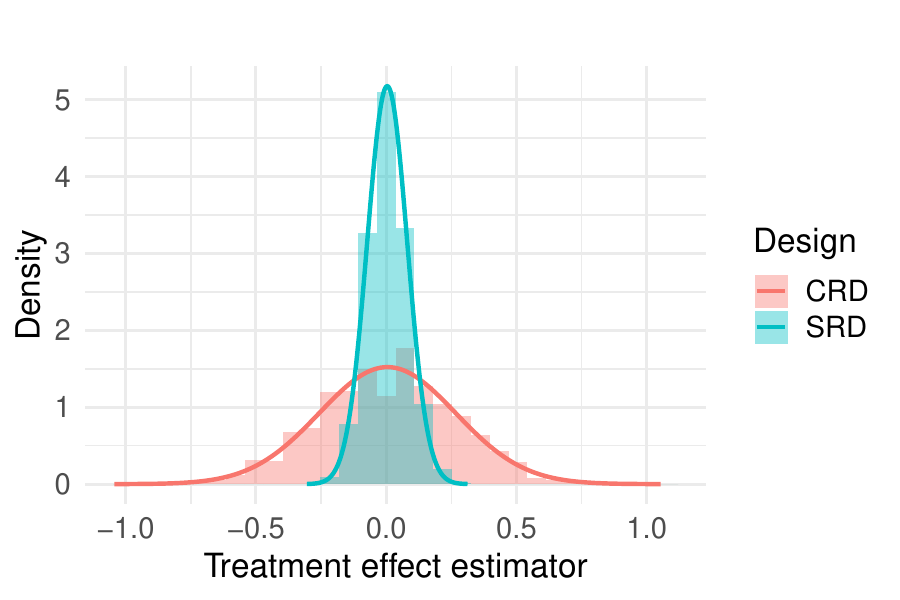}
\end{center}
\caption{Histograms of the difference-in-means estimator and the IPW estimator under the completely randomized design and the sequential rerandomization design across 1,000 independent replications, along with densities from their Gaussian approximations.}
\label{fig:seqre}
\end{figure} 

\subsection{Influence of vanishing $\uet$ on inference}\label{sec:simu_vanishing_etz}
We illustrate below the influence of $\uet$ on inference; c.f. the comments after \prop~\ref{prop:ipw_sc} and Remark \ref{rmk:rate}. 
Consider a two-armed bandit with $T = 200$ units, indexed by $t = \ot{200}$, and arms labeled as $z = 1, 2$. 
We generate potential outcomes as $\yto \sim \mn(t, 1)$ and $Y_t(2) \sim \mn(2t,1)$ for $t = \ot{200}$, and assign treatment levels by the following rule:
\begine[(i)]
\item For $t = \ot{50}$, randomly assign unit $t$ to arm 1 or 2 with equal probability $0.5$. 
\item For $t = 51, \ldots, 200$, assign unit $t$ to the arm with higher sample mean of the observed outcomes from units 1 to $t-1$ with probability $1 - t^{-1/2 + \delta}$.  
\ende 
This assignment mechanism implies 
$\uet =  \min_{H_t,\, z \in \{1,2\}}\etz = t^{-1/2 + \delta}$, 
which vanishes asymptotically at rate $-1/2 + \delta$.  

Figure \ref{fig:vanishing_etz} illustrates the performance of the IPW estimator $\hti$ for estimating the average treatment effect $\tau = \by(2) - \by(1)$ at $\delta = 0.2, 0, -0.2, -0.4$. 
As $\delta$ decreases below 0, the distribution of $\hti-\tau$ becomes increasingly skewed, deviating from normality.  
The table below reports the corresponding empirical bias, sample skewness, and mean squared error of the point estimator, along with the coverage rate and average length of the 95\% confidence interval.  
Coherent with our theory and the discussion following \prop~\ref{prop:ipw_sc}, inference based on $\hti$ and normal approximation is valid when $\delta \geq 0$, but gradually fails to ensure coverage as $\delta$ decreases below 0. 

\begin{figure}[ht!]
\begin{center}
\includegraphics[width = .44\textwidth]{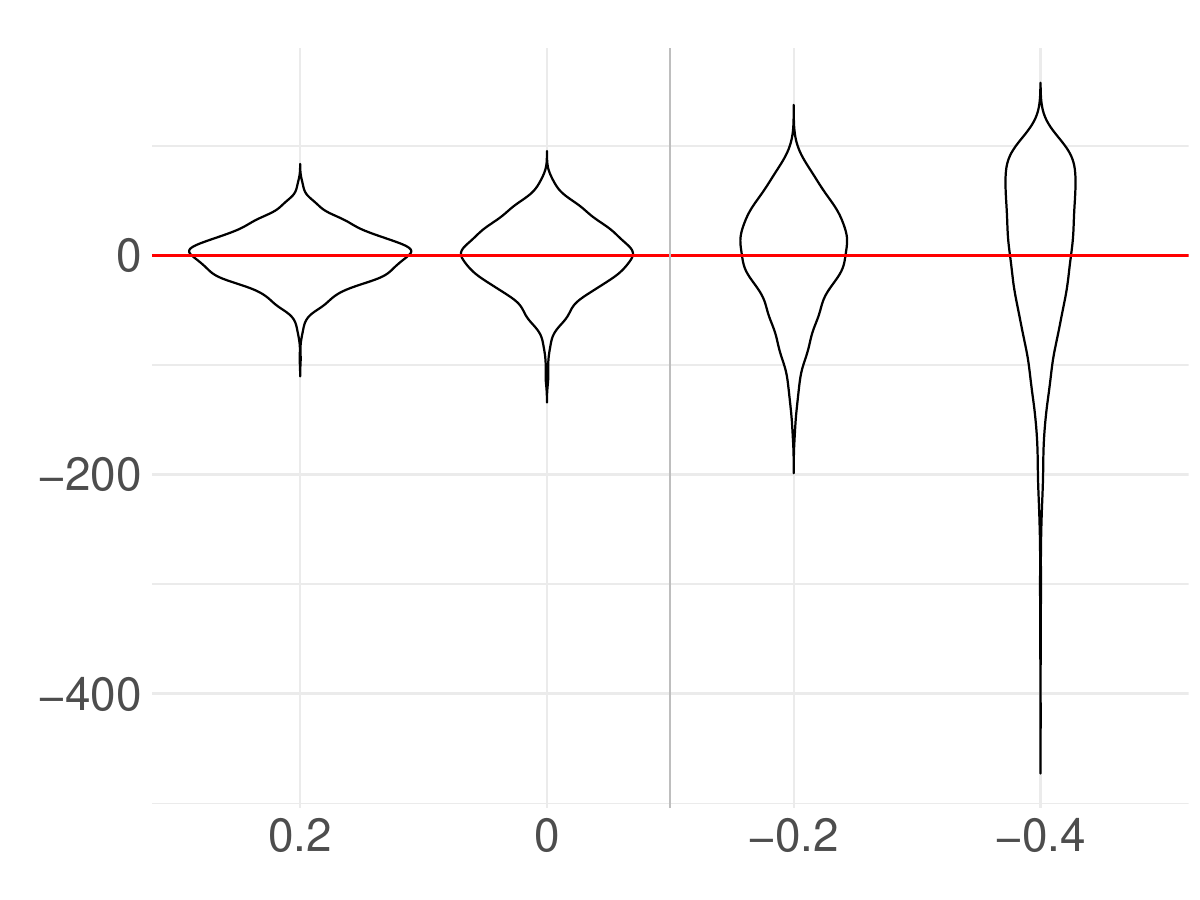}
\end{center}

\centering
\begin{tabular}{l|rrrr}\hline
& \multicolumn{4}{c}{$\delta \ (\uet)$}\\\cline{2-5}
 & \multicolumn{1}{c}{0.2} & \multicolumn{1}{c}{0} &  \multicolumn{1}{c}{$-0.2$}    &  \multicolumn{1}{c}{$-0.4$}    \\ 
  & ($\uet = 1/t^{0.3}$) &  ($\uet = 1/t^{0.5}$) &  ($\uet = 1/t^{0.7}$) &  ($\uet = 1/t^{0.9}$) \\\hline
Empirical bias & $-$1.051 & $-$1.020 & $-$2.282 & 0.342 \\ 
   \hline Sample skewness & $-$0.112 & $-$0.312 & $-$0.517 & $-$1.096 \\ 
   \hline Mean squared error & 575.447 & 957.284 & 2486.650 & 6088.563 \\ 
   \hline Coverage rate & 0.949 & 0.945 & 0.886 & 0.754 \\ 
   \hline Average CI length & 95.909 & 126.777 & 186.430 & 263.494 \\ 
   \hline
\end{tabular}
\caption{\label{fig:vanishing_etz}Violin plots of the distributions of $\hti - \tau$ across 1,000 independent replications at $\delta = 0.2, 0, -0.2, -0.4$.}
\end{figure}

\section{Discussion}\label{sec:discussion}
We establish the finite-population, {\db} theory for causal inference under adaptive designs that accommodate nonexchangeable units, both nonconverging and vanishing treatment probabilities, and nonconverging, possibly black-box, outcome models.
Our framework encompasses widely used designs such as \mab\ algorithms, \car, and \sr. 
As an application, we proposed the \aca\ approach for analyzing even nonadaptive experiments, which offers multiple advantages over standard \naa\ methods.

The \ars\ proposed by \cite{hadad2021confidence} and \cite{bibaut2021post} reweights the units based on their propensity scores, and may yield biased and inconsistent estimates of \fp\ average treatment effects when units are nonexchangeable. 
An alternative is to reweight units based on covariates to better align the study population with a target population of interest. 
Our framework accommodates this setting.

\bibliographystyle{chicago}
\bibliography{refs_adaptive_rand.bib}

\newpage
\setcounter{equation}{0}
\setcounter{section}{0}
\setcounter{figure}{0}
\setcounter{example}{0}
\setcounter{proposition}{0}
\setcounter{corollary}{0}
\setcounter{theorem}{0}
\setcounter{table}{0}
\setcounter{condition}{0}
\setcounter{lemma}{0}
\setcounter{remark}{0}

\renewcommand {\theproposition} {S\arabic{proposition}}
\renewcommand {\theexample} {S\arabic{example}}
\renewcommand {\thefigure} {S\arabic{figure}}
\renewcommand {\thetable} {S\arabic{table}}
\renewcommand {\theequation} {S\arabic{equation}}
\renewcommand {\thelemma} {S\arabic{lemma}}
\renewcommand {\thesection} {S\arabic{section}}
\renewcommand {\thetheorem} {S\arabic{theorem}}
\renewcommand {\thecorollary} {S\arabic{corollary}}
\renewcommand {\thecondition} {S\arabic{condition}}
\renewcommand {\thepage} {S\arabic{page}}

\setcounter{page}{1}
\spacingset{1.5}
\begin{center}
\bf\Large 
Supplementary Material 
\end{center}

\sec~\ref{sec:add_res} provides the additional results that complement the main paper. 

\sec~\ref{sec:aipw_app} provides the proofs of the results in Sections \ref{sec:ipw}--\ref{sec:aipw} of the main paper. 

\sec~\ref{sec:block_app} provides the proofs of the results in Section  \ref{sec:block_ad} of the main paper. 

\sec~\ref{sec:proof_sc_app} provides the proofs of the results in \sec~\ref{sec:add_res}. 

\paragraph*{Notation and useful facts.} 	
Let $\I(\cdot)$ denote the indicator function.
Let $\mn(\cdot, \cdot)$ denote the normal distribution.   
For a positive integer $m$, let $0_m$ denote the $m\times 1$ zero vector, $I_m$ the $m\times m$ identity matrix, and $\chi^2_{m, 1-\alpha}$ the $(1-\alpha)$th quantile of the chi-squared distribution with $m$ degrees of freedom.
We omit the subscript $m$ in $0_m$ and $I_m$ when the dimension is clear from context.
%
For a collection of numbers $\{a_z \in \mathbb R: z\in\mz\}$, where $\mz$ is the index set, let $\diag(a_z)_{z\in\mz}$ denote the diagonal matrix with $a_z$ on the diagonal.
For a set of vectors $\{a_i \in \mathbb R^m: i = \ot{n}\}$, 
define its {\it sample covariance matrix} as $(n-1)^{-1}\sum_{i=1}^n (a_i - \bar a)(a_i - \bar a)^\T$, where $\bar a = n^{-1}\sum_{i=1}^n a_i$.

For an $m\times n$ matrix $A = (a_{ij})_{m\times n}$, let $\fn{A} = (\sum_{i=1}^m\sum_{j=1}^n a_{ij}^2)^{1/2}$ denote the {\it Frobenius norm} of $A$, and $\tn{A}$ denote the {\it spectral norm} of $A$. 
For an $n\times n$ square matrix $B$, let $\lambda_{\max}(B)$ and $\lm(B)$ denote the largest and smallest eigenvalues of $B$, respectively.
Standard results imply that the Frobenius and spectral norms are both sub-multiplicative with 
\beginy\label{eq:l2}
\begin{array}{c}
1\leq\fn{A}/\tn{A}\leq\sqrt{\rank{A}},\medskip\\
\tn{A}^2 =\tn{A^\T A} =\tn{A A^\T} =\lambda_{\max}(A^\T A) =\lambda_{\max}( AA^\T).
\end{array}
\endy


Assume the probability measure induced by the treatment assignment mechanism throughout.
Let $\converged$ denote convergence in distribution, and let $\op$ denote a sequence converging to 0 in probability.

\section{Additional results}\label{sec:add_res}	
\subsection{Covariate adjustment using cross-fitting in \sec~\ref{sec:simu_adaptive covariate}}
Recall that $\hy_{\cf}(z)$ denotes the estimator of $\byz$ using \cff\ in Strategy ``cf'' in \sec~\ref{sec:simu_adaptive covariate}. We construct it in the following steps:
\begine[-]
\item  Split the units into $G$ folds, indexed by $g = \ot{G}$. Let $\mtg$ denote the index set of fold $g$. 

\item For fold $g$, 
estimate the fold  average potential outcome $\bar Y_{[g]}(z) = |\mathcal T_g|^{-1}\sum_{t\in \mathcal T_g} Y_t(z)$ using 
\cll{
\hat Y_{[g]}(z) = |\mtgz|^{-1}\sum_{t\in\mtgz}\{Y_t -\hmt^\cf(z)\} + |\mtg|^{-1}\sum_{t\in\mtg}\hmt^\cf(z), 
}

where $\mtgz = \{t\in\mtg: Z_t = z\}$ denotes the index set of units receiving treatment $z$ in fold $g$ and
$\hmt^\cf(z)$ is an estimator of $\ytz$  based on units in the other $G-1$ folds.

\item Estimate $\byz$ by \cls{
\hy_{\cf}(z) =\sum_{g=1}^G w_g\hat Y_{[g]}(z),}

where $w_g = |\mtg|/T$ denotes the proportion of units in fold $g$. This weighting scheme is motivated by \cls{
\byz = \meant \ytz = T^{-1}\sum_{g=1}^G \left\{ \sum_{t\in \mtg} \ytz\right\} =  \sum_{g=1}^G w_g   \bar Y_{[g]}(z).}
\ende

\subsection{Possible biases of adaptive reweighting}\label{sec:ex_bias}
\cite{hadad2021confidence} and \cite{bibaut2021post} proposed an {\it \arw} strategy for causal inference under adaptive designs with exchangeable units. 
We illustrate below the possible bias of this strategy when applied to nonexchangeable units. 

Consider a two-armed bandit with $T = 200$ units, indexed by $t = \ot{200}$, and arms labeled as $z = 1, 2$. 
We generate potential outcomes as $\yto \sim \mn(t, 1)$ and $Y_t(2) \sim \mn(2t,1)$ for $t = \ot{200}$, and assign treatment levels using the following $\epsilon$-greedy rule:
\begine[(i)]
\item For $t = \ot{50}$, assign unit $t$ to arm 1 or 2 with equal probability $0.5$; 
\item For $t = 51, \ldots, 200$, assign unit $t$ to the arm with higher sample mean of observed outcomes among units 1 to $t-1$ with probability $0.8$.
\ende  

We estimate $\by(1)$, $\by(2)$, and $\tau = \by(2) - \by(1)$ using the IPW estimators defined in \eqref{eq:hyi} and the \ares\ proposed by \cite{hadad2021confidence}.
Figure \ref{fig:ex_bias} shows the distributions of their deviations from the true values across 1,000 independent replications, with empirical biases summarized in the table below.
The IPW estimators in \eqref{eq:hyi} appear empirically unbiased, whereas the \ares\ show noticeable biases.

To see why, note that the data-generating process implies upward trends in $\yto$, $\yt(2)$, and the individual treatment effect $\tau_t = Y_t(2) - Y_t(1)$, with arm 2 being the ``better'' arm.  
Accordingly, after time $t = 50$, most units are assigned to arm 2 with probability $0.8$, which is higher than the initial $0.5$. 
The propensity-score based reweighting strategy in \cite{hadad2021confidence} thus assigns greater weights to later units in arm 2, leading to an overestimation of $\by(2)$.
Conversely, after time $t = 50$, most units are assigned to arm 1 with probability $0.2$, lower than the initial $0.5$. 
The reweighting strategy in \cite{hadad2021confidence} thus assigns greater weights to earlier units in arm 1, leading to an underestimation of $\by(1)$. 
Together, these two biases lead to an overestimation of the average treatment effect $\tau$.

\begin{figure}[t!]
\begin{center}
\includegraphics[width = .48\textwidth]{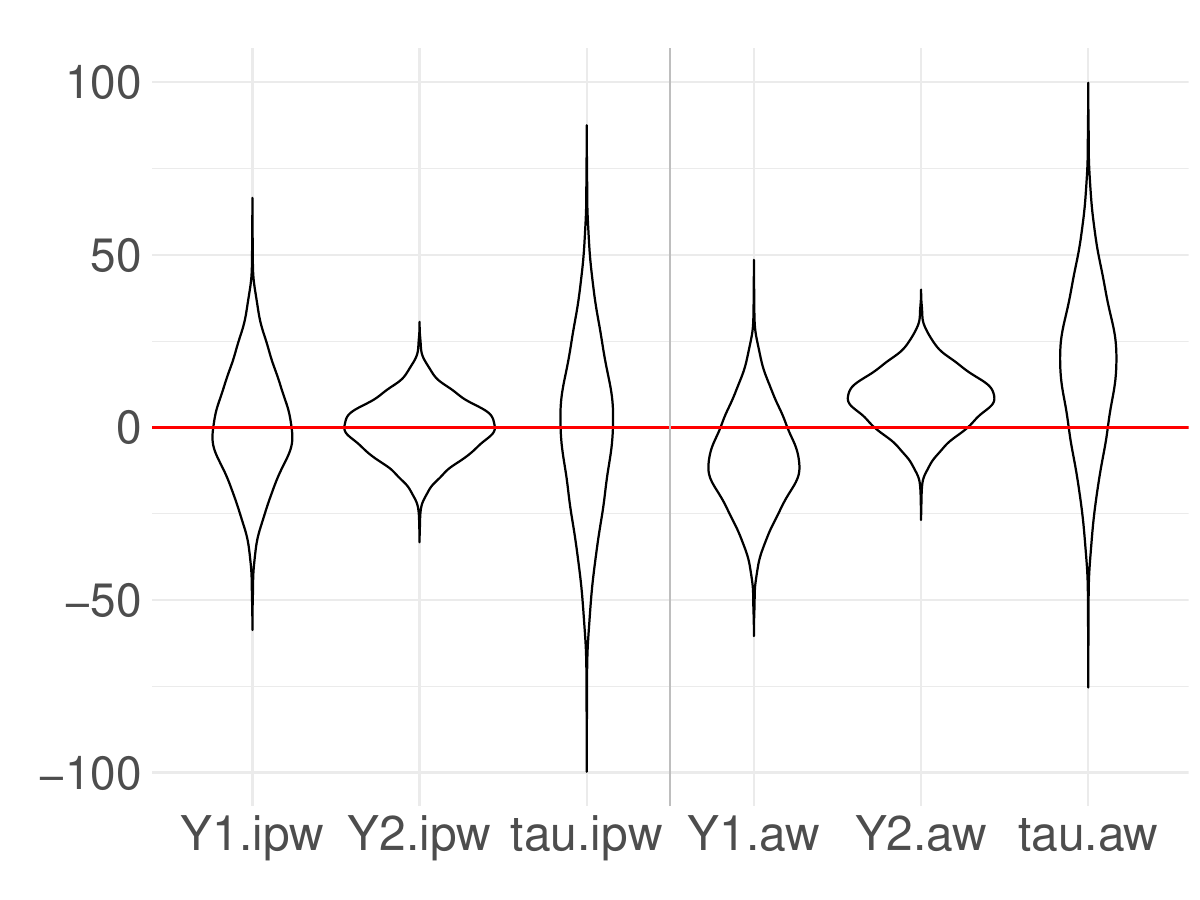}

\medskip 

\begin{tabular}{l|rrr|rrr}
  \hline
 & Y1.ipw & Y2.ipw & tau.ipw & Y1.aw & Y2.aw & tau.aw \\ 
  \hline
Empirical bias & 0.543 & $-$0.286 & $-$0.830 & $-$9.034 & 7.991 & 17.025 \\ 
   \hline
\end{tabular}
\end{center}
\caption{\label{fig:ex_bias}
Violin plots of the distributions of deviations of the estimators from the true values $(\by(1), \by(2), \tau)$ over 1,000 independent replications. The suffix ``.ipw'' refers to the standard IPW estimators in \eqref{eq:hyi}. The suffix ``.aw'' refers to the \ares\ proposed by \cite{hadad2021confidence}.}
\end{figure}  

\subsection{Sufficient conditions for \cond~\ref{cond:aipw_clt}}\label{sec:sc_app}
Assume the adaptive randomization in Definition~\ref{def:ad}. \prop~\ref{prop:aipw_sc_app} below generalizes the sufficient conditions in \prop~\ref{prop:aipw_sc} by allowing for vanishing $\etz$.
Recall that 
\begina
\uet =\min_{H_t,\, \ziz}\etz, 
\quad
v_t =\maxz\var\{\etzi\}, 
\quad \omt =\maxz\var\{\hmtz\}.
\enda

\begin{proposition}\label{prop:aipw_sc_app}
\prear. 
If as $T\to\infty$, (a) $\lm(\va)$ is uniformly bounded away from 0, and (b) both $\ytz $ and $ \hmtz$ are uniformly bounded, then \cond~\ref{cond:aipw_clt}\eqref{item:clt_aipw_var_conv} is ensured by $\fn{\cva - \va} = \op$ and holds if any of the following conditions is satisfied:
\begine[(i)]
\item\label{item:V_i_main} \underline{\bf Vanishing $\var\{\etzi\}$ and $\var\{\hmtz\}$:}

 $T^{-1}\sumt v_t = o(1)$; \quad $T^{-1}\sumt\omt\cdot\uet^{-2} = o(1)$.
 
 
\item\label{item:V_ii_main}
\underline{\bf Vanishing $\var\{\etzi\}$ and \lrd\ of $\hmtz$:}

$\meant (\maxs\ue_s^{-1})\cdot \sqrt{v_t} =\oo$; 

\prelimdeps
\begine[-]
\item $(\meant\uet^{-1})\cdot (\maxt\omt\cdot\uet^{-2}) = o(T^{1-\beta})$; 
\item $(\meant\uet^{-1} )\cdot (\maxt v_t) = o(T^{1-\beta})$; 
\item \limdept, $\{\hmtz:\ziz\}$ \limdep.
\ende

\item\label{item:V_iii_main} 
\underline{\bf Vanishing $\var\{\hmtz\}$ and \lrd\ of $\etz$:}

$T^{-1}\sumt\omt = o(1)$; 
\quad 
$\meant (\maxs\ue_s^{-1}) \cdot\sqrt{\omt}\cdot\uet^{-1} = o(1)$;

\prelimdeps 
\begine[-]
\item $\maxt\omega_t \cdot\uet^{-2} = o(T^{1-\beta})$; 
\item $\maxt v_t = o(T^{1-\beta})$;
\item \limdept, $\{\etz:\ziz\}$ \limdep.
\ende


\item\label{item:V_iv_main} 
\underline{\bf {\llrd} of $\etz$ and $\hmtz$:} \prelimdepsc 
\begine[-]
\item $\maxt\omega_t \cdot\uet^{-2} = o(T^{1-\beta})$;
\item $\maxt v_t = o(T^{1-\beta})$; 
\item \limdept, $\{\etz,\hmtz:\ziz\}$ \limdep. 
\ende
\ende 
\end{proposition} 

\prop~\ref{prop:aipw_sc_app} clarifies the rate at which $\etz$ can diminish to 0 under \cond~\ref{cond:aipw_clt}.
If further $\etz$ is uniformly bounded away from 0, then  the four conditions in \prop~\ref{prop:aipw_sc_app} are guaranteed by the corresponding conditions in \prop~\ref{prop:aipw_sc}.

\section{Proofs of the results in Sections \ref{sec:ipw}--\ref{sec:aipw}}\label{sec:aipw_app}
We provide below the proofs of the results on the augmented \ipwd\ (\aipwc) estimator in \sec~\ref{sec:aipw} of the main paper. 
The results on the \ipwd\ (\ipwc) estimator in \sec~\ref{sec:ipw} are a special case with $\hmtz = 0$ for all $t$ and $z$. 
 
Assume the adaptive randomization in Definition~\ref{def:ad} throughout this section with $\etz =\pr(Z_t = z\mid\hht)$. 
Let 
\beginy\label{eq:wtz}
\rtz =\dfrac{1}{\etz } - 1, \quad 
\wtz =\dfrac{\indz }{\etz } - 1=\left\{ 
\begin{array}{cll}
 \rtz &\ \text{if $Z_t = z$}\\
-1 &\ \text{if $Z_t\neq z$} 
\end{array}\right.
\endy
with 
\beginy\label{eq:rtz_1}
\begin{array}{c}
|\wtz|, \ \rtz, \
 |\rtz - 1| \le\max\{r_t(z), 1\}\le\etzi\le\ueti,\medskip\\
%
%
\begin{array}{rcl}
\rtz^2\cdot\etz +1 -\etz&=&\rtz^2\cdot\etz +\rtz\cdot\etz \\
&=& \rtz \cdot \etz \cdot \left\{\rtz + 1\right\} =\rtz. 
\end{array}
\end{array}
\endy
Recall that $\atz =\ytz-\hmtz$ denotes the residual from the adaptive \oreg. Let 
\beginy\label{eq:dtz}
\dtz &=& \tsrif\left\{\hytaz -\ytz \right\} = \tsrif\wtz\cdot\atz, 
\endy
where the second equality follows from 
\begina
\hytaz &\overset{\eqref{eq:hytiz}+\eqref{eq:hytaz}}{=}& \dfrac{\indz }{\etz } Y_t +\left\{ 1 -\dfrac{\indz }{\etz }\right\}\hmtz\nnb\\
&\overset{\eqref{eq:wtz}}{=}& \{\wtz +1\}\ytz -\wtz\hmtz\nnb\\
&=&\ytz+\wtz\cdot\atz. 
\enda
Let 
\beginy\label{eq:dtz_def}
\begin{array}{rcl}
\mdt &=& (D_t(1),\ldots, D_t(K))^\T =
\tsrif\left(\hmyta -\myt\right),\medskip\\
\bmd  &=& T^{-1}\displaystyle\sumt\mdt  = (\bar{D}(1),\ldots,\bar{D}(K))^\T, \quad \text{where} \ \  \bar{D}(z) = \ti\displaystyle\sumt\dtz,   
\end{array}
\endy
to write 
\beginy\label{eq:mdt} 
\hmyta -\myt =  \tsr\mdt, \quad \hmya -\bmy = \meant (\hmyta -\myt ) = \tsrif\displaystyle\sumt\mdt = \tsr\bmd.
\endy
Recall that
\begina
\cva =\diag\left[\meant\dfrac{\atzs}{\etz }\right]_{\ziz} - \meant\mat\matt, 
\enda
with
$\E(\cva)=\va$. 
Let ${\cv} _{\aipw, zz'}$ denote the $(z,z')$th element of $\cva $. 

\subsection{Lemmas}
\subsubsection{Lemmas for the finite-sample exact results}
\begin{lemma}\label{lem:mds_cov}
Let $\{\xi_{t}\in\mathbb{R}^m: 1\le t\le T\}$ be a 
square-integrable martingale difference sequence with respect to filtration $\{\mf_t:\ott\}$. 
Let 
\begina
V = \cov\left(\sumt\xi_t\right),\quad 
{\cv}  = \sumt\cov(\xi_t\mid\mf_{t-1}), \quad
\hat{V} = T (T-1)^{-1}\sumt (\xi_t -\bar{\xi}) (\xi_t -\bar{\xi})^\T,
\enda
where $\mf_0 =\emptyset$ and $\bar{\xi} = T^{-1}\sumt\xi_t$. 
For $t\neq t'\in \{\ot{T}\}$, 
\begine[(i)]
\item\label{item:lem11} $\E(\xi_t\mid\mf_{t-1}) = 0_m$, \quad $\cov(\xi_t\mid\mf_{t-1}) =\E(\xi_t\xi_t^\T\mid\mf_{t-1})$; 
\item\label{item:lem12} $
 \E(\xi_t) = 0_m$,\quad $\cov(\xi_t) =\E(\xi_t\xi_t^\T) =\E\{\cov(\xi_t\mid\mf_{t-1})\}$, \quad $
 \cov(\xi_t,\xi_{t'}) = 0_{m\times m}$; 
\item\label{item:lem13} $V 
=\ds\sumt\cov(\xi_t) 
= 
\E({\cv} ) 
=\E(\hat{V})$, where $\E(\cv) = \sumt\E\{\cov(\xi_t\mid\mf_{t-1})\}$. 
\ende  
\end{lemma}

\begin{proof}[\bf Proof of \lem~\ref{lem:mds_cov}]
\lem~\ref{lem:mds_cov}\eqref{item:lem11} follows from the definition of martingale difference sequence with $\E(\xi_t\mid\mf_{t-1}) = 0_m$.
\lem~\ref{lem:mds_cov}\eqref{item:lem12} follows from
 \beginy 
 \begin{array}{ll}
&\E(\xi_t) =\E\left\{\E(\xi_t\mid\mf_{t-1})\right\} = 0_m,\medskip\\
& \cov(\xi_t) =\E(\xi_t \xi_t^\T) =\E\left\{\E(\xi_t \xi_t^\T \mid\mf_{t-1})\right\} =\E\left\{\cov(\xi_t\mid\mf_{t-1})\right\},\medskip\\
 \text{for  $ t' < t$,}& \cov(\xi_t,\xi_{t'}) =\E(\xi_t\xi_{t'}^\T) =\E\left\{\E(\xi_t\xi_{t'}^\T\mid\mf_{t-1})\right\}
= 
\E\left\{ \E(\xi_t\mid\mf_{t-1}) \cdot \xi_{t'}^\T\right\} = 0_{m\times m}. 
\end{array}
 \label{eq:cov_xx} 
 \endy
Equation \eqref{eq:cov_xx} implies the first and second equalities regarding $V$ in \lem~\ref{lem:mds_cov}\eqref{item:lem12} as follows: 
\begin{align}\label{eq:ext_1}
V & =\cov\left(\sumt\xi_t\right)
=\sumt\cov(\xi_t) + 2\sum_{t'<t}\cov(\xi_t,\xi_{t'})
\overset{ \eqref{eq:cov_xx}}{=} \sumt\cov(\xi_t) 
\\
& \overset{ \eqref{eq:cov_xx}}{=}
\sumt\E\left\{\cov(\xi_t\mid\mf_{t-1})\right\}
= 
\E\left\{\sumt\cov(\xi_t\mid\mf_{t-1})\right\}
= 
\E({\cv} )\nnb. 
\end{align}
In addition, \eqref{eq:cov_xx}--\eqref{eq:ext_1} imply
\beginy\label{eq:ext_2}
\sumt\E(\xi_t\xi_t^\T) \overset{ \eqref{eq:cov_xx}}{=}\sumt\cov(\xi_t) \overset{ \eqref{eq:ext_1}}{=} V,\quad \E(\bar{\xi}\bar{\xi}^\T ) \overset{ \eqref{eq:cov_xx}}{=}\cov(\bar{\xi}) = T^{-2} V. 
 \endy
It then follows from $T^{-1} (T-1)\hv =\sumt (\xi_t -\bar{\xi}) (\xi_t -\bar{\xi})^\T =\sumt\xi_t\xi_t^\T - T\bar\xi\bar\xi^\T$ that 
\begina
T^{-1}(T-1) \E(\hat{V}) =
\sumt\E(\xi_t\xi_t^\T ) - T\E(\bar{\xi}\bar{\xi}^\T )\overset{ \eqref{eq:ext_2}}{=} V- T^{-1} V = T^{-1}(T-1) V,
 \enda
and therefore $
\E(\hat{V}) = V. 
$
\end{proof}

\begin{lemma}\label{lem:w_hist}
For $z \neq z'\in\mathcal{Z}$ and $\ott $,
\begina
&&\E\{\wtz\mid\hht\}=0, \quad
\cov\{\wtz\mid\hht\} =\E\{\wtz^2\mid\hht\} = 
	\rtz,\\
&&\cov\{\wtz,\wtzp\mid\hht\} =\E\{\wtz\wtzp\mid\hht\} =
	-1. 
\enda
\end{lemma}
\begin{proof}[\bf Proof of \lem~\ref{lem:w_hist}]
The expressions of $\rtz$ and $\wtz$ in \eqref{eq:wtz} implies  
\begina
\E\{\wtz\mid\hht\} &\overset{\eqref{eq:wtz}}{=}&  \rtz\cdot\pr(Z_t = z\mid H_t) + ( -1)\cdot\pr(Z_t\neq z\mid H_t) \\
&=& \rtz\cdot\etz + ( -1)\cdot\{1-\etz\} \overset{\eqref{eq:wtz}}{=} 0.
\enda
This ensures 
\begina
\cov\{\wtz\mid\hht\} &=& \E\{\wtz^2\mid\hht\} \\
&\overset{\eqref{eq:wtz}}{=}& 
\rtz^2\cdot\pr(Z_t = z\mid H_t) + ( -1)^2\cdot\pr(Z_t\neq z\mid H_t)\\
&=&
\rtz^2\cdot\etz + 1- \etz\\
&\overset{\eqref{eq:rtz_1}}{=}&  \rtz,\\
\cov\{\wtz,\wtzp\mid\hht\} 
&=&\E\{\wtz\wtzp\mid\hht\}\\
&\overset{\eqref{eq:wtz}}{=}&
		-\rtz\cdot\pr(Z_t = z\mid\hht) 
+\{-r_t(z')\}\cdot\pr(Z_t = z'\mid\hht)
+ 1\cdot\pr(Z_t\neq z, z'\mid\hht)
\\
&\overset{\eqref{eq:wtz}}{=}&
-\{\etzi -1\}\cdot\etz 
-\{e_t^{-1}(z')-1\}\cdot e_t(z')
+\{1 -\etz - e_t(z')\}
\\
&=& -1.
\enda
\end{proof}

\begin{lemma}\label{lem:D}
$\{\mdt:\ott\}$ as defined in \eqref{eq:dtz_def} is {\mart} that satisfies the following: 
\begine[(i)]
\item\label{item:lem 3 1} $\E(\mdt\mid\hht) = 0_K$,
\begina
\cov(\mdt \mid\hht) 
&=& \E(\mdt\mdt^\T\mid\hht)\\
&=& T^{-1}
\left(
\begin{array}{rrcr}
r_t(1)\at(1)^2 &-\at(1)\at(2) &\cdots & -\at(1)\at(K)\\
-\at(2)\at(1) & r_t(2)\at(2)^2 &\cdots & -\at(2)\at(K)\\
\vdots&\vdots&&\vdots\\
-\at(K)\at(1)&-\at(K)\at(2) &\cdots & r_t(K)\at(K)^2 
\end{array}
\right)\\
&=& T^{-1}\left[\diag\left\{\dfrac{\atzs}{\etz}\right\}_{\ziz} - \mat\matt\right], 
\enda
$\ds\sumt\cov(\mdt\mid\hht) = \ds\sumt\E(\mdt\mdt^\T\mid\hht) =\cva $.

\medskip

\item\label{item:lem 3 2} $\E(\mdt) = 0_K$, \quad $\cov\left(\ds\sumt\mdt\right) =\ds\sumt\cov(\mdt) =\va$.
\ende
\end{lemma}

\begin{proof}[\bf Proof of {\lemd}]
Recall from \eqref{eq:dtz} that $\dtz = \tsri\wtz\cdot\atz$, where $\atz$ is fixed given $H_t$.
This, together with \lem~\ref{lem:w_hist}, implies  
\begina
\E\{\dtz\mid\hht\} &\oeq{\eqref{eq:dtz}}& T^{-1/2}\cdot\E\{\wtz\mid\hht\}\cdot\atz  \oeqt{\lem~\ref{lem:w_hist}}  0,\\
\cov\left\{\dtz, \dtzp\mid\hht\right\}
&=& \E\left\{\dtz\dtzp\mid\hht\right\}\\
&\oeq{\eqref{eq:dtz}}&
T^{-1}\E\left\{\wtz w_t(z')\mid\hht\right\}\cdot\atz\atzp\\
&\oeqt{\lem~\ref{lem:w_hist}}& 
\begin{cases}
T^{-1}\rtz\atz^2 &\text{if } z=z'\\
- T^{-1}\atz\atzp &\text{if } z\ne z' 
\end{cases}
\enda
 for all $\ott $ and $z,z'\in\mathcal{Z}$. This implies {\lemd}\eqref{item:lem 3 1}.
{\lemd}\eqref{item:lem 3 2} then follows from \lem~\ref{lem:mds_cov}.
\end{proof}

\subsubsection{Lemmas for the central limit theorem}
\lem~\ref{lem:martingale_clt} below reviews the martingale central limit theorem from \citet[\thm~2]{brown1971martingale}; see also \citet[][Corollary 3.1]{hall2014martingale}. 

\begin{lemma}[Martingale central limit theorem]
\label{lem:martingale_clt}
Let $\{\xi_{Tt}\in\mathbb R,\mf_{Tt}: 1\le t\le T, T\ge 1\}$ be an array of square-integrable martingale difference sequences. 
Let \medskip\\
\cl{$
\cvt =\ds\sumt\E(\xi^2_{Tt}\mid\mf_{T,t-1}),\where\mf_{T0} =\emptyset; \quad V_{T} =\E({\cv} _{T}).$}

As $T\to\infty$, if
\begine[(i)]
\item (Variance convergence condition) ${\cv} _T/V_T = 1 +\op$; 
\item (Lindeberg condition) $V_T^{-1}\sumt\E\big\{\xi_{Tt}^2\cdot\I\big( |\xi_{Tt} |\geq\epsilon V_T^{-1/2}\big)\big\} =\op$ for all $\epsilon > 0$, 
\ende
then 
$
V_T^{-1/2}\sumt\xi_{Tt}\converged\mathcal{N}(0,1). 
$
\end{lemma}


\subsubsection{Lemmas for variance estimation}
 \paragraph*{A useful fact.} \eqref{eq:mdt} implies
\begina
\hmyta  -\hmya \oeq{\eqref{eq:mdt}} (\tsr\mdt + \myt) -( \tsr\bmd + \bmy) = \tsr (\mdt -\bmd ) +\myt -\bmy, 
\enda
so that 
\beginy\label{eq:cov_ss}
\begin{array}{rl}
  & \ds\sumt(\hmyta  -\hmya) (\hmyta  -\hmya)^\T \medskip\\
=& \ds\sumt\Big\{ \tsr  (\mdt -\bmd) +\myt  -\bmy\Big\}
\Big\{ \tsr  (\mdt -\bmd)^\T + (\myt  -\bmy)^\T\Big\}\medskip\\
=&  T \ds\sumt(\mdt -\bmd)(\mdt -\bmd)^\T 
+ 
\tsr  \ds\sumt (\mdt -\bmd) (\myt -\bmy)^\T  \medskip\\
& + \tsr \ds\sumt (\myt -\bmy)(\mdt -\bmd)^\T + \ds\sumt(\myt -\bmy) (\myt -\bmy)^\T \medskip\\
=& 
T \displaystyle\ds\sumt (\mdt -\bmd) (\mdt -\bmd)^\T 
+ 
\tsr\displaystyle\ds\sumt\mdt(\myt -\bmy)^\T + \tsr\displaystyle\ds\sumt(\myt -\bmy)\mdt^\T + (T-1)S.
\end{array}
\endy

\begin{lemma}[Chebyshev's inequality]\label{lem:cheb}
If 
$(X_n)$ is a stochastic sequence such that each element has finite variance, then
$X_n -\E(X_n) =\sqrt{\var(X_n)}\cdot O_\pr(1)$. 
\end{lemma}

\begin{lemma}\label{lem:var_wtzwtzp}
For $\ott$ and $z\neq z'\in\mathcal{Z}$, 
\cll{
\var\{\wtz^2\mid\hht\}
\le\uet ^{-3}, \quad 
\var\{ w_t(z)w_t(z')\mid\hht\}\le 
6\uet ^{-1}.}
\end{lemma}

\begin{proof}[\bf Proof of \lem~\ref{lem:var_wtzwtzp}]
\lem~\ref{lem:w_hist} ensures $\E\{\wtz^2\mid\hht\} =\rtz$
so that 
\begina
\var\left\{\wtz^2\mid\hht\right\} 
&\overset{\text{\lem~\ref{lem:w_hist}}}{=}&\E\Big[\left\{\wtz^2  -\rtz\right\}^2\mid\hht\Big]\\
&=&\E\Big[\left\{\wtz^2  -\rtz\right\}^2\mid Z_t = z,\hht\Big]\cdot\pr(Z_t = z\mid\hht)\\
&&+\E\Big[\left\{\wtz^2  -\rtz\right\}^2\mid Z_t\neq z,\hht\Big]\cdot\pr(Z_t\neq z\mid\hht)\\
&\overset{\eqref{eq:wtz}}{=}&\left\{\rtz^2 -\rtz\right\}^2 \cdot \etz 
+\left\{1 -\rtz\right\}^2 \cdot \{1-e_t(z)\}\\
&=& 
\left\{\rtz -1\right\}^2\left\{\rtz^2\cdot\etz  +1 -\etz\right\}
\\
&\overset{ \eqref{eq:rtz_1}}{=}& \left\{r_t(z)-1\right\}^2\rtz\\
&\overset{ \eqref{eq:rtz_1}}{\le}&  \uet ^{-3}. 
\enda
Similarly, \lem~\ref{lem:w_hist} ensures $\E\{w_t(z)w_t(z')\mid\hht\} = -1$ for $z\neq z'$ so that 
\begina
\var\{ w_t(z)w_t(z')\mid\hht\}
&\overset{\text{\lem~\ref{lem:w_hist}}}{=}&\E\Big[\big\{ w_t(z)w_t(z') + 1\big\}^2\mid\hht\Big]
\\
&=&\E\Big[\big\{ w_t(z)w_t(z') + 1\big\}^2\mid Z_t = z,\hht\Big]\cdot\pr(Z_t = z\mid\hht)\\
&&+\E\Big[\big\{ w_t(z)w_t(z') + 1\big\}^2\mid Z_t = z',\hht\Big]\cdot\pr(Z_t = z'\mid\hht)\\
&&+\E\Big[\big\{ w_t(z)w_t(z') + 1\big\}^2\mid Z_t\neq z, z',\hht\Big]\cdot\pr(Z_t\neq z, z'\mid\hht)\\
&\overset{ \eqref{eq:wtz}}{=}&\left\{-\rtz +1\right\}^2\cdot\etz +\left\{-r_t(z')+1\right\}^2\cdot e_t(z')
+2^2\cdot \{1-e_t(z)-\etzp\}
\\
&\overset{ \eqref{eq:rtz_1}}{\le}&
\etzi +\etzp^{-1} + 4\\
&\le& 6\uet ^{-1}. 
\enda
\end{proof}

\begin{lemma}\label{lem:Dtz_Dtzp}
For $z, z'\in\mathcal{Z}$, we have
\begina
\sumt\dtz\dtzp  -{\cv} _{\aipw, zz'} = \left\{\dfrac{1}{\tsq}\sumt\dfrac{ (\lt +\mt) ^4 }{\uet ^3}\right\}^{1/2}\cdot \oop . 
\enda
\end{lemma}

\begin{proof}[\bf Proof of \lem~\ref{lem:Dtz_Dtzp}]
{\lemd}\eqref{item:lem 3 1} ensures ${\cv} _{\aipw, zz'} =\sumt\E\{\dtz\dtzp\mid\hht\}$ so that 
\beginy\label{eq:dd_1}
\sumt\dtz\dtzp -{\cv} _{\aipw, zz'} =
\sumt\Big[\dtz\dtzp -\E\{D_t(z)\dtzp\mid\hht\}\Big] =\sumt\dt,
\endy
where 
\cll{\dt =\dtz\dtzp -\E\{D_t(z)\dtzp\mid\hht\}\with\E(\dt \mid \hht) = 0.}

Therefore, $\{\dt:\ott\}$ is {\mart} with
\beginy\label{eq:delta_1}
\E\left(\sumt \dt\right) = 0, \quad
\var\left(\sumt\dt\right)=\sumt\var(\dt) =\sumt\E\left\{\var(\dt\mid\hht )\right\}   
\endy
by \lem~\ref{lem:mds_cov}, where, given  $\dtz \oeq{\eqref{eq:dtz}} \tsri\wtz\cdot\atz$,  
\beginy\label{eq:delta_2}
\var(\dt\mid\hht )
&=&\var\left\{\dtz\dtzp\mid\hht\right\}\nnb\\
&\overset{ \eqref{eq:dtz}}{=}& T^{-2}\atz^2\atzp^2\cdot\var\left\{\wtz\cdot\wtzp\mid\hht\right\} \nnb\\
&\overset{\text{\lem~\ref{lem:var_wtzwtzp}}}{\leq}& T^{-2} \lmtp^4\cdot 6\uet^{-3}. 
\endy
This, together with \eqref{eq:dd_1} and \lem~\ref{lem:var_wtzwtzp}, ensures 
\begina
\E\left\{\sumt\dtz\dtzp -{\cv} _{\aipw, zz'}\right\}
&\overset{ \eqref{eq:dd_1}}{=}&\E\left(\sumt\dt\right)
\overset{\eqref{eq:delta_1}}{=} 0,\\
\var\left\{\sumt\dtz\dtzp -{\cv} _{\aipw, zz'}\right\}
&\overset{ \eqref{eq:dd_1}}{=}&\var\left(\sumt\dt\right) \overset{\eqref{eq:delta_1}+ \eqref{eq:delta_2}}{\le} 
\dfrac{6}{T^2}\sumt\dfrac{(\lt +\mt) ^4}{\uet ^3}. 
\enda
The result then follows from {\lemchebf}. 
\end{proof}

\begin{lemma}\label{lem:var_DtYt}
For $z, z'\in\mathcal{Z}$, we have 
\begina
\sumt\dtz\left\{Y_t(z') -\bar{Y}(z')\right\} = 
\tsrif\left\{\sumt\dfrac{\lt^2 (\lt +\mt) ^2 }{\uet}
+\max_{z\in\mathcal{Z}}\bar{Y}^2(z)\cdot\sumt\dfrac{ (\lt +\mt)^2 }{\uet}\right\}^{1/2}\cdot O_\pr(1).
\enda
\end{lemma}

\begin{proof}[\bf Proof of \lem~\ref{lem:var_DtYt}]
Let $\dt =\dtz\{Y_t(z') -\bar{Y}(z')\}$ with 
\cll{\E(\dt\mid H_t) =\{Y_t(z') -\bar{Y}(z')\}\cdot\E\{\dtz\mid H_t\} = 0}

by {\lemd}. Therefore,  
$\{\dt:\ott\}$ is {\mart}, and it follows from \lem~\ref{lem:mds_cov} that \beginy\label{eq:delta_4}
\E\left(\sumt\dt\right) = 0, \quad
\var\left(\sumt\dt\right) =\sumt\E\left\{\var(\dt\mid\hht )\right\}, 
\endy
where
\beginy\label{eq:delta_3}
\var(\dt\mid H_t) 
&=& 
\left\{Y_t(z') -\bar{Y}(z')\right\}^2\cdot\var\left\{\dtz\mid H_t\right\}\nnb\\
&\overset{\text{{\lemd}}}{=}&
T^{-1}\left\{Y_t(z') -\bar{Y}(z')\right\}^2\cdot \rtz\atz^2 
\nnb\\
&\overset{ \eqref{eq:rtz_1}}{\le}& 
 T^{-1} \cdot 2\left\{Y_t^2(z') +\bar{Y}^2(z')\right\}\cdot \ueti (\lt +\mt)^2\nnb\\
&\le& 
 2T^{-1}\left\{\lt^2+\max_{z\in\mathcal{Z}}\bar{Y}^2(z)\right\}\cdot \ueti (\lt +\mt)^2.
\endy
Plugging \eqref{eq:delta_3} into \eqref{eq:delta_4} ensures
\begina
\E\left[\sumt\dtz\left\{Y_t(z') -\bar{Y}(z')\right\}\right]
&=&\E\left(\sumt\dt\right) = 0,\\
\var\left[\sumt\dtz\left\{Y_t(z') -\bar{Y}(z')\right\}\right]
&=&\var\left(\sumt\dt\right)\\
&\le& 
2 T^{-1}\left\{\sumt\dfrac{\lt^2 (\lt +\mt) ^2 }{\uet}
+\max_{z\in\mathcal{Z}}\bar{Y}^2(z)\cdot\sumt\dfrac{ (\lt +\mt)^2 }{\uet}\right\}. 
\enda
The result then follows from \lem~\ref{lem:cheb}. 
\end{proof}

\begin{lemma}\label{lem:V_hat_aipw_diff}
\prear. Then 
\begina
\hva  -\dfrac{T}{T-1}\va  - S
 =\left\{\left(\dfrac{\psi}{T}\right)^{1/2} 
+\fn{\cva  -\va}\right\} \cdot 
\oop ,
\enda
where $\psi = T^{-1}\sumt \lmtp^4 /\uet ^3$. 
\end{lemma}

\begin{proof}[\bf Proof of \lem~\ref{lem:V_hat_aipw_diff}]
Let 
$S_{zz'} = (T-1)^{-1}\sumt\{\ytz -\bar{Y}(z)\}\{\ytzp -\bar{Y}(z')\}$ denote the $(z,z')$th element of $S$. 
Equation \eqref{eq:cov_ss} implies 
\begina
(T-1)\hat{V}_{\aipw, zz'}
&=& 
T\sumt\dtz\dtzp - T^2\bar{D}(z)\bar{D}(z') 
+ 
\tsr\sumt\dtz\{\ytzp -\bar{Y}(z')\} 
\\
&& 
+ 
\tsr\sumt\dtzp\{\ytz -\bar{Y}(z)\}+ (T-1)S_{zz'}. 
\enda
This ensures
\begina
&&(T-1)\hat{V}_{\aipw, zz'} - T V_{\aipw,zz'} - (T-1) S_{zz'}\nnb\\
& =&
T\left\{\sumt\dtz\dtzp  -{\cv} _{\aipw, zz'}\right\} 
+ T\left({\cv} _{\aipw, zz'} - V_{\aipw,zz'}\right)\nnb\\
    && - T^2\bar{D}(z)\bar{D}(z') + 
\tsr\sumt\dtz\{\ytzp -\bar{Y}(z')\} 
+ 
\tsr\sumt\dtzp\{\ytz -\bar{Y}(z)\},  
\enda
so that
\beginy\label{eq:diff_Vhat_V_aipw}
&&\dfrac{T-1}{T}\left( \hat{V}_{\aipw, zz'} - \dfrac{T}{T-1} V_{\aipw,zz'} -   S_{zz'} \right) \nnb\\
& =&
 \left\{\sumt\dtz\dtzp  -{\cv} _{\aipw, zz'}\right\} 
+  \left({\cv} _{\aipw, zz'} - V_{\aipw,zz'}\right)\nnb\\
    && - T \bar{D}(z)\bar{D}(z') + 
\tsrif\sumt\dtz\{\ytzp -\bar{Y}(z')\} 
+ 
\tsrif\sumt\dtzp\{\ytz -\bar{Y}(z)\}. \qquad \quad
\endy
Below we bound each of the terms on the right-hand side of \eqref{eq:diff_Vhat_V_aipw}. Some useful facts are 
\beginy\label{eq:bound_ss}
\meant\dfrac{(\lt +\mt) ^2 }{\uet}
&\le& 
\left\{\meant\dfrac{ (\lt +\mt) ^4 }{\uet ^2}\right\}^{1/2}
\le 
\left\{\meant\dfrac{ (\lt +\mt) ^4 }{\uet ^3}\right\}^{1/2}
=  \psi ^{1/2},\nnb\\
\meant\dfrac{\lt ^2 (\lt +\mt)^2 }{\uet}
&\le& \meant \dfrac{ (\lt +\mt) ^4 }{\uet ^3} =  \psi, 
\\
|\bar{Y}(z)| 
&\leq&  T^{-1}\sumt (\lt+\mt)\nnb\\
&\le&\left\{ T^{-1}\sumt (\lt +\mt) ^4\right\}^{1/4}
\le 
\left\{ T^{-1}\sumt\dfrac{ (\lt +\mt) ^4 }{\uet ^3}\right\}^{1/4}
= 
\psi ^{1/4} \nnb,
\endy
so that  
\begina
T^{-1}\max_{z\in\mathcal{Z}} |\bar{Y}(z)|
\cdot\left\{
\sumt\dfrac{ (\lt +\mt)^2 }{\uet}
\right\}^{1/2} \overset{\eqref{eq:bound_ss}}{\leq} T^{-1}\psi ^{1/4} 
\cdot
\left(
 T \psi^{1/2}\right)^{1/2} =\left(\dfrac{\psi}{T}\right)^{1/2}.
\enda

\paragraph*{\underline{First term on the RHS of \eqref{eq:diff_Vhat_V_aipw}.}} \lem~\ref{lem:Dtz_Dtzp} ensures 
\begina
\sumt\dtz\dtzp  -{\cv} _{\aipw, zz'} 
\overset{\text{\lem~\ref{lem:Dtz_Dtzp}}}{=}
  \left\{\dfrac{1}{\tsq}\sumt\dfrac{ (\lt +\mt) ^4 }{\uet ^3}\right\}^{1/2}\cdot \oop = \left(\dfrac{\psi}{T}\right)^{1/2}\cdot \oop . 
\enda

\paragraph*{\underline{Second term on the RHS of \eqref{eq:diff_Vhat_V_aipw}.}} The definition of Frobenius norm ensures
\begina
\left|{\cv} _{\aipw, zz'} - V_{\aipw,zz'}\right| 
\leq 
\fn{\cva  -\va}. 
\enda	

\paragraph*{\underline{Third term on the RHS of \eqref{eq:diff_Vhat_V_aipw}.}} Equation \eqref{eq:wtz} ensures $\rtz < \etzi$, so that
\begina
  V_{\aipw, zz} = T^{-1}\sumt\E\left\{r_t(z)\atzs\right\}\leq T^{-1}\sumt\E\left\{\etzi \atzs\right\}\leq T^{-1}   
\sumt  \dfrac{(\lt +\mt) ^2}{\uet} \overset{\eqref{eq:bound_ss}}{\leq} \psi^{1/2}.
   \enda
This, together with Lemmas \ref{lem:D} and \ref{lem:cheb}, ensures  
\begina
T \bar{D}(z)\bar{D}(z') 
&\overset{\text{Lemmas \ref{lem:D} + \ref{lem:cheb}}}{=}&T^{-1} \cdot \sqrt{\var\left\{ T\bar{D}(z)\right\}}\cdot
\sqrt{\var\left\{ T\bar{D}(z')\right\}}\cdot
O_{\pr} (1)\\
&\overset{\text{{\lemd}}}{=}& 
T^{-1} \cdot V^{1/2}_{\aipw, zz}\cdot V^{1/2}_{\aipw, z'z'}\cdot O_{\pr} (1)\\
&=&
 T^{-1}   
\psi^{1/2}\cdot O_{\pr} (1).
\enda

\paragraph*{\underline{Fourth and fifth terms on the RHS of \eqref{eq:diff_Vhat_V_aipw}.}} \lem~\ref{lem:var_DtYt} ensures 
\begina
&& \tsrif \sumt\dtz\{\ytzp -\bar{Y}(z')\} 
\\
&=&
 T^{-1} \sqrt{
\sumt\dfrac{\lt ^2 (\lt +\mt)^2 }{\uet}
+ 
\max_{z\in\mathcal{Z}}\bar{Y}^2(z)\cdot 
\sumt\dfrac{ (\lt +\mt)^2 }{\uet}
} \cdot \oop 
\\
&=& 
T^{-1}\left\{
\sqrt{\sumt\dfrac{\lt ^2 (\lt +\mt)^2 }{\uet}} 
+ 
\max_{z\in\mathcal{Z}} |\bar{Y}(z)|
\cdot \sqrt{ 
\sumt\dfrac{ (\lt +\mt)^2 }{\uet}}
\right\}
\cdot \oop \\
&\overset{\eqref{eq:bound_ss}}{=}& T^{-1}\left\{ (T\psi)^{1/2} + T   \left(\dfrac{\psi}{T}\right)^{1/2} \right\} \cdot \oop \\
&=& \left(\dfrac{\psi}{T}\right)^{1/2} \cdot \oop. 
\enda
and, by symmetry, 
\begina
\tsrif\sumt\dtzp\{\ytz -\bar{Y}(z)\} = \left(\dfrac{\psi}{T}\right)^{1/2} \cdot \oop.
\enda
Plugging these bounds into \eqref{eq:diff_Vhat_V_aipw} ensures that for all $z,z'\in\mathcal{Z}$,
\begina
\hat{V}_{\aipw, zz'} - \dfrac{T}{T-1}\cdot V_{\aipw,zz'} - S_{zz'}
&=& 
\left\{\left(\dfrac{\psi}{T}\right)^{1/2} 
+ 
\fn{\cva  -\va}
+
\dfrac{\psi ^{1/2}}{T}
+
\left(\dfrac{\psi}{T}\right)^{1/2} \right\} \cdot \oop\\ 
&=& 
\left\{\left(\dfrac{\psi}{T}\right)^{1/2} 
+\fn{\cva  -\va}\right\} \cdot \oop. 
\enda
\end{proof}

\subsection{Proofs of the finite-sample results}
\begin{proof}[\bf Proof of \thm~\ref{thm:aipw}\eqref{item:aipw_mean_var}]
Recall from \eqref{eq:mdt} that 
$
\hmya -\bmy = T^{-1/2}\sumt\mdt. 
$
This, together with {\lemd}, ensures 
\begina
\E(\hmya -\bmy ) \oeq{\eqref{eq:mdt}}\tsrif\E\left(\sumt\mdt\right) = 0,
\quad 
\cov(\hmya -\bmy )
\oeq{\eqref{eq:mdt}} T^{-1}\cov\left(\sumt\mdt\right)
= T^{-1}\va .
\enda
\end{proof}

\begin{proof}[\bf Proof of \prop~\ref{prop:aipw_cov_est}\eqref{item:aipw_var_est_exact}, finite-sample exact expectation]
Lemma~\ref{lem:D} ensures that $\{\mdt: t =\ot{T}\}$ is {\mart}. It then follows from Lemma~\ref{lem:mds_cov} that  
\beginy\label{eq:cov_est_ss1}
\E(\mdt) = 0,\quad\E\left\{\dfrac{T}{T-1} \sumt (\mdt -\bmd) (\mdt -\bmd)^\T\right\} = \cov\left(\sumt\mdt\right) =  \va.  
\endy
This, together \eqref{eq:cov_ss}, ensures
\begina
\E\left\{(T-1)\hva\right\} 
&\overset{\eqref{eq:cov_ss}+\eqref{eq:cov_est_ss1}}{=}&
\E\left\{T\sumt (\mdt -\bmd) (\mdt -\bmd)^\T\right\} 
+ (T-1)S \\
&\overset{\eqref{eq:cov_est_ss1}}{=}& (T-1)\va + (T-1) S 
\enda
so that $\E(\hva) =\va +S$. 

The necessary and sufficient condition for $CS\ct = 0$ follows from 
\begina
CS\ct &=&  C\left\{ \dfrac{1}{T-1}\sumt (\myt -\bmy ) (\myt -\bmy )^\T \right\}\ct \\
&=&  \dfrac{1}{T-1}\sumt  (C\myt -C\bmy ) (C\myt -C\bmy )^\T \\
&=& \dfrac{1}{T-1} \sumt(\tau_{C,t} - \tc)(\tau_{C,t} - \tc)^\T.
\enda 

\end{proof}

\subsection{Proofs of the asymptotic results}	
\begin{proof}[\bf Proof of \thm~\ref{thm:aipw}\eqref{item:aipw_clt}]
Recall from \eqref{eq:mdt} that $\hmya -\bmy = T^{-1/2}\sumt\mdt$. 
The goal is to prove
\cls{( C\va \ct)^{-1/2}\cdot\sqrt{T} C(\hmya -\bmy )
 \oeq{\eqref{eq:mdt}} ( C\va \ct)^{-1/2} C\cdot\sumt\mdt \converged \snrc}

for \allcf.
By the Cram\'er--Wold theorem, it suffices to verify that, for any constant unit vector $\eta$,
\beginy\label{eq:goal_clt}
\eta^\T ( C\va \ct)^{-1/2} C\cdot\sumt\mdt =\sumt\dt\converged\mathcal{N}(0, 1),
\endy
where $\dt =\eta^\T ( C\va \ct)^{-1/2} C\cdot\mdt$. It follows from $\E(\dt\mid\hht) = \eta^\T ( C\va \ct)^{-1/2} C\cdot\E(\mdt \mid \hht) = 0$ by \lemd\ that $\{\dt: \ott\}$ is {\mart}. 
Let 
\cll{\cvt =\sumt\E(\dt^2\mid\hht ), \quad\vt =\E(\cvt ),}
with
\beginy\label{eq:vt}
\cvt & =&\eta^\T ( C\va \ct)^{-1/2} C\cdot\left\{\sumt\E(\mdt\mdt ^\T\mid\hht )\right\}\cdot\ct ( C\va \ct)^{-1/2}\eta \nnb\\
&\oeqt{{\lemd}}&\eta^\T ( C\va \ct)^{-1/2} C\cva\ct ( C\va \ct)^{-1/2}\eta, \nnb\\
\vt &=&\eta^\T ( C\va \ct)^{-1/2} C\va\ct ( C\va \ct)^{-1/2}\eta 
=\eta^\T\eta = 1
\endy
by {\lemd}. 
We verify below that $\cvt $ and $\vt $ satisfy the following conditions in the martingale central limit theorem in \lem~\ref{lem:martingale_clt}, so that \eqref{eq:goal_clt} holds:
\beginy\label{eq:conds_clt}
\begin{array}{rl}
\text{Variance convergencen:} & \cvt /\vt = 1 +\op,\medskip\\
\text{Lindeberg:} & 
\vt ^{-1}\sumt\E\left\{\dt ^2\cdot\I ( |\dt  |\geq\epsilon\vt ^{1/2} )\right\} =\op \ \ \text{for all $\epsilon>0$.}
\end{array} 
\endy 
A useful fact is 
\cll{
\begin{array}{rclcl}
\lm( C\va\ct) &=& \inf_{\|\tilde{\eta}\|_2 = 1}\tilde{\eta}^\T C\va\ct\tilde{\eta}\\
&\ge&\lm(\va) \cdot \inf_{\|\tilde{\eta}\|_2 = 1} \tilde{\eta}^\T C\ct\tilde{\eta}
&=& \lm(\va)\cdot\lm( C\ct ) 
\end{array}
}

so that 
\beginy\label{eq:lambda_min_CVCT}
\| ( C\va \ct)^{-1/2}\|_2^2\cdot\| C\|_2^2
&\overset{ \eqref{eq:l2}}{=}&
\lambda_{\max}\left\{( C\va \ct)^{-1}\right\}\cdot\lambda_{\max} (C\ct) \nnb\\
&=&
\lm ^{-1}( C\va \ct)\cdot\lambda_{\max} (C\ct)\nnb\\
&\leq&\lmai\cdot \kappa(C\ct), 
\endy 
where $\kappa(C\ct) =\lmax(\cct)/\lmin(\cct)$.

%
\underline{\textbf{Proof of the variance convergence condition $\cvt /\vt = 1 +\op$ in \eqref{eq:conds_clt}.}}

Given $\vt =1$ from \eqref{eq:vt}, it suffices to verify $\cvt - \vt = \op$. Note that
\begina 
(\cvt -\vt )^2 
&=& \tn{\cvt -\vt}^2\\
&\oeq{\eqref{eq:vt}}&
\left\|
\eta^\T ( C\va \ct)^{-1/2} C\left(\cva -\va\right)\ct ( C\va \ct)^{-1/2}\eta\right\|_2^2 
\\
&\overset{\subm}{\le}&
\| ( C\va \ct)^{-1/2}\|_2^4\cdot\| C\|_2^4\cdot \|\cva -\va\|_2^2\\
&\overset{ \eqref{eq:lambda_min_CVCT}+\eqref{eq:l2}}{\le} &
\lmap^{-2}
\cdot \kappa(C\ct) ^2\cdot \fn{\cva -\va}^2. 
\enda
Therefore, \cond~\ref{cond:aipw_clt}\eqref{item:clt_aipw_var_conv} ensures that $\cvt - \vt = \op$.

%
\underline{\textbf{Proof of the Lindeberg condition $\vt ^{-1}\sumt\E\{\dt ^2\cdot\I ( |\dt  |\geq\epsilon\vt ^{1/2} )\} =\op$ in \eqref{eq:conds_clt}.}}

Given $\vt =1$ from \eqref{eq:vt}, it suffices to verify 
\beginy\label{eq:linde_2}
\sumt\E\{\dt ^2\cdot\I ( |\dt|\geq\epsilon)\} =\op\quad\text{for all $\epsilon>0$.}
\endy
Recall from \eqref{eq:wtz}--\eqref{eq:rtz_1} that $
|\wtz| \oeq{\eqref{eq:wtz}} |\indz \cdot\etzi - 1| \overset{\eqref{eq:rtz_1}}{\leq} \uet ^{-1}
$. This implies
\begina
|\dtz | \overset{\eqref{eq:dtz}}{=}\left| \tsrif \wtz\cdot \atz\right|
\le \tsrif(\lt +\mt)/\uet,
\enda
so that
\beginy\label{eq:dtz_sum}
\tn{\mdt}^2 = \sum_{z\in\mathcal{Z}}\dtz^2 
\le 
K\cdot T^{-1}(\lt +\mt)^2/\uet^2.
 \endy 
Equations \eqref{eq:lambda_min_CVCT} and \eqref{eq:dtz_sum} together ensure
\beginy\label{eq:aipw_clt_proof_lindeberg}
\dt^2 =\tn{\dt}^2 
&\overset{ \eqref{eq:goal_clt}}{=}&\left\|\eta^\T ( C\va \ct)^{-1/2} C\cdot\mdt\right\|_2^2 \nnb\\
&\overset{\subm}{\le}&\|( C\va \ct)^{-1/2}\|_2^2\cdot\|C\|_2^2\cdot\|\mdt\|_2^2\nnb\\ 
&\overset{\text{ \eqref{eq:lambda_min_CVCT}+\eqref{eq:dtz_sum}}}{\le}&
\left[\lmai \cdot \kappa(C\ct)\right] \cdot
\left\{ K\cdot T^{-1}(\lt +\mt)^2/\uet^2 \right\} \nnb\\
&\le&
K\cdot\kappa(C\ct)\cdot \lmai 
\cdot 
 T^{-1}\maxt(\lt +\mt)^2/\uet^2.\quad
\endy
\cond~\ref{cond:aipw_clt}\eqref{item:clt_aipw_lindeberg} requires $\lmai 
\cdot 
 T^{-1}\maxt(\lt +\mt)^2/\uet^2 = \oo$, so that for any $\epsilon>0$, there exists $T_0$ such that for all $T\geq T_0$, 
\begina
\lmai 
\cdot 
 T^{-1}\maxt(\lt +\mt)^2/\uet^2  <  \dfrac{\epsilon^2}{K\cdot\kappa(C\ct)}\\
 \text{so that, from \eqref{eq:aipw_clt_proof_lindeberg},} \quad \dt ^2 \overset{\eqref{eq:aipw_clt_proof_lindeberg}}{<}  \epsilon^2 \quad\text{for all $\ott $.}
 \enda
This ensures
$\sumt\E \{\dt ^2\cdot\I ( |\dt  |\geq\epsilon) \} = 0$ for all $T\geq T_0$, which implies \eqref{eq:linde_2}. 
\end{proof}

\begin{proof}[\bf Proof of \prop~\ref{prop:aipw_sc}]
The result follows from \prop~\ref{prop:aipw_sc_app}, which we prove in Section \ref{sec:proof_sc_app}.
\end{proof}

\begin{proof}[\bf Proof of \prop~\ref{prop:aipw_cov_est}\eqref{item:aipw_var_est_asym}, large-sample results]
\lem~\ref{lem:V_hat_aipw_diff} implies 
\beginy\label{eq:var_est_2}
&&\hva  -\dfrac{T}{T-1}\va  - S \nnb \medskip \\
&=&\left\{\left(\dfrac{\psi}{T}\right)^{1/2} 
+\fn{\cva  -\va}\right\} \cdot 
\oop \nnb \\
&=& \lma \cdot \left\{\dfrac{1}{\lma}\left(\dfrac{\psi}{T}\right)^{1/2} 
+ \dfrac{1}{\lma}\cdot\fn{\cva  -\va}\right\} \cdot \oop,\quad
\endy
where $\psi = T^{-1}\sumt \lmtp^4 /\uet ^3$. Conditions~\ref{cond:aipw_cs} and \ref{cond:aipw_clt} ensure
\beginy\label{eq:op}
\dfrac{1}{\lma}\left(\dfrac{\psi}{T}\right)^{1/2} = \op, \quad \dfrac{1}{\lma}\cdot\fn{\cva  -\va} = \op,
\endy
respectively, in the middle term of \eqref{eq:var_est_2}, so that 
\cll{\hva - \va  - S = \lma \cdot \op\cdot \oop = \lma \cdot \op.} 

\end{proof}

\begin{proof}[\bf Proof of \prop~\ref{prop:aipw_cov_est_2}]
For two sequences of vectors $(a_t)_{t=1}^T$ and $(b_t)_{t=1}^T$, 
define
\cll{
\scov{a_t, b_t} = (T-1)^{-1}\sumt (a_t - \bar a)(b_t - \bar b)^\T}

as the sample covariance matrix of $(a_t,b_t)_{t=1}^T$, where $\bar a = \meant a_t$ and $\bar b = \meant b_t$. 
Let 
\beginy\label{eq:aha}
\Sigma = \scov{\myt, \hbmt}, \quad \hsig = \scov{\hmyta, \hbmt}
\endy
denote the sample covariance matrices of $(\myt, \hbmt)_{t=1}^T$ and $(\hmyta, \hbmt)_{t=1}^T$, respectively.  
Then 
\beginy\label{eq:tva}
\tva = \scov{\tmyta} 
&=&\scov{\hmyta  -  \hbmt}\nnb\\
&=&\scov{\hmyta}-\scov{\hmyta, \hbmt} -\scov{ \hbmt, \hmyta }+\scov{\hbmt}\nnb\\
&=&\hva - \hsig - \hsig^\T + \scov{\hbmt},
\endy
\beginy\label{eq:tva_S_Delta}
S = \scov{\myt} &=& \scov{\myt - \hbmt + \hbmt}\nnb\\
&=& \scov{\myt - \hbmt}+ \scov{\myt - \hbmt, \hbmt} + \scov{\hbmt, \myt - \hbmt} + \scov{\hbmt}\nnb\\
&=& \Omega + \scov{\myt , \hbmt} - \scov{\hbmt} + \scov{\hbmt, \myt } - \scov{\hbmt} + \scov{\hbmt}\nnb\\
&=& \Omega + \Sigma  + \Sigma^\T - \scov{\hbmt},\where \Omega = \scov{\myt-\hbmt}.
\endy
Equations \eqref{eq:tva}--\eqref{eq:tva_S_Delta} and  \prop~\ref{prop:aipw_cov_est} together ensure
\begina
\tva - \va &\overset{\eqref{eq:tva}}{=}&  \hva    - \hsig - \hsig^\T + \scov{\hbmt}- \va\\
&\overset{\text{\prop~\ref{prop:aipw_cov_est}}}{=}& S   - \hsig - \hsig^\T + \scov{\hbmt}+ \lm(\va) \cdot\op\\
&\overset{\eqref{eq:tva_S_Delta}}{=}& \Omega - (\hsig - \Sigma) - (\hsig - \Sigma)^\T + \lm(\va) \cdot \op,
\enda
so that it suffices to verify that 
\beginy\label{eq:goal_cov_est}
 \hsig - \Sigma = \lm(\va) \cdot \op.
 \endy
From the discussion above, we essentially provide an consistent estimation for $S-\Omega$. Therefore, such an improvement on reducing
conservativeness in variance estimation can also be applied to $\hvi$ in \eqref{eq:hvi}.
 
\paragraph*{A sufficient condition for \eqref{eq:goal_cov_est}.} Let $\bbm = \meant \hbmt$. Then 
\beginy\label{eq:aha_2}
\hsig - \Sigma &\oeq{\eqref{eq:aha}}& \scov{\hmyta-\myt, \hbmt} \nnb\\
 &=&  \dfrac{1}{T-1}\left\{\sumt (\hmyta-\myt)\cdot\hbmt^\T  - T(\hmya -\bmy)\cdot \bbm^\T\right\}\nnb\\
&=& \dfrac{T}{T-1}(\Sigma_1  -  \Sigma_2),
\endy
where
\begina
\Sigma_1 = \meant(\hmyta-\myt)\cdot\hbmt^\T, \quad \Sigma_2 =(\hmya -\bmy)\cdot \bbm^\T. 
\enda
Recall from \eqref{eq:mdt} that 
\cls{\hytaz-\ytz = \tsr\dtz, \quad \hyaz - \byz = \tsri \sumt \dtz.}

The $(z,z')$th elements of $\Sigma_1$ and $\Sigma_2$ equal 
\beginy\label{eq:a1a2}
\beginar{l}\ds
\sigozzp 
 = \meant \left\{\hytaz-\ytz\right\}\cdot \hmtzp \oeq{\eqref{eq:mdt}} \tsrif \sumt \dtz \cdot \hmtzp = \tsrif \sumt \dt,\bigskip\\
\ds \sigtzzp = \left\{ \hyaz -\byz \right\} \cdot \bar m(z') \oeq{\eqref{eq:mdt}}  \tsrif \left\{\sumt  \dtz \right\} \cdot \left\{\meant \hmtzp\right\}.
\endar
\endy 
respectively, 
where $\dt =   \dtz\cdot\hmtzp$. We verify below 
\beginy\label{eq:tva_a1a2_goal}
\sigozzp =\lm(\va) \cdot \op, \quad \sigtzzp = \lm(\va) \cdot\op 
\endy for all $z,z'\in \mz$. This, together with \eqref{eq:aha_2}, implies \eqref{eq:goal_cov_est}.  
 
\paragraph*{\underline{Proof of $\sigozzp =\lm(\va) \cdot \op$ in \eqref{eq:tva_a1a2_goal}.}} Given $\dt = \dtz\cdot\hmtzp$ from \eqref{eq:a1a2}, \lemd\ ensures 
\beginy\label{eq:tva_dt}
\beginar{rcl}
\E(\dt\mid H_t) &=& \E\{\dtz\mid\hht\}\cdot\hmtzp  \oeqt{{\lemd}}  0,\medskip\\
\var(\dt \mid \hht)
 &=& \hmtzps \cdot\var \left\{ \dtz \mid \hht \right\} \\
&\oeqt{{\lemd}}& \hmtzps \cdot T^{-1}\rtz \atzs \\
&\overset{ \eqref{eq:rtz_1}}{\leq}& T^{-1} \lmtp^4/\uet^3.
\endar
\endy
Therefore, $\dt$ is {\mart}, and it follows from \lem~\ref{lem:mds_cov} that 
\beginy\label{eq:a1}
\beginar{l}\ds
\E(\sigozzp) = \tsrif\sumt\E\left(\dt\right) = 0,\medskip\\
\ds \var(\sigozzp) = T^{-1}\var\left(\sumt \dt\right) = T^{-1}\sumt \E\left\{\var(\dt \mid \hht)\right\}
\overset{ \eqref{eq:tva_dt}}{\leq}
T^{-2}\sumt    \lmtp^4 / \uet^3,
 \endar
\endy
where $T^{-2}\sumt    \lmtp^4 / \uet  = \lm^2(\va) \cdot \oo $ under \cond~\ref{cond:aipw_cs}.  
This implies
\begina
\sigozzp \oeqt{\lem~\ref{lem:cheb} +\eqref{eq:a1}} \sqrt{\var(\sigozzp)} \cdot \oop \oeqt{\eqref{eq:a1} + \cond~\ref{cond:aipw_cs}}  \lm(\vi) \cdot \op
\enda
under \cond~\ref{cond:aipw_cs} by {\lemchebf}.

\paragraph*{\underline{Proof of $\sigtzzp = \lm(\va) \cdot \op$ in \eqref{eq:tva_a1a2_goal}.}}
\lemd\ implies   
\cll{\E\{\dtz\mid \hht\} = 0, \quad \var\left\{ \dtz\mid \hht \right\} = T^{-1}\rtz \atzs \overset{\eqref{eq:rtz_1}}{\leq} T^{-1}   \ueti  \lmts,}

so that $\E\left\{\sumt \dtz\right\} = 0$ and 
\cll{\var\left\{\sumt \dtz \right\} \oeqt{{\lemd}}\sumt \E\Big[\var\left\{ \dtz \mid \hht \right\}\Big] \overset{\eqref{eq:rtz_1}}{\leq} \meant \ueti  \lmts.}

This, together with {\lemchebf}, ensures  
\beginy\label{eq:a2_D}
\left\{\sumt  \dtz \right\}^2 = \var\left\{\sumt \dtz \right\} \cdot \oop = \left\{\meant \ueti  \lmts\right\} \cdot \oop. 
\endy
In addition, 
\beginy\label{eq:a2_m}
\left|\meant \hmtzp\right| \leq \meant \mt \leq \left(\meant \mt^2\right)^{1/2} 
\endy
by \csf. 
Plugging \eqref{eq:a2_D}--\eqref{eq:a2_m} into the expression of $\sigtzzp$ in \eqref{eq:a1a2} implies  
\begina
\sigtzzp^2 
&\overset{\eqref{eq:a1a2}}{=}& T^{-1} \left\{\sumt  \dtz \right\}^2 \cdot \left|\meant \hmtzp\right|^2\\
&\overset{\eqref{eq:a2_D} + \eqref{eq:a2_m}}{=}& T^{-1} \left\{\meant \ueti  \lmts\right\} \cdot \left(\meant \mt^2\right) \cdot \oop \\
&\leq& T^{-1} \left\{\meant \ueti  \lmts\right\}^2 \cdot \oop\\
&\overset{\text{\cs}}{\leq}& T^{-1} \left\{\meant \uet^{-2} \lmtp^4\right\} \cdot \oop\\
&\leq& T^{-2} \sumt \uet^{-3} \lmtp^4 \cdot \oop\\
&\overset{\text{\cond~\ref{cond:aipw_cs}}}{=}&\{\lma\}^2 \cdot \op. 
\enda
  
\end{proof}

 \begin{proof}[\bf Proof of \thm~\ref{thm:aipw_cs}]
We verify below the results for $\sa$. The proof of the results for $\tsa$ is identical after replacing $\hva$ by $\tva$, hence omitted. 

From \lem~\ref{lem:V_hat_aipw_diff} and \eqref{eq:lambda_min_CVCT}, under Conditions~\ref{cond:aipw_clt}--\ref{cond:aipw_cs}, for \allcf,  
\begina
&&\left\| \left(C\va \ct\right)^{-1/2} C \left\{\hva  - T/(T-1)\cdot \va  - S\right\}\ct (C\va C)^{-1/2}\right\|_2 
\\
&\overset{\subm}{\le}&  \|  (C\va \ct )^{-1/2} \|_2^2\cdot \left\|C\right\|_2^2\cdot\left\|\hva  - T/(T-1)\cdot \va  - S\right\|_2
\\
        &\overset{\eqref{eq:lambda_min_CVCT} + \eqref{eq:l2}}{\le}& 
\kappa(C\ct)\cdot\lm^{-1} (\va)\cdot\|\hva  - T/(T-1)\cdot \va  - S\|_{\textsc f}\\
&\overset{\text{\lem~\ref{lem:V_hat_aipw_diff}}}{=}& 
\kappa(C\ct)\cdot\lm^{-1} (\va)
\cdot  \left\{
 \left(\dfrac{\psi}{T}\right)^{1/2} 
+\fn{\cva  -\va}
\right\} \cdot \oop
\\
&\oeq{\eqref{eq:op}}& \op . 
\enda
This ensures
\beginy\label{eq:ci_ss}
&& \left(C\va \ct\right)^{-1/2} \left( C\hva\ct - CS\ct \right) \left(C\va \ct\right)^{-1/2} \nnb
\\
& =& \left(C\va \ct\right)^{-1/2} \left\{T/(T-1)\cdot C \va \ct \right\} \left(C\va \ct\right)^{-1/2} + \op, \nnb\\
&=& \irc + \op . 
\endy
Consequently, we have 
\beginy\label{eq:chi}
&&T (\hat{\tau}_{\aipw, C} -\tau_C )^\T \left(C\hva \ct - CS\ct\right)^{-1} 
(\hat{\tau}_{\aipw, C} -\tau_C )
\nnb\\
&=& 
\left\{\sqrt{T} \left(C\va \ct\right)^{-1/2} (\hat{\tau}_{\aipw, C} -\tau_C )\right\}^\T \nnb\\
&& \cdot 
\left\{ \left(C\va \ct\right)^{-1/2} \left( C\hva\ct - CS\ct \right) \left(C\va \ct\right)^{-1/2}\right\}^{-1} 
\nnb\\
&&\cdot 
\left\{\sqrt{T} \left(C\va \ct\right)^{-1/2} (\hat{\tau}_{\aipw, C} -\tau_C )\right\}
\nnb\\
&\overset{\eqref{eq:ci_ss}}{=}& 
\left\{\sqrt{T} \left(C\va \ct\right)^{-1/2} (\hat{\tau}_{\aipw, C} -\tau_C )\right\}^\T\cdot\big\{ \irc + \op \big\}\cdot 
\left\{\sqrt{T} \left(C\va \ct\right)^{-1/2} (\hat{\tau}_{\aipw, C} -\tau_C )\right\}
\nnb\\
& =&
\big\|\sqrt{T} \left(C\va \ct\right)^{-1/2} (\hat{\tau}_{\aipw, C} -\tau_C )\big\|_2^2 + \op 
\converged\chi^2_\rc, 
\endy
where $\chi^2_\rc$ denotes the chi-squared distribution with $\rc$ degrees of freedom, and the last equality and the last convergence in \eqref{eq:chi} hold due to \thm~\ref{thm:aipw}\eqref{item:aipw_clt} and Slutsky's theorem. 

This, together with $(C\hva\ct )^{-1}\leq  (C\hva \ct - CS\ct )^{-1}$,
 ensures
\begina
&&\liminf_{T\to\infty}\pr\left\{ T (\hta -\tau_C )^\T \left(C\hva\ct\right)^{-1} 
(\hta -\tau_C )\le\chisqa\right\}
\\
&\ge& 
\lim_{T\to\infty}\pr\left\{ T (\hta -\tau_C )^\T \left(C\hva\ct - C S\ct\right)^{-1} 
(\hta -\tau_C )\le \chisqa\right\}\\
&\overset{\eqref{eq:chi}}{=}&
1-\alpha. 
\enda
\end{proof}

\section{Proof of the result in Section \ref{sec:block_ad}}\label{sec:block_app}
\def\prophvab{\prop~\ref{prop:hvab}}
\def\bt{b_t}
\def\bmat{\bm a_t}
\def\bba{\bar{\bm a}}
\def\bmct{\bm c_t}
\def\bbc{\bar{\bm c}}
All results on block adaptive randomization in \sec~\ref{sec:block_ad}, except those on $\hvab$ and $\tvab$ in \sec~\ref{sec:block_covb}, follow directly from the theory for the AIPW estimator under unit-level \ar\ in Section~\ref{sec:aipw} by treating each group as a unit and renewing (i) $\ytz$ as $\rhot \bar Y_t(z)$, where $\rhot = T \cdot \pt = (T/N) \cdot n_t$ and $\bar Y_t(z) = \meani \ytiz$, and (ii) $\hyta(z)$ as $\rhot \cdot \meani\hytiaz$. 

We verify below \prophvab\ on $\hvab$  and then provide intuition for  \beginy\label{eq:tvab_intuition}
\tvab = T\cov(\hmya) + \cov_b (\bmyt-\hbmt ) + \op, 
\endy
where $\cov_b(\bmyt-\hbmt)$ is define as follows. For tuples of vectors $\{(\bmat, \bmct)\}_{t=1}^T$, 
let
\cll{
\cov_b (\bmat, \bmct)=\sumt \bt \left(\bmat -\bba \right)\left(\bmct -\bbc\right)^{\top},\quad \cov_b(\bmat) = \cov_b(\bmat, \bmat) 
}

where $\bba = \sumt \pt \bmat $, $\bbc = \sumt \pt \bmct$, and 
\beginy\label{eq:bt_app}
  b_t = T \cdot \dfrac{\pt^2/(1-2\pt)}{1 +  \sums \ps^2/(1-2\ps)}\quad\text{with $\pt = \nt/N$},
\endy
as defined in the main paper. Then 
\beginy\label{eq:hvab_app}
\begin{array}{l}
\hvab =\ds \sumt  \bt \left(\hmyta -\hmya \right)\left(\hmyta -\hmya \right)^\T = \cov_b(\hmyta),\\
\tvab =\ds \sumt \bt \left(\tmyta -\tmya \right)\left(\tmyta -\tmya \right)^\T  = \cov_b(\tmyta). 
\end{array} 
\endy
Assume the \br\ in Definition~\ref{def:block_ad} throughout this section.

\subsection{\bf Proof of \prophvab}
\prophvab\eqref{item:hva_bt} follows from the definition of $b_t$ and the assumption of $\pt  < 1/2$. 
We verify below \prophvab\eqref{item:hva_i}--\eqref{item:hva_ii}, respectively.
\prophvab\eqref{item:hva_iii} then follows from \prophvab\eqref{item:hva_bt} and \eqref{item:hva_ii}.  

\paragraph*{Proof of \prophvab\eqref{item:hva_i}.} When blocks are of equal sizes, $\pt=1/T$ for all $1 \leq t \leq T$, which implies that $b_t = 1/(T-1)$ for $\ott$ so that $\hvab = \hva$.

\paragraph*{Proof of \prophvab\eqref{item:hva_ii}.} Write
\beginy
\hvab &\oeq{\eqref{eq:hvab_app}}&\underbrace{\sumt \bt \hmyta \hmyta ^\T}_{B_1} +
\underbrace{ \sumt \bt \hmya \hmya ^\T}_{B_2} -
\underbrace{ \sumt \bt \hmyta \hmya ^\T}_{B_3}  
- \sumt \bt  \hmya \hmyta ^\T \nnb\\
&=& B_1 + B_2 - B_3 - B_3^\T,\label{eq:hvab_decomp}
\endy
where
\begina
B_1 = \sumt \bt \hmyta \hmyta ^\T, \quad B_2 =  \sumt \bt \hmya \hmya ^\T, \quad B_3 = \sumt \bt \hmyta \hmya ^\T.
\enda
Recall that $\myti = (Y_{ti}(1), \ldots, Y_{ti}(K))^\T$ denote the potential outcomes vector for unit $ti$, and $\bmyt = \meani \myti = (\byt(1), \ldots, \byt(K))^\T$ denote the group average within block $t$.
Let
\begina
\bmy = (\by(1), \ldots, \by(K))^\T = \sumt \pt \cdot \bmyt
\enda 
denote the average potential outcomes vector for all units.
It follows from 
\beginy\label{eq:hvab_ss}
\begin{array}{lll}
\E(\hmyta) = \bmyt, &&\E(\hmya) = \bmy, \\
\hmya = \ds\sumt \pt \cdot \hmyta, &&\cov(\hmyta , \hat{\bm Y}_{t',\aipw})=0 \ ( t\neq t' )
\end{array}
\endy
that
\beginy\label{eq:e123}
\begin{array}{c}
\begin{array}{lllll}
\E(B_1)&\oeq{\eqref{eq:hvab_decomp}}& \ds\sumt \bt \cdot \E(\hmyta\hmyta^\T) &\oeq{\eqref{eq:hvab_ss}}& \ds\sumt \bt \cdot \cov(\hmyta) + \sumt \bt  \bmyt\bmyt^\T,\medskip\\ 
\E(B_2)&\oeq{\eqref{eq:hvab_decomp}}& \ds \sumt \bt \cdot \E(\hmya\hmya^\T) &\oeq{\eqref{eq:hvab_ss}}&\ds \sumt \bt \cdot \cov(\hmya)+ \sumt \bt \bmy \bmy^\T,
\end{array}\bigskip\\
\begin{array}{rcl}
\E(B_3)
&\oeq{\eqref{eq:hvab_decomp}}&\ds\E\left\{\left(\sumt \bt \hmyta \right)  \cdot \hmya \right\}\medskip\\
&\oeq{\eqref{eq:hvab_ss}}&\ds \cov\left(\sumt \bt  \hmyta, \hmya \right)+
\left(\sumt \bt  \cdot \bmyt \right) \cdot \bmy^\T\medskip\\
&\oeq{\eqref{eq:hvab_ss}}&\ds \cov\left(\sumt \bt  \hmyta , \sumt  \pt \hmyta \right)+
 \sumt \bt \bmyt \bmy^\T\medskip\\
&=&\ds\sumt \bt  \pt \cdot \cov(\hmyta)+  \sumt\bt \bmyt \bmy^\T.
\end{array}
\end{array}
\endy
In addition, let $B_0 = \sumt \pt^2/(1-2\pt)$ to write  
\beginy\label{eq:bt_2}
\bt \oeq{\eqref{eq:bt_app}} \dfrac{T}{1+B_0}\cdot \dfrac{\pt^2}{1-2\pt},
\quad b_t (1-2\pt)
=
\frac{T }{1+B_0} \cdot \pt^2, \quad \sumt \bt = \dfrac{T}{1+B_0}\cdot B_0 . 
\endy
This implies 
\beginy\label{eq:b4}
B_4 &\equiv& \sumt \bt \left(1-2\pt\right) \cdot \cov(\hmyta)+\sumt \bt  \cdot \cov(\hmya) \nnb\\
 &\oeq{\eqref{eq:bt_2}}& \ds\frac{T}{1+B_0} \cdot \sumt \pt^2 \cdot \cov(\hmyta)+ \dfrac{TB_0}{1+B_0} \cdot \cov(\hmya) \nnb\\
 &=& \ds\frac{T}{1+B_0} \cdot \cov(\hmya)+\dfrac{TB_0}{1+B_0}\cdot \cov(\hmya) \nnb\\
 &=& T \cdot \cov(\hmya).
\endy
Equations \eqref{eq:hvab_decomp}, \eqref{eq:e123}, and \eqref{eq:b4} ensure
\begina
\E(\hvab)
&\oeq{\eqref{eq:hvab_decomp}+\eqref{eq:e123}}&\underbrace{\sumt \bt \bmyt \bmyt^{\top}
+\sumt \bt  \bmy \bmy^\T -\sumt\bt \bmyt \bmy^\T - \sumt\bt \bmy\bmyt^\T }_{\sumt \bt (\bmyt -\bmy) (\bmyt -\bmy )^\T }\\
&&+\underbrace{\sumt \bt \left(1-2\pt\right) \cdot \cov(\hmyta)+\sumt \bt  \cdot \cov(\hmya)}_{\oeq{\eqref{eq:b4}} B_4}\\
&\oeq{\eqref{eq:b4}}& \sumt \bt (\bmyt -\bmy) (\bmyt -\bmy )^\T + T \cdot \cov\left(\hmya \right)
\enda

\subsection{Intuition for \eqref{eq:tvab_intuition}: $
\tvab = T\cov(\hmya) + \cov_b(\bmyt-\hbmt) + \op.$}
The definition of $\cov_b(\cdot, \cdot)$ and \eqref{eq:hvab_app} ensure 
\beginy\label{eq:tvab}
\tvab  &\oeq{\eqref{eq:hvab_app}}& \cov_b(\tmyta) \nnb \\
&=& \cov_b(\hmyta -\hbmt ) \nnb \\
&=& \cov_b(\hmyta )+\cov_b(\hbmt )-\cov_b(\hmyta , \hbmt )-\cov_b(\hbmt , \hmyta )\nnb \\
&\oeq{\eqref{eq:hvab_app}}& \hvab +\cov_b(\hbmt )-\cov_b(\hmyta , \hbmt )-\cov_b(\hbmt , \hmyta ),\\
\label{eq:covb_bmyt}
\cov_b (\bmyt  ) &=& \cov_b(\bmyt - \hbmt  + \hbmt  ) \nnb\\
&=& \cov_b(\bmyt  - \hbmt ) 
+ \cov_b(\hbmt ) 
+ \cov_b(\bmyt  - \hbmt , \hbmt ) 
+ \cov_b(\hbmt , \bmyt  - \hbmt )\nnb \\
&=& \cov_b(\bmyt  - \hbmt ) 
 - \cov_b(\hbmt ) 
+ \cov_b(\bmyt  , \hbmt ) 
+ \cov_b(\hbmt , \bmyt ).  
\endy
Let $\Delta = \hvab - \cov_b(\bmyt)$ with $\E(\Delta) = T\cov(\hmya)$ by \prophvab\eqref{item:hva_ii}. Equations~\eqref{eq:tvab}--\eqref{eq:covb_bmyt} ensure
\begina
\tvab  -\cov_b (\bmyt-\hbmt ) 
&=&    \tvab  -\hvab  + \cov_b(\bmyt) - \cov_b(\bmyt-\hbmt)  + \left\{\hvab - \cov_b(\bmyt)\right\}  \\ 
&\oeq{\eqref{eq:tvab}+\eqref{eq:covb_bmyt}}&  \cov_b(\hbmt )-\cov_b(\hmyta , \hbmt )-\cov_b(\hbmt , \hmyta )  \\
&&  
 - \cov_b(\hbmt ) 
+ \cov_b(\bmyt  , \hbmt ) 
+ \cov_b(\hbmt , \bmyt )   + \Delta\\
&=&  
\Delta- \cov_b(\hmyta  - \bmyt , \hbmt ) 
- \cov_b(\hbmt , \hmyta  - \bmyt ) \\
&=& \Delta -B- B^\T,
\enda 
where $B  = \cov_b (\hmyta  - \bmyt , \hbmt  )$, 
with  
\begina
\E(\tvab)  -\cov_b(\bmyt-\hbmt) 
&=& \E(\Delta) -\E(B) - \E(B)^\T \\
&\oeq{\text{\prophvab\eqref{item:hva_ii}}}& T\cov(\hmya) -\E(B) - \E(B)^\T.  
\enda
Renew $\bbm =  \sumt  \pt \hbmt = \meant \rhot \hbmt$. 
Then
\begina
B &=& \cov_b(\hmyta  - \bmyt , \hbmt )\\
 &=&\sumt \bt \left\{ (\hmyta -\bmyt )- (\hmya -\bmy)\right\}(\hbmt -\bbm) \\
&=& \sumt \underbrace{\bt (\hmyta -\bmyt )(\hbmt -\bbm)}_{\text{unbiased for 0}}-\underbrace{(\hmya -\bmy)}_{\text{close to 0}} \cdot \sumt \bt (\hbmt -\bbm),
\enda
so that intuitively $B = \op$.

\section{Proofs of the results in \sec~\ref{sec:sc_app}}	\label{sec:proof_sc_app}

Recall that 
\begina
\begin{array}{llllll}
\ds \lt =\maxz |\ytz|,&&\ds \uet =\min_{H_t,\, \ziz}\etz, &&\ds \mt=\max_{H_t,\, \ziz} |\hmtz|,\medskip\\
&&\ds v_t =\maxz\var\{\etzi\}, &&\ds \omt =\maxz\var\{\hmtz\}.
\end{array}
\enda
Further let 
\beginy\label{eq:gt}
\gamma_t =\maxz\var\left\{\dfrac{\atzs}{\etz}\right\},\quad
\phit =\max_{z,z'\in\mz}\var\left\{\atz\atzp\right\}.  
\endy
\paragraph*{Useful facts.} Let $\infn{\cdot}$ denote the maximum value of a random variable. 
For random variables $X_1, \ldots, X_m \in \mbr$ with finite variances, standard results ensure 
\beginy\label{eq:var_lem}
\beginar{rcl}
\ds\var\left(\sum_{i=1}^m X_i\right) &\leq& m \ds\sum_{i=1}^m \var(X_i),\medskip\\
\var(X_1 X_2)&\leq& 2 \cdot \var(X_1)\cdot\infns{X_2} +2 \cdot \var(X_2)\cdot\infns{X_1}.
\endar
\endy 
%

\subsection{Lemmas}  

\begin{lemma}\label{lem:v1}
\prear. Let
\beginy\label{eq:v1_def}
\voz=\meant\left[\dfrac{\atzs}{\etz } -\E\left\{\dfrac{\atzs}{\etz }\right\}\right] \quad \text{for} \ \ \ziz.
\endy
As $T\to\infty$, we have 
\cls{\var(\voz) = o(1) \ \ \text{for all $\ziz$}}

if any of the following conditions holds:
\begine[(i)]
\item\label{item:v1_i}
 $\meant\gamma_t = o(1)$.

\item\label{item:v1_ii}
$\meant\left\{\maxs ( L_s+ M_s)^2/\ue_s\right\} \cdot \lmtp^2 \cdot\sqrt{v_t} =\oo$; 

\prelimdeps $
\meant\left(\maxs\gamma_s +\phit\right)\cdot\uet^{-1} = o(T^{1-\beta})$ and, \limdept, $\{\hmtz:\ziz\}$ \limdep.

\item\label{item:v1_iii} 
$T^{-2}\sumt\gamma_t = o(1)$; 

$\meant \left\{\maxs ( L_s+ M_s)^2/\ue_s\right\}\cdot\sqrt{\phit}\cdot\uet^{-1} =\oo$;

\prelimdeps $
\meant \lmtp^2\cdot \left(\maxs\gamma_s + v_t\right) = o(T^{1-\beta})$ and, \limdept, $\{\etz:\ziz\}$ \limdep.

\item\label{item:v1_iv} 
\prelimdepsc $\meant\maxs\gamma_s = o(T^{1-\beta}) $ and, \limdept, $\{\etz,\hmtz:\ziz\}$ \limdep.
\ende 

\end{lemma}

\begin{proof}[\bf Proof of \lem~\ref{lem:v1}]
Recall that $\gt = \maxz\var\left\{\dfrac{\atzs}{\etz}\right\}$ from \eqref{eq:gt}. This ensures 
\begina
\var(\voz) \oeq{\eqref{eq:v1_def}} T^{-2}\var\left\{\sumt\dfrac{\atzs}{\etz} \right\} 
\overset{\eqref{eq:var_lem}}{\leq} \meant\var\left\{\dfrac{\atzs}{\etz}\right\} 
\leq\meant\gt,
\enda
which verifies the sufficiency of \lem~\ref{lem:v1}\eqref{item:v1_i}. 

We verify below the sufficiency of \lem~\ref{lem:v1}\eqref{item:v1_ii}--\eqref{item:v1_iv}, respectively. 
For their respective $T_0$,
a useful decomposition is 
\beginy\label{eq:decomp_v1}
T^2\var(\voz) &\oeq{\eqref{eq:v1_def}}& \ds \var\left\{\sumt\dfrac{\atzs}{\etz}\right\} \nnb\\ 
& =& \ds \sumt\var\left\{\dfrac{\atzs}{\etz}\right\}+ 2\sum_{t=2}^\T \sum_{k=1}^{t-1}\cov\left\{\dfrac{\bkzs }{\ekz },\dfrac{\atzs}{\etz} \right\} \nnb\\ 
& =& \ds \sumt\var\left\{\dfrac{\atzs}{\etz}\right\}  + 2\sum_{t=2}^{T_0}\sum_{k=1}^{t-1}\cov\left\{\dfrac{\bkzs }{\ekz },\dfrac{\atzs}{\etz} \right\} 
 \nnb\\ 
&&\ds  + 2\sum_{t=T_0 + 1}^\T \sum_{k=1}^{t-1}\cov\left\{\dfrac{\bkzs }{\ekz },\dfrac{\atzs}{\etz} \right\}  \nnb\\ 
&=& S_1 + 2 S_2 + 2 S_3,
\endy
where
\begina
S_1 = \sumt\var\left\{\dfrac{\atzs}{\etz}\right\}, \quad S_2 = \sum_{t=2}^{T_0}\sum_{k=1}^{t-1}\cov\left\{\dfrac{\bkzs }{\ekz },\dfrac{\atzs}{\etz} \right\}, \quad S_3 =\sum_{t=T_0 + 1}^\T \sum_{k=1}^{t-1}\cov\left\{\dfrac{\bkzs }{\ekz },\dfrac{\atzs}{\etz} \right\}. 
\enda
Note that $S_2$ is a fixed finite number and therefore satisfies $S_2 = o(T^2)$. Given \eqref{eq:decomp_v1}, a sufficient condition for $\var(\voz)=\oo$ is 
\beginy\label{eq:v1_goal}
S_1 = o(T^2), \quad S_3 = o(T^2). 
\endy
We verify below this sufficient condition \eqref{eq:v1_goal} under \lem~\ref{lem:v1}\eqref{item:v1_ii}--\eqref{item:v1_iv}, respectively. 

\paragraph*{\underline{Proof of \eqref{eq:v1_goal} under \lem~\ref{lem:v1}\eqref{item:v1_ii}}.}
First, the definition of $\gt$ in \eqref{eq:gt} implies 
\beginy\label{eq:v1_ii_1}
\var\left\{\dfrac{\atzs}{\etz}\right\} \overset{\eqref{eq:gt}}{\leq} \gamma_t\leq \left(\maxs\gamma_s +\phit\right)
\cdot \uet^{-1},
\endy
so that 
\begina
S_1 \oeq{\eqref{eq:decomp_v1}} \sumt\var\left\{\dfrac{\atzs}{\etz}\right\} \overset{\eqref{eq:v1_ii_1}}{\leq} \sumt  \left(\maxs\gamma_s +\phit\right)\cdot \uet^{-1} \oeqt{\lem~\ref{lem:v1}\eqref{item:v1_ii}} \ottb = o(T^2).
\enda
Next, let $\detzi =\etzi - \E\{\etzi\}$ to write
\beginy\label{eq:v1_2_2}
\cov\left\{\dfrac{\bkzs }{\ekz },\dfrac{\atzs}{\etz} \right\}
&=& \cov\left\{\dfrac{\bkzs }{\ekz }, \atzs \cdot \Big[\E\{\etzi\} + \detzi\Big] \right\} \nnb\\
&=&\cov\left\{\dfrac{\bkzs }{\ekz },\atzs\right\}\cdot\E\{\etzi\} + \cov\left\{\dfrac{\bkzs }{\ekz } ,\atzs\cdot\detzi\right\}, \quad \qquad
\endy
and 
\beginy\label{eq:sigma12_2}
S_3 \oeq{\eqref{eq:decomp_v1}} \sum_{t=T_0 + 1}^\T \sum_{k=1}^{t-1}\cov\left\{\dfrac{\bkzs }{\ekz },\dfrac{\atzs}{\etz} \right\} \overset{\eqref{eq:v1_2_2}}{=} \Sigma_1 + \Sigma_2,
\endy
where
\begina
\Sigma_1 = \sum_{t=T_0+1}^\T \sum_{k=1}^{t-1}\cov\left\{\dfrac{\bkzs }{\ekz },\atzs\right\} \cdot \E\{\etzi\},\quad
\Sigma_2 = \sum_{t=T_0+1}^\T \sum_{k=1}^{t-1}\cov\left\{\dfrac{\bkzs }{\ekz } ,\atzs\cdot\detzi\right\}.  \quad 
\enda
We show below 
\beginy\label{eq:sig12_2_goal}
\Sigma_1 = o(T^2), \quad \Sigma_2 = o(T^2),
\endy
respectively, which together ensure $S_3 = o(T^2)$ from \eqref{eq:sigma12_2}. 

\paragraph*{Proof of $\Sigma_1 = o(T^2)$ in \eqref{eq:sig12_2_goal}.}
Assume \wlg, $\{\hmtz:\ziz\}$ depend only on the previous $\min\{t-1, c_0 t^\beta\} = c_0 t^\beta$ time points. 
For $t> T_0$, we have
\beginy\label{eq:v1_ii_ss}
\begin{array}{lllll}
\cov\left\{\dfrac{\bkzs }{\ekz },\atzs\right\} &=& 0& \text{for} \ k < t-c_0t^\beta,\bigskip\\
\cov\left\{\dfrac{\bkzs }{\ekz },\atzs\right\} 
&\leq& 2^{-1}\left[\var\left\{\dfrac{\bkzs }{\ekz }\right\} +\var\left\{\atzs\right\}\right]\\
&\overset{\eqref{eq:gt}}{\leq}& 2^{-1}\left(\maxs\gamma_s +\phit\right)&\text{for} \ k \in [t -c_0 t^\beta, t).
\end{array}
\endy
This ensures
\begina
\Sigma_1&\overset{\eqref{eq:sigma12_2}}{=}&\sum_{t=T_0+1}^\T \sum_{k=1}^{t-1}\cov\left\{\dfrac{\bkzs }{\ekz },\atzs\right\} \cdot \E\{\etzi\}  \\ 
&\overset{\eqref{eq:v1_ii_ss}}{=}&\sum_{t=T_0+1}^\T \sum_{k=t-c_0t^\beta}^{t-1}\cov\left\{\dfrac{\bkzs }{\ekz },\atzs\right\} \cdot \E\{\etzi\} \\
&\overset{\eqref{eq:v1_ii_ss}}{\leq}&\sum_{t=T_0+1}^\T  c_0 t^\beta\cdot  2^{-1}\left(\maxs\gamma_s +\phit\right)\cdot\uet^{-1}\\
&\leq& 2^{-1}c_0 T^\beta\cdot\sumt  \left(\maxs\gamma_s +\phit\right)\uet^{-1} \\
&\oeqt{\lem~\ref{lem:v1}\eqref{item:v1_ii}}& c_0 T^\beta \cdot \ottb = o(T^2).
\enda

\paragraph*{Proof of $\Sigma_2 = o(T^2)$ in \eqref{eq:sig12_2_goal}.} 
For $k < t$, we have
\beginy\label{eq:v1_2_3}
&&\left|\cov\left\{\dfrac{\bkzs }{\ekz } ,\atzs\cdot\detzi\right\}\right|\nnb\\
&=&\left|\E\left\{\dfrac{\bkzs }{\ekz } \cdot \atzs\cdot\detzi\right\}  -  \E\left\{\dfrac{\bkzs }{\ekz }\right\} \cdot\E\left\{\atzs\cdot\detzi\right\}\right| \nnb\\
&\leq&\E\left\{\dfrac{\bkzs }{\ekz }\cdot\atzs\cdot\left|\detzi\right|\right\} +\E\left\{\dfrac{\bkzs }{\ekz }\right\}\cdot\E\Big\{\atzs\cdot\left|\detzi\right|\Big\}\nnb\\
&\leq&\infn{\dfrac{\bkzs }{\ekz }}\cdot\infn{\atzs} \cdot\E\left\{ \left|\detzi\right|\right\}+\infn{\dfrac{\bkzs }{\ekz }}\cdot\infn{\atzs}\cdot\E\left\{\left|\detzi\right| \right\}\nnb\\
&{\leq}& 2\cdot\left\{\maxs ( L_s+ M_s)^2/\ue_s\right\} \cdot \lmtp^2 \cdot \sqrt{v_t},
\endy
where the last inequality follows from 
$
\E\left\{ |\detzi |\right\}\leq\sqrt{\E\left\{ |\detzi |^2\right\}} =\sqrt{\var\{\etzi\}}\leq\sqrt{v_t}.$
This ensures
\begina
 T^{-2}\left|\Sigma_2\right| 
 &\overset{\eqref{eq:sigma12_2}}{=}& T^{-2}\left|\sum_{t=T_0+1}^\T \sum_{k=1}^{t-1}\cov\left\{\dfrac{\bkzs }{\ekz } ,\atzs\cdot\detzi\right\}\right|\\ 
 &\leq& 
 T^{-2}\sumt \sum_{k=1}^{t-1} \left|\cov\left\{\dfrac{\bkzs }{\ekz } ,\atzs\cdot\detzi\right\}\right|\\
&\overset{\eqref{eq:v1_2_3}}{\leq}& T^{-2}\sumt  T \cdot 2\cdot\left\{\maxs ( L_s+ M_s)^2/\ue_s\right\} \cdot \lmtp^2 \cdot\sqrt{v_t}\\
&\oeqt{\lem~\ref{lem:v1}\eqref{item:v1_ii}}& o(1).
\enda

\paragraph*{\underline{Proof of \eqref{eq:v1_goal} under \lem~\ref{lem:v1}\eqref{item:v1_iii}}.}
First, the definition of $\gt$ in \eqref{eq:gt} ensures
\begina
S_1 \oeq{\eqref{eq:decomp_v1}} \sumt\var\left\{\dfrac{\atzs}{\etz}\right\}\overset{\eqref{eq:gt}}{\leq}\sumt\gt = o(T^2).
\enda
Next, renew $\detzi$ as $\detzi =\atzs -\E\{\atzs\}$ to write
\beginy\label{eq:v1_3_22}
\cov\left\{\dfrac{\bkzs }{\ekz },\dfrac{\atzs}{\etz} \right\}
&=&\cov\left\{\dfrac{\bkzs }{\ekz } ,\Big[\E\{\atzs\} + \detzi \Big]\cdot\etzinv\right\} \nnb\\
&=&\E\{\atzs\}\cdot\cov\left\{\dfrac{\bkzs }{\ekz } , \etzinv\right\} +\cov\left\{\dfrac{\bkzs }{\ekz } ,\detzi\cdot\etzinv\right\},\quad \qquad
\endy
and 
\beginy\label{eq:sigma12_3} 
S_3 \oeq{\eqref{eq:decomp_v1}} \sum_{t=T_0 + 1}^\T \sum_{k=1}^{t-1}\cov\left\{\dfrac{\bkzs }{\ekz },\dfrac{\atzs}{\etz} \right\} \overset{\eqref{eq:v1_3_22}}{=} \Sigma_1 + \Sigma_2,
\endy
where
\begina
\Sigma_1 =\sum_{t=T_0+1}^\T \sum_{k=1}^{t-1}\E\{\atzs\}\cdot\cov\left\{\dfrac{\bkzs }{\ekz } , \etzinv\right\}, \quad \Sigma_2 =\sum_{t=T_0+1}^\T \sum_{k=1}^{t-1}\cov\left\{\dfrac{\bkzs }{\ekz } ,\detzi\cdot\etzinv\right\}.
\enda
We show below  
\beginy\label{eq:sig12_3_goal}
\Sigma_1 = o(T^2), \quad \Sigma_2 = o(T^2),
\endy
respectively, which together ensure $S_3 = o(T^2)$ from \eqref{eq:sigma12_3}.

\paragraph*{Proof of $\Sigma_1 = o(T^2)$ in \eqref{eq:sig12_3_goal}.} 
Assume \wlg, $\{\etz:\ziz\}$ depend only on the previous $\min\{t-1, c_0 t^\beta\} = c_0 t^\beta$ time points. 
For $t>T_0$, we have
\beginy\label{eq:v1_3_ss}
\begin{array}{llll}
\cov\left\{\dfrac{\bkzs }{\ekz },\etzinv\right\} 
&=& 0&\text{for} \ k < t-c_0t^\beta,\bigskip\\
\cov\left\{\dfrac{\bkzs }{\ekz },\etzinv\right\} 
&\leq& 2^{-1}\left[\var\left\{\dfrac{\bkzs }{\ekz }\right\} +\var\{\etzi\}\right]\\
&\leq& 2^{-1}\left(\maxs\gamma_s + v_t\right)&\text{for} \ k \in [	t -c_0 t^\beta, t).
\end{array}	
\endy
This ensures
\begina
\Sigma_1&\overset{\eqref{eq:sigma12_3}}{=}&\sum_{t=T_0+1}^\T \sum_{k=1}^{t-1}\E\{\atzs\}\cdot\cov\left\{\dfrac{\bkzs }{\ekz } ,\etzinv\right\}\\ 
&\overset{\eqref{eq:v1_3_ss}}{=}&\sum_{t=T_0+1}^\T \sum_{k=t-c_0t^\beta}^{t-1}\E\{\atzs\}\cdot \cov\left\{\dfrac{\bkzs }{\ekz } ,\etzinv\right\}\\
&\overset{\eqref{eq:v1_3_ss}}{\le}&\sum_{t=T_0+1}^\T  c_0 t^\beta\cdot \lmtp^2\cdot 2^{-1}\left(\maxs\gamma_s + v_t\right)	\\
&\leq& 2^{-1}c_0 T^\beta\cdot\sumt \lmtp^2\cdot \left(\maxs\gamma_s + v_t\right)	\\
&\oeqt{\lem~\ref{lem:v1}\eqref{item:v1_iii}}& c_0 T^\beta \cdot \ottb = o(T^2).
\enda

\paragraph*{Proof of $\Sigma_2 = o(T^2)$ in \eqref{eq:sig12_3_goal}.} 
For $k < t$, we have 
\beginy\label{eq:v1_3_2}
&&\left|\cov\left\{\dfrac{\bkzs }{\ekz } ,\detzi\cdot\etzinv\right\}\right| \nnb\\
&=&\left|\E\left\{\dfrac{\bkzs }{\ekz } \cdot \detzi\cdot\etzinv\right\}
- \E\left\{\dfrac{\bkzs }{\ekz }  \right\} 
\cdot 
\E\left\{ \detzi\cdot\etzinv\right\} \right| \nnb\\
%
&\leq&\E\left\{\dfrac{\bkzs }{\ekz }\cdot\left|\detzi\right|\cdot\etzinv\right\} +\E\left\{\dfrac{\bkzs }{\ekz }\right\}\cdot\E\left\{\left|\detzi\right|\cdot\etzinv\right\}\nnb\\
&\leq& \infn{\dfrac{\bkzs }{\ekz } }
\cdot\E\left\{\left|\detzi\right|\right\}\cdot\infn{\etzinv} + \infn{\dfrac{\bkzs }{\ekz } }
\cdot\E\left\{\left|\detzi\right|\right\}\cdot\infn{\etzinv}\nnb\\
&\leq& 2\cdot\left\{\maxs ( L_s+ M_s)^2/\ue_s\right\}\cdot\sqrt{\phit}\cdot\uet^{-1},
\endy
where the last inequality follows from 
\cll{
\E\left\{ |\detzi |\right\}\leq\sqrt{\E\left\{ |\detzi |^2\right\}}=\sqrt{\var\{\atzs\}}\overset{\eqref{eq:gt}}{\leq}\sqrt{\phit}.}
This ensures
\begina
T^{-2}\left|\Sigma_2\right| &\overset{\eqref{eq:sigma12_3}}{=}& T^{-2}\left|\sum_{t=T_0+1}^\T \sum_{k=1}^{t-1}\cov\left\{\dfrac{\bkzs }{\ekz } ,\detzi\cdot\etzinv\right\}\right|\\ 
&\leq& 
T^{-2} \sumt \sum_{k=1}^{t-1}\left|\cov\left\{\dfrac{\bkzs }{\ekz } ,\detzi\cdot\etzinv\right\}\right| \\
&\overset{\eqref{eq:v1_3_2}}{\leq}&
T^{-2} \sumt  T \cdot 2\cdot\left\{\maxs ( L_s+ M_s)^2/\ue_s\right\}\cdot\sqrt{\phit}\cdot\uet^{-1}\\
&\oeqt{\lem~\ref{lem:v1}\eqref{item:v1_iii}}& o(1). 
\enda

\paragraph*{\underline{Proof of \eqref{eq:v1_goal} under \lem~\ref{lem:v1}\eqref{item:v1_iv}}.}
The definition of $\gt$ in \eqref{eq:gt} ensures
\begina
S_1 \oeq{\eqref{eq:decomp_v1}} \sumt\var\left\{\dfrac{\atzs}{\etz}\right\} \overset{\eqref{eq:gt}}{\leq}\sumt\gt \leq\sumt\maxs\gamma_s \oeqt{\lem~\ref{lem:v1}\eqref{item:v1_iv}} \ottb = o(T^2). 
\enda
Next, assume \wlg, $\{\etz,\hmtz:\ziz\}$ depend only on the previous $\min\{t-1, c_0 t^\beta\} = c_0 t^\beta$ time points. 
For $t> T_0$, we have
\beginy\label{eq:v1_4_ss}
\begin{array}{lllll}
\cov\left\{\dfrac{\bkzs }{\ekz },\dfrac{\atzs }{\etz }\right\} &=& 0&\text{for} \ k < t-c_0t^\beta,\bigskip\\
\cov\left\{\dfrac{\bkzs }{\ekz },\dfrac{\atzs }{\etz }\right\} 
&\leq& 2^{-1}\left[\var\left\{\dfrac{\bkzs }{\ekz }\right\} +\var\left\{\dfrac{\atzs }{\etz }\right\}\right]\\
&\leq& 2^{-1}(\gamma_k +\gamma_t)\leq\maxs\gamma_s
&\text{for} \ k \in [	t -c_0 t^\beta, t).
\end{array}
\endy
This ensures 
\begina
S_3 &\oeq{\eqref{eq:decomp_v1}}&\sum_{t=T_0+1}^\T \sum_{k=1}^{t-1}\cov\left\{\dfrac{\bkzs }{\ekz },\dfrac{\atzs }{\etz }\right\}\\
 &\overset{\eqref{eq:v1_4_ss}}{=}&\sum_{t=T_0+1}^\T \sum_{k=t-c_0t^\beta}^{t-1}\cov\left\{\dfrac{\bkzs }{\ekz },\dfrac{\atzs }{\etz }\right\}\\
  &\overset{\eqref{eq:v1_4_ss}}{\leq}&\sum_{t=T_0+1}^\T  c_0 t^\beta\cdot \maxs\gamma_s \\ 
&\leq & c_0 T^\beta\cdot \sumt\maxs\gamma_s\\
&\oeqt{\lem~\ref{lem:v1}\eqref{item:v1_iv}} & c_0 T^\beta\cdot \ottb = o(T^2).
\enda	

\end{proof}

\begin{lemma}\label{lem:V2}
\prear. Let 
\beginy\label{eq:v2_def}
\vtzzp =\meant\Big[\atz\atzp -\E\big\{\atz\atzp\big\}\Big]\quad \text{for} \ z, z'\in \mz. 
\endy
As $T\to\infty$, we have 
\cls{\var(\vtzzp) = o(1) \quad \text{for all $z,z'\in\mz$}}
if either of the following conditions holds: 
\begine[(i)]
\item\label{item:V2_vm=0} \textbf{Vanishing $\var\{\hmtz\}$:} $\meant\phit = o(1)$.

\item\label{item:V2_m_limited dep} \textbf{{\llrd} of $\hmtz$:} 
\prelimdepsc $\meant\maxs\phi_s = o( T^{1-\beta})$ and, \limdept, $\{\hmtz:\ziz\}$ \limdep. 
\ende 

\end{lemma}

\begin{proof}[\bf Proof of \lem~\ref{lem:V2}]
The sufficiency of \lem~\ref{lem:V2}\eqref{item:V2_vm=0} follows from 
\begina
\var( V_{2, zz'}) \oeq{\eqref{eq:v2_def}} \ds \dfrac{1}{\tsq}\var\left\{\sumt\atz\atzp\right\} 
 \oleq{\eqref{eq:var_lem}} \meantf\var\left\{\atz\atzp\right\} 
  \oleq{\eqref{eq:gt}} \meantf\phit.  
\enda
We verify below the sufficiency of \lem~\ref{lem:V2}\eqref{item:V2_m_limited dep}.

First, write 
\beginy\label{eq:decomp_v2}
  T^2\var( V_{2, zz'}) &\oeq{\eqref{eq:v2_def}}&\var\left\{\sumt\atz\atzp\right\}\nnb\\
  &=&\sumt\var\{\atz\atzp\}  + 2\sum_{t=2}^\T \sum_{k=1}^{t-1}\cov\{\bkz\bkzp,\atz\atzp \}\nnb\\
 &=&\sumt\var\{\atz\atzp\}  + 2\sum_{t=2}^{T_0}\sum_{k=1}^{t-1}\cov\{\bkz\bkzp,\atz\atzp \}\nnb\\
&&+ 2\sum_{t=T_0+1}^\T \sum_{k=1}^{t-1}\cov\{\bkz\bkzp,\atz\atzp \} \nnb\\
 &=& S_1 + 2 S_2 + 2 S_3, 
\endy
where 
$
S_1 = \sumt\var\{\atz\atzp\}$, $S_2 = \sum_{t=2}^{T_0}\sum_{k=1}^{t-1}\cov\{\bkz\bkzp,\atz\atzp \}$, and 
$
S_3 = \sum_{t=T_0+1}^\T \sum_{k=1}^{t-1}\cov\{\bkz\bkzp,\atz\atzp \}.$
The definition of $\phit$ in \eqref{eq:gt} implies 
\cll{
\var\{\atz\atzp\} \oeq{\eqref{eq:gt}} \phit \leq \maxs\phi_s}

so that 
\begina
S_1 \oeq{\eqref{eq:decomp_v2}} \sumt\var\{\atz\atzp\}\leq  \sumt\maxs\phi_s
\oeqt{\lem~\ref{lem:V2}\eqref{item:V2_m_limited dep}} o( T^{2-\beta}) = o( T^2).
\enda
under \lem~\ref{lem:V2}\eqref{item:V2_m_limited dep}.
In addition, $S_2$ is a fixed finite number and therefore satisfies $S_2 = o(T^2)$. Given \eqref{eq:decomp_v2}, it suffices to verify that 
\beginy\label{eq:v2_s3_goal}  
S_3 = o(T^2)
\endy
under \lem~\ref{lem:V2}\eqref{item:V2_m_limited dep}. 

\paragraph*{Proof of \eqref{eq:v2_s3_goal} under \lem~\ref{lem:V2}\eqref{item:V2_m_limited dep}.} Assume \wlg, $\{\hmtz:\ziz\}$ depend only on the previous $\min\{t-1, c_0 t^\beta\} = c_0 t^\beta$ time points. 
For $t > T_0$, we have 
\beginy\label{eq:V2_ss}
\begin{array}{llll}
\cov\{\bkz\bkzp,\atz\atzp \} &=& 0 &\text{for} \ k < t - c_0 t^\beta,\\
\cov\{\bkz\bkzp,\atz\atzp \} &\leq & 2^{-1}\left[\var\{\bkz\bkzp\} +\var\{\atz\atzp \}\right]\\
 &\leq & 2^{-1}(\phi_k +\phit) 
\leq  
\maxs\phi_s &\text{for} \ k \in [t - c_0 t^\beta, t).
\end{array}
\endy
This ensures \eqref{eq:v2_s3_goal} as follows:
\begina
S_3 &\oeq{\eqref{eq:decomp_v2}}&\sum_{t=T_0+1}^\T \sum_{k=1}^{t-1}\cov\{\bkz\bkzp,\atz\atzp \}\\
&\overset{\eqref{eq:V2_ss}}{=}&   \sum_{t=T_0 +1}^\T \sum_{k=t-c_0 t^\beta}^{t-1} \cov\{\bkz\bkzp,\atz\atzp \} 
\\		 
&\overset{\eqref{eq:V2_ss}}{\leq}&  
\sum_{t=T_0+1}^T c_0t^\beta\cdot\maxs\phi_s\\ 
&\leq& 
 c_0  T^\beta\cdot\sumt\maxs\phi_s\\
 & \oeqt{\lem~\ref{lem:V2}\eqref{item:V2_m_limited dep}}& c_0  T^\beta\cdot \ottb = o(T^2).		
\enda
\end{proof}

\begin{lemma}\label{lem:aipw_sc_app}
\prear. 
If $\lm(\va)$ is uniformly bounded away from 0 as $\ttinf$, then 
a sufficient condition for \cond~\ref{cond:aipw_clt}\eqref{item:clt_aipw_var_conv} is
\cll{\fn{\cva - \va} = \op,}

 which holds if any of the following conditions is satisfied:
\begine[(i)]
\item\label{item:V_i}
 $\meant\gamma_t = o(1)$; 
 
 $\meant\phit = o(1)$. 
\item\label{item:V_ii}
$\meant\left\{\maxs ( L_s+ M_s)^2/\ue_s\right\} \cdot \lmtp^2 \cdot\sqrt{v_t}=\oo$; 

\prelimdeps $
\meant \left(\maxs\gamma_s +\phit\right)\cdot\uet^{-1} = o(T^{1-\beta})$, $\meant\maxs\phi_s = o( T^{1-\beta})$, and, \limdept, $\{\hmtz:\ziz\}$ \limdep.

\item\label{item:V_iii} 
$T^{-2}\sumt\gamma_t = o(1)$; 

$\meant\left\{\maxs ( L_s+ M_s)^2/\ue_s\right\} \cdot\sqrt{\phit}\cdot\uet^{-1} = o(1)$; 

$\meant\phit = o(1)$; 

\prelimdeps $
\meant \lmtp^2\cdot \left(\maxs\gamma_s + v_t\right) = o(T^{1-\beta})$ and, \limdept, $\{\etz:\ziz\}$ \limdep.

\item\label{item:V_iv} 
\prelimdepsc $\meant\maxs\gamma_s = o(T^{1-\beta}) $, $\meant\maxs\phi_s = o( T^{1-\beta})$, and, \limdept, $\{\etz,\hmtz:\ziz\}$ \limdep.
\ende

	\end{lemma} 	

\begin{proof}[\bf Proof of \lem~\ref{lem:aipw_sc_app}]
Recall that 
\begina
\va &=&\diag\left[\meant\E\left\{\dfrac{\atzs}{\etz }\right\}\right]_{\ziz} - \meant\E(\mat\matt),\\
\cva &=&\diag\left[\meant\dfrac{\atzs}{\etz }\right]_{\ziz} - \meant\mat\matt. 
\enda
We have 
\beginy\label{eq:S1_v1-v2}
\cva-\va = V_1 - V_2,
\endy 
where
\begina
 V_1 &=&\diag\left(\meant\left[\dfrac{\atzs}{\etz } -\E\left\{\dfrac{\atzs}{\etz }\right\}\right]\right)_{\ziz} =\diag(\voz)_{z\in\mz},\\
 V_2 &=&\meant\Big\{\mat\matt -\E(\mat\matt)\Big\} = (\vtzzp)_{z,z'\in\mz} 
\enda
with 
\begina
\voz= \meant\left[\dfrac{\atzs}{\etz } -\E\left\{\dfrac{\atzs}{\etz }\right\}\right], \quad \vtzzp =\meant\Big[\atz\atzp -\E\big\{\atz\atzp\big\}\Big]
\enda as defined in \eqref{eq:v1_def} and \eqref{eq:v2_def}, respectively. 
From \eqref{eq:S1_v1-v2}, we have
\begina
\fn{\cva-\va} =\fn{V_1 - V_2}\leq\fn{V_1} +\fn{V_2} =  \left(\sumz \voz^2\right)^{1/2} + \left(\sum_{z,z'\in\mz} \vtzzp^2\right)^{1/2},
\enda
so that $\fn{\cva-\va} =\op$ if 
\beginy\label{eq:V1V2_suff}
\voz=\op, \quad \vtzzp =\op \quad \text{for all $z,z'\in\mz$.}
\endy 
Given $\E(\voz) =\E(\vtzzp) = 0$ by definition, \lemchebf\ ensures that a sufficient condtion for \eqref{eq:V1V2_suff} is
\beginy\label{eq:ss_v1v2}
\var(\voz) = o(1),\quad\var(\vtzzp) = o(1)\quad\text{for all $z,z'\in\mz$}. 
\endy

\lem~\ref{lem:aipw_sc_app}\eqref{item:V_i} ensures \eqref{eq:ss_v1v2} by 
\lem~\ref{lem:v1}\eqref{item:v1_i} and \lem~\ref{lem:V2}\eqref{item:V2_vm=0}. 

\lem~\ref{lem:aipw_sc_app}\eqref{item:V_ii} ensures \eqref{eq:ss_v1v2} by 
\lem~\ref{lem:v1}\eqref{item:v1_ii} and \lem~\ref{lem:V2}\eqref{item:V2_m_limited dep}. 

\lem~\ref{lem:aipw_sc_app}\eqref{item:V_iii} ensures \eqref{eq:ss_v1v2} by 
\lem~\ref{lem:v1}\eqref{item:v1_iii} and \lem~\ref{lem:V2}\eqref{item:V2_vm=0}.

\lem~\ref{lem:aipw_sc_app}\eqref{item:V_iv} ensures \eqref{eq:ss_v1v2} by 
\lem~\ref{lem:v1}\eqref{item:v1_iv} and \lem~\ref{lem:V2}\eqref{item:V2_m_limited dep}. 

\end{proof}

\begin{lemma}\label{lem:var_bounds}
For infinite sequences $(a_t)_{t=1}^\infty$ and $(b_t)_{t=1}^\infty$, write $a_t\ls b_t$ if there exists constant $c_0 <\infty$ such that $a_t\leq c_0 b_t$ for all $t$. 
For all $t$ and $z,z'\in\mz$, we have 
\begina
\phit \ls \omt (\lt +\mt)^2 ,\quad 
\gamma_t \ls \omt (\lt +\mt)^2\cdot\uet^{-2} + v_t  \lmtp^4. 
\enda
\end{lemma}

\begin{proof}[\bf Proof of \lem~\ref{lem:var_bounds}]
Recall that $|\atz| =|\ytz-\hmtz|\leq\lt +\mt$ and $\var\{\atz\} = \var\{\hmtz\} \leq \omt$ by definition.
The results follow from \eqref{eq:var_lem} as follows:
\beginy\label{eq:varBB} 
\var\{\atz\atzp\} 
 &\oleq{\eqref{eq:var_lem}}&   2\cdot \var\{\atz\}\cdot\infns{\atzp} + 2\cdot\var\{\atzp\}\cdot\infns{\atz}\nnb\\
 &\leq&  4\cdot\omt (\lt +\mt)^2,\\
\var\left\{\dfrac{\atzs}{\etz}\right\}
&\oleq{\eqref{eq:var_lem}}& 2\cdot\var\{\atzs\}\cdot\infns{\etzinv} + 2\cdot\var\{\etzi\}\cdot\infns{\atzs}\nnb\\
&\oleq{\eqref{eq:varBB}}& 8\cdot\omt (\lt +\mt)^2\cdot\uet^{-2} + 2\cdot v_t  \lmtp^4.  \nnb 
\endy
\end{proof}

\subsection{Proof of \prop~\ref{prop:aipw_sc_app}}
\begin{proof}[\bf Proof of \prop~\ref{prop:aipw_sc_app}]
Let 
\beginy\label{eq:b1b2}
b_1 = \maxt\omt \cdot\uet^{-2}, \quad b_2 = \maxt v_t,
\endy
where we make their dependence on $T$ implicit for notational simplicity.
When $\ytz$ and $\hmtz$ are uniformly bounded, it follows from \lem~\ref{lem:var_bounds} that 
\beginy\label{eq:max_gamma_2}
\beginar{lll}
 \phit  \ls \omt, &&\maxt\phit \ls b_1; \\
 \gamma_t \ls \omt \cdot\uet^{-2} + v_t, && \maxt\gamma_t\ls b_1 + b_2.
 \endar 
\endy
This ensures the results as follows.  
\begine[(i)]
\item \prop~\ref{prop:aipw_sc_app}\eqref{item:V_i_main} ensures \lem~\ref{lem:aipw_sc_app}\eqref{item:V_i} as follows: 
\begina
\meant\gamma_t  \overset{\eqref{eq:max_gamma_2}}{\ls}   \meant \omt \cdot \uet^{-2} + \meant  v_t \oeqt{\prop~\ref{prop:aipw_sc_app}\eqref{item:V_i_main}} \oo,\\
\meant\phit  \overset{\eqref{eq:max_gamma_2}}{\ls}  \meant \omt \leq \meant \omt \cdot \uet^{-2} \oeqt{\prop~\ref{prop:aipw_sc_app}\eqref{item:V_i_main}} \oo. 
\enda 

\item \prop~\ref{prop:aipw_sc_app}\eqref{item:V_ii_main} ensures 
\beginy\label{eq:s1_2}
\beginar{l}
\ds \meant\left(\maxs \ue_s^{-1}\right) \cdot\sqrt{v_t} = o(1),\medskip\\
b_1 \cdot \left(\meant \ueti\right) \oeq{\eqref{eq:b1b2}} o(T^{1-\beta}), \quad b_2 \cdot \left(\meant \ueti\right) \oeq{\eqref{eq:b1b2}} o(T^{1-\beta}), 
\endar
\endy 
so that \lem~\ref{lem:aipw_sc_app}\eqref{item:V_ii} holds as follows:
\begina
\beginar{l}
\beginar{l}
\bullet \quad \ds \meant\left\{\maxs ( L_s+ M_s)^2/\ue_s\right\} \cdot \lmtp^2 \cdot\sqrt{v_t}  \ls \ds \meant\left(\maxs \ue_s^{-1}\right) \cdot\sqrt{v_t}  \oeq{\eqref{eq:s1_2}}  o(1),
\endar\bigskip\\
 \beginar{lcl}
\bullet \quad \ds \meant \left(\maxs\gamma_s +\phit\right)\cdot\uet^{-1}
&\leq&\ds \left(\maxt \gamma_t + \maxt \phi_t \right) \cdot \left(\meant \ueti\right)\\
&\overset{\eqref{eq:max_gamma_2}}{\ls}& \ds
(b_1 + b_2 + b_1) \cdot \left(\meant \ueti\right)\\
&\oeq{\eqref{eq:s1_2}}& o(T^{1-\beta}),
\endar \bigskip\\
\beginar{l}
\bullet \quad \ds\meant\maxs\phi_s  \leq  \maxt\phit  \overset{\eqref{eq:max_gamma_2}}{\ls}  b_1 \leq b_1 \cdot \left(\meant \ueti\right) \oeq{\eqref{eq:s1_2}}  o(T^{1-\beta}).
\endar
\endar
\enda

\item \prop~\ref{prop:aipw_sc_app}\eqref{item:V_iii_main} ensures 
\beginy\label{eq:s1_3}
\meant \omt = o(1), \quad \meant\left(\displaystyle\maxs \ue_s^{-1}\right) \cdot\sqrt{\omt}\cdot\uet^{-1} 
= o(1),\quad
b_1   = o(T^{1-\beta}), \quad b_2  = o(T^{1-\beta}), \quad
\endy 
so that \lem~\ref{lem:aipw_sc_app}\eqref{item:V_iii} holds as follows:
\begina
\begin{array}{l}
\begin{array}{l}
\bullet \quad T^{-2} \displaystyle\sumt\gamma_t 
 \overset{\eqref{eq:max_gamma_2}}{\ls}  T^{-2}\displaystyle\sumt \omt \cdot \uet^{-2} +  T^{-2}\displaystyle\sumt  v_t 
  \oleq{\eqref{eq:b1b2}}  T^{-1} b_1 +  T^{-1} b_2  \oeq{\eqref{eq:s1_3}}  o(T^{-\beta})  =  o(1),
 \end{array}
 \bigskip\\
\begin{array}{llllllll}
\bullet \quad\ds\meant\left\{\displaystyle\maxs ( L_s+ M_s)^2/\ue_s\right\} \cdot\sqrt{\phit}\cdot\uet^{-1}
\overset{\eqref{eq:max_gamma_2}}{\ls} \displaystyle\meant\left(\displaystyle\maxs \ue_s^{-1}\right) \cdot\sqrt{\omt}\cdot\uet^{-1} 
\oeq{\eqref{eq:s1_3}} o(1),
\end{array}
 \bigskip\\
\begin{array}{l}
\bullet \quad\displaystyle\meant\phit \overset{\eqref{eq:max_gamma_2}}{\ls} \displaystyle\meant \omt \oeq{\eqref{eq:s1_3}}  o(1) ,
\end{array}\bigskip\\
\beginar{rcl}
 \bullet \quad\ds\meant  \lmtp^2 \cdot \left(\maxs\gamma_s + v_t\right)  &\ls&   \ds\meant \left(\maxs\gamma_s + v_t\right) \\
&\leq& \ds  \maxt\gamma_t  + \meant v_t 
\overset{\eqref{eq:max_gamma_2}}{\ls} b_1 + b_2 + b_2 \oeq{\eqref{eq:s1_3}} o(T^{1-\beta}). 
\endar
\end{array}
\enda

\item \prop~\ref{prop:aipw_sc_app}\eqref{item:V_iv_main} ensures 
\begina
b_1 = o(T^{1-\beta}), \quad b_2 = o(T^{1-\beta})
\enda
so that \lem~\ref{lem:aipw_sc_app}\eqref{item:V_iv} holds as follows: 
\begina
\meant\maxs\gamma_s \leq \maxt \gt \overset{\eqref{eq:max_gamma_2}}{\ls} b_1 + b_2 = o(T^{1-\beta}),\\
\meant\maxs\phi_s \leq \maxt \phit \overset{\eqref{eq:max_gamma_2}}{\ls} b_1= o(T^{1-\beta}).
\enda
\ende
\end{proof}  

\end{document}